%% file: main.tex
\documentclass[10pt, a4paper]{article}

\input{basic.tex}

\title{BKL bounces outside homogeneity: \\ Gowdy symmetric spacetimes}
\author{Warren Li}
\affil{\small Princeton University, Department of Mathematics, Fine Hall, Washington Road, Princeton, NJ 08544, USA}
\begin{document}

\maketitle

\begin{abstract}
    We study the phenomenon of bounces, as predicted by Belinski Khalatnikov and Lifshitz (BKL), as an instability mechanism within the setting of the Einstein vacuum equations in Gowdy symmetry. In particular, for a wide class of inhomogeneous initial data we prove that the dynamics near the $t = 0$ singularity are well--described by ODEs reminiscent of Kasner bounces.

    Unlike previous works regarding bounces, our spacetimes are not necessarily spatially homogeneous, and a crucial step is proving so-called asymptotically velocity term dominated (AVTD) behaviour, even in the presence of nonlinear BKL bounces and other phenomena such as spikes. (A similar phenomenon involving bounces and AVTD behaviour, though not spikes, can also be seen in our companion paper \cite{MeSurfaceSymPaper}, albeit in the context of the Einstein--Maxwell--scalar field model in surface symmetry.) 

    One particular application is the study of (past) instability of certain polarized Gowdy spacetimes, including some Kasner spacetimes. Perturbations of such spacetimes are such that the singularity persists, but the intermediate dynamics -- between initial data and the singularity -- feature BKL bounces. 
\end{abstract}

\setcounter{tocdepth}{2}
\tableofcontents
\input{introduction.tex}


\input{theorem.tex}

\input{energy.tex}

\input{energy_estimates.tex}

\input{interpolation.tex}

\input{ode.tex}

\input{stability.tex}

\input{bounce.tex}

\bibliography{bibliography_master.bib}
\bibliographystyle{abbrvnat_mod}

\end{document}

%% file: basic.tex
\usepackage[german,english]{babel} 
\usepackage{booktabs} 
\usepackage{comment} 
\usepackage[utf8]{inputenc} 
\usepackage[T1]{fontenc} 
\usepackage{caption, soul}
\usepackage{authblk}

\usepackage{bigints,bbm,slashed,mathtools,amssymb,amsmath,amsfonts,amsthm} 
\usepackage{physics, upgreek}
\usepackage[mathscr]{euscript}
\usepackage{cancel}
\usepackage{url}
\usepackage{enumerate}
\usepackage{tikz} 
\usepackage{tikz-cd}
\usepackage{pgfplots}
\pgfplotsset{compat=1.16}
\usetikzlibrary{babel}
\usetikzlibrary{positioning,arrows} 
\usetikzlibrary{decorations.pathreplacing, decorations.markings} 
\usetikzlibrary{patterns}
\usepackage{tikzsymbols} 
\usepackage{blindtext} 

\usepackage{geometry} 

\usepackage{setspace} 
\usepackage[T1]{fontenc} 
\usepackage[utf8]{inputenc} 

\usepackage{XCharter} 

\usepackage{fancyhdr} 
\usepackage[hang,flushmargin]{footmisc}

\usepackage[square, numbers]{natbib} 

\usepackage{har2nat} 
\usepackage[colorlinks=true, citecolor=blue]{hyperref}

\numberwithin{equation}{section}
\theoremstyle{definition}
\newtheorem{definition}{Definition}[section]
\theoremstyle{plain}
\newtheorem{theorem}{Theorem}[section]
\newtheorem{proposition}[theorem]{Proposition}
\newtheorem{lemma}[theorem]{Lemma}
\newtheorem{corollary}[theorem]{Corollary}
\newtheorem{conjecture}[theorem]{Conjecture}

\newtheorem{innercustomgeneric}{\customgenericname}
\providecommand{\customgenericname}{}
\newcommand{\newcustomtheorem}[2]{%
  \newenvironment{#1}[1]
  {%
   \renewcommand\customgenericname{#2}%
   \renewcommand\theinnercustomgeneric{##1}%
   \innercustomgeneric
  }
  {\endinnercustomgeneric}
}

\newcustomtheorem{customthm}{Theorem}
\newcustomtheorem{customlemma}{Lemma}
\newcustomtheorem{customcorollary}{Corollary}

\theoremstyle{remark}
\newtheorem*{remark}{Remark}

\newcommand{\R}{\mathbb{R}}

\newcommand{\N}{\mathbb{N}}

\DeclareFontFamily{U}{mathx}{}
\DeclareFontShape{U}{mathx}{m}{n}{<-> mathx10}{}
\DeclareSymbolFont{mathx}{U}{mathx}{m}{n}
\DeclareMathAccent{\widehat}{0}{mathx}{"70}
\DeclareMathAccent{\widecheck}{0}{mathx}{"71}

%% file: introduction.tex
\section{Introduction} \label{sec:intro}

\subsection{The Einstein equations in Gowdy symmetry} \label{sub:intro_einstein}

In this article, we study the structure and (in)stability properties of \emph{spacelike singularities} arising in solutions of the $1+3$-dimensional Einstein vacuum equations. The Einstein vacuum equations are given by the following geometric equation for the Ricci curvature of a Lorentzian metric $\mathbf{g}$ on a manifold $\mathcal{M}^{1+3}$:
\begin{equation} 
    \mathbf{Ric}_{\mu\nu}[\mathbf{g}] = 0, \label{eq:einstein}.
\end{equation}

More specifically, we study phenomena related to spacelike singularities in the context of $\mathbb{T}^3$ Gowdy symmetric solutions to the Einstein equations. A simple characterization of $\mathbb{T}^3$ Gowdy symmetry\footnote{
    A more geometric characterization is a spacetime with $\mathbb{T}^3$ spatial topology and containing two spacetlike Killing vector fields, here given by $\frac{\partial}{\partial \sigma}$ and $\frac{\partial}{\partial \delta}$, and whose associated \emph{twist constants} vanish. See the original paper of Gowdy \cite{Gowdy}, or \cite{ChruscielU1U1, RingstromGowdyReview}.
} is a spacetime $(\mathcal{M}, \mathbf{g})$ such that there exists a global coordinate system $(t, \theta, \sigma, \delta) \in (0, + \infty) \times (\mathbb{S}^1)^3$ such that the spacetime metric $\mathbf{g}$ may be written as:
\begin{equation} \label{eq:gowdy}
    \mathbf{g} = - t^{-\frac{1}{2}}e^{\frac{\lambda}{2}} ( - dt^2 + d \theta^2) + t [ e^P (d \sigma + Q d \delta)^2 + e^{-P}d \delta^2],
\end{equation}
where the functions $P$, $Q$ and $\lambda$ depend only on the variables $t \in (0, + \infty)$ and $\theta \in \mathbb{S}^1$. In the sequel, we view $\mathbb{S}^1$ as the closed interval $[- \pi, \pi]$ with endpoints identified.

For a Gowdy symmetric spacetime with metric $\mathbf{g}$ in the gauge of \eqref{eq:gowdy}, the Einstein vacuum equations \eqref{eq:einstein} may be written as the following system of PDEs for the variables $(P, Q)$:
\begin{gather}
    \label{eq:P_evol}
    (t \partial_t)^2 P - (t \partial_{\theta})^2 P = e^{2P} (t \partial_t Q)^2 - e^{2P} (t \partial_{\theta} Q)^2, \\[0.3em]
    \label{eq:Q_evol}
    (t \partial_t)^2 Q - (t \partial_{\theta})^2 Q = - 2 (t \partial_t P) (t \partial_t Q) + 2 (t \partial_{\theta} P) (t \partial_{\theta} Q)
\end{gather}
together with two equations for $\lambda$.
\begin{gather}
    \label{eq:L_evol}
    t \partial_t \lambda = (t \partial_t P)^2 + (t \partial_{\theta} P)^2 + e^{2P} (t \partial_t Q)^2 + e^{2P} (t \partial_{\theta} Q)^2, \\[0.3em]
    \label{eq:L_constraint}
    \partial_{\theta} \lambda = 2 ( t \partial_t P \, \partial_{\theta} P + e^{2P} t \partial_t Q \, \partial_{\theta} Q ).
\end{gather}

Note that the former equations \eqref{eq:P_evol} and \eqref{eq:Q_evol} decouple from the latter two equations for $\lambda$, and in studying the evolutionary problem one usually solves for $P$ and $Q$ using \eqref{eq:P_evol} and \eqref{eq:Q_evol}, then afterwards one integrates \eqref{eq:L_evol} to solve for $\lambda$. The final equation \eqref{eq:L_constraint} is a constraint that is propagated by the remaining equations, as long as it is satisfied at some initial time.

A further remarkable feature of $\mathbb{T}^3$ Gowdy symmetric spacetimes is that the equations \eqref{eq:P_evol} and \eqref{eq:Q_evol} have a hidden geometric structure. It turns out these two equations are exactly a system of \emph{wave maps} (see \cite[Chapter 6]{TaoNonlinearDispersive} for a general definition of wave maps) between a domain manifold $\R_{> 0} \times \mathbb{T}^2$ with metric
\[
    g_{0} = - dt^2 + d \theta^2 + t^2 d \chi^2, 
\]
and a target manifold $\R^2$ with metric
\[
    g_R = dP^2 + e^{2P} dQ^2.
\]
More precisely, the equations \eqref{eq:P_evol}--\eqref{eq:Q_evol} describe wave maps from $(\R_{>0} \times \mathbb{T}^2, g_0)$ to $(\R^2, g_R)$ which are independent of $\chi$. Such wave maps conserve the energy functional:
\begin{equation} \label{eq:wavemap_energy}
    \mathcal{E}(t) \coloneqq \frac{1}{2} \int_{\mathbb{S}^1} \left[ (\partial_t P)^2(t, \theta) + (\partial_{\theta} P)^2(t, \theta) + e^{2P} (\partial_t Q)^2 (t, \theta) + e^{2P} (\partial_{\theta} Q)^2(t, \theta)\right] \, d \theta.
\end{equation}

Note that $(\R^2, g_R)$ is isometric to hyperbolic space. Indeed, the change of variables $x = Q$, $y = e^{-P}$ puts $g_R$ into the familiar upper half plane model $g_R = y^{-2}(dx^2 + dy^2)$. See \cite{RingstromGowdySCC1, RingstromGowdySCC2} and references therein for further details regarding the wave map structure.

A consequence of conservation of energy \eqref{eq:wavemap_energy} is that solutions to the Gowdy symmetric system \eqref{eq:P_evol}--\eqref{eq:Q_evol} arising from regular initial data will persist and remain regular in the entire interval $t \in (0, + \infty)$. As a result, $\mathbb{T}^3$ Gowdy symmetric spacetimes are a popular model for studying the Strong Cosmic Censorship conjecture, with the strategy often to show geodesic completeness in the $t \to + \infty$ direction\footnote{Note that geodesic completeness is not necessary to show future inextendibility, see \cite{DafRendallLetter}.}  and singularity formation in the $t \to 0$ direction.

The first major breakthrough regarding Strong Cosmic Censorship came in studying the subclass of polarized Gowdy spacetimes, defined as follows:
\begin{definition}
    A \emph{polarized Gowdy spacetime} is a $\mathbb{T}^3$ Gowdy symmetric solution to the Einstein vacuum equation with $Q \equiv 0$ everywhere. In particular $P(t, \theta)$ solves the linear hyperbolic PDE:
    \begin{equation} \label{eq:P_evol_polarized}
        (t \partial_t)^2 P - (t \partial_{\theta})^2 P = 0.
    \end{equation}
\end{definition}

By studying \eqref{eq:P_evol_polarized}, in \cite{SCC_PolarizedGowdy} the authors prove Strong Cosmic Censorship in the polarized Gowdy class by showing that solutions are geodesically complete in the $t \to \infty$ solution and that \emph{generically} the spacetime metric $\mathbf{g}$ exhibits curvature blow-up in the $t \to 0$ direction. The essential step, particularly in the context of $t \to 0$, is understanding asymptotics for $P$, as we explain in Section~\ref{sub:intro_sing}.

In a series of influential works, Ringstr\"om \cite{RingstromGowdySCC1, RingstromGowdySCC2} then extended the proof of Strong Cosmic Censorship to the full class of Gowdy symmetric spacetimes. Ringstr\"om's strategy was again to prove geodesic completeness in the $t \to \infty$ direction, and that generically one can show asymptotics for $P$ as $t \to 0$ that imply curvature blow-up to the past. As we see later, the asymptotics are more complicated than in the polarized case, and Ringstr\"om importantly allows for a finite number of mild irregularities in the asymptotic profile for $P$ known as ``spikes''. Such spike-containing solutions were constructed earlier in \cite{RendallWeaverSpikes}.

\subsection{Asymptotics towards the singularity} \label{sub:intro_sing}

As mentioned above, solutions to the Gowdy symmetric system \eqref{eq:P_evol}--\eqref{eq:Q_evol} arising from regular initial data exist globally in time $t \in (0, + \infty)$. Armed with this global existence result, a reasonable goal is to understand in furher detail the behaviour of solutions as one approaches the \emph{spacelike singularity} at $t = 0$. 

For instance, one would like to understand the asymptotic behaviour of the functions $P(t, \theta)$ and $Q(t, \theta)$ in the limit $t \to 0$. By considering a (Fuchsian) series expansion at $t = 0$, i.e.~an asymptotic expansion whose terms are powers of $t$ with coefficients depending on $\theta$, one could conjecture that, e.g.~from \cite{IsenbergMoncriefGowdy}:
\begin{gather}
    P(t, \theta) = - V(\theta) \log t + \Phi(\theta) + o(1), \label{eq:P_asymp} \\[0.4em]
    Q(t, \theta) = q(\theta) + t^{2V(\theta)}[ \Psi(\theta) + o(1) ],  \label{eq:Q_asymp}
\end{gather}
where $V, \Phi, q, \Psi: \mathbb{S}^1 \to \R$ are regular functions of $\theta$. Upon substituting into the equations \eqref{eq:P_evol}--\eqref{eq:Q_evol}, one sees that these expansions are self-consistent, at least under the assumption that for each $\theta \in \mathbb{S}^1$, one of the following holds:
\begin{enumerate}[(i)]
    \setlength\itemsep{-0.1em}
    \item
        The coefficient $V(\theta)$ satisfies $0 < V(\theta) < 1$, or
    \item
        $V(\theta) > 0$, but $\partial_{\theta} q(\theta) = 0$, or
    \item
        $V(\theta) < 1$, but $\Psi(\theta) = 0$.
\end{enumerate}
\underline{Note}: in the polarized Gowdy case with $Q \equiv 0$, any value of $V(\theta)$ is permitted.

Deriving asymptotics of the form \eqref{eq:P_asymp}--\eqref{eq:Q_asymp} was a crucial step in Ringstr\"om's proof of the Strong Cosmic Censorship for Gowdy symmetric spacetimes \cite{SCC_PolarizedGowdy, RingstromGowdySCC1, RingstromGowdySCC2}. In particular, for our applications we shall make use of the following, see \cite[Step 1]{SCC_PolarizedGowdy} in the polarized case and \cite[Proposition 1.5]{RingstromGowdySCC1} for the unpolarized case but with with the assumption $0 < V(\theta) < 1$. 

\setcounter{theorem}{-1}
\begin{theorem} \label{thm:asymp_smooth}
    Let $(P, Q)$ be a solution to the Gowdy symmetric wave map system \eqref{eq:P_evol}--\eqref{eq:Q_evol}, arising from regular initial data $(P, Q, t \partial_t P, t \partial_t Q)|_{t = t_0} \in (C^{k+1})^2 \times (C^k)^2$ for some $k \geq 1$. Then for all $\theta_0 \in \mathbb{S}^1$, there exists some $V(\theta_0) \in \R$ such that $- t \partial_t P (t, \theta_0) \to V(\theta_0)$ as $t \to 0$, in other words $- t \partial_t P(t, \theta) \to V(\theta)$ \emph{pointwise}.

    Suppose moreover that either $Q \equiv 0$ i.e.~the solution lies in the polarized Gowdy subclass, or we assume a priori that $0 < V(\theta) < 1$ for all $\theta \in \mathbb{S}^1$. Then $V(\theta)$ is $C^{k-1}$ and moreover there exist $C^{k-1}$ functions $\Phi(\theta)$, $q(\theta)$ and $r(\theta)$ such that the following limits hold uniformly in the $C^{k-1}$ norm as $t \to 0$:
    \begin{gather}
        - t \partial_t P(t, \theta) \to V(\theta), \qquad P(t, \theta) + V(\theta) \log t \to \Phi(\theta), \label{eq:P_asymp_2} \\[0.4em]
        e^{2P} t \partial_t Q (t, \theta) \to r(\theta), \qquad e^{2P} (Q(t, \theta) - q(\theta)) \to \frac{r(\theta)}{2 V(\theta)}. \label{eq:Q_asymp_2}
    \end{gather}
    Note that these imply \eqref{eq:P_asymp} and \eqref{eq:Q_asymp}, given $\Psi(\theta) = \frac{r(\theta)}{2V(\theta)} e^{-2\Phi(\theta)}$.

    Under the same assumptions, if $\lambda(t, \theta)$ solves the remaining equations \eqref{eq:L_evol}--\eqref{eq:L_constraint}, then 
    \begin{equation}
        \lambda(t, \theta) = V^2 (\theta) \log t + L(\theta) + o(1), \label{eq:L_asymp_2}
    \end{equation}
    where $L(\theta)$ obeys the asymptotic momentum constraint $dL = - 2 V \cdot d \Phi$.
\end{theorem}

In a context where \eqref{eq:P_asymp_2}--\eqref{eq:L_asymp_2} hold, a straightforward change of coordinates from the gauge \eqref{eq:gowdy} implies that the metric behaves in a \emph{Kasner-like fashion} near the $t = 0$ singularity -- see Section~\ref{sub:intro_bkl} for a discussion of Kasner-like singularities -- and moreover one can describe the \emph{Kasner exponents} by
\begin{equation} \label{eq:kasnerlike_intro}
    p_1(\theta) = \frac{V^2(\theta) - 1}{V^2(\theta) + 3}, \quad p_2(\theta) = \frac{2 - 2 V(\theta)}{V^2(\theta) + 3}, \quad  p_3(\theta) = \frac{2 + 2 V(\theta)}{V^2(\theta) + 3}.
\end{equation}
These exponents depend on $\theta$, but always satisfy the Kasner relations $p_1 + p_2 + p_3 = p_1^2 + p_2^2 + p_3^2 = 1$.

We remark that in \cite{RingstromGowdySCC1, RingstromGowdySCC2}, Ringstr\"om proves much stronger statements than Theorem~\ref{thm:asymp_smooth}. For instance, he shows that for \emph{any} $\theta_0 \in \mathbb{S}^1$ with $0< V(\theta_0) < 1$, then there exists a neighborhood $I$ of $\theta_0$ in $\mathbb{S}^1$ such that $V(\theta)$ is $C^{k-1}$ in $I \subset \mathbb{S}^1$ and moreover the limits \eqref{eq:P_asymp_2}--\eqref{eq:L_asymp_2} hold when restricted to $\theta \in I$. Furthermore, it is proven that for \emph{generic}\footnote{Here generic means open and dense in the $C^{\infty}$-topology on initial data sets $(P, Q, t \partial_t P, t \partial_Q)|_{t = t_0}$.} smooth initial data, $0 < V(\theta) < 1$ for all $\theta \in \mathbb{S}^1 \setminus S$, where $S$ is a finite set of points known as ``spikes''. See the review article \cite{RingstromGowdyReview} for details.

However, understanding Strong Cosmic Censorship in full generality and the issue of spikes is not the main focus of the present article. Our starting point is instead the following two interrelated questions.

\vspace{3pt}
\begin{quote}
    Can one characterize a wide class of initial Cauchy data (for Gowdy spacetimes) at $t = t_0 > 0$, including data which is not at first glance completely compatible with the asymptotics of Theorem~\ref{thm:asymp_smooth} (e.g.~data such that $0 < - t \partial_t P(t_0, \theta) < 1$ does \underline{not} hold everywhere but $Q(t_0, \theta) \not\equiv 0$), such that the corresponding spacetime reaches a $t = 0$ singularity for which one has the asymptotics \eqref{eq:P_asymp}--\eqref{eq:Q_asymp} with quantitative estimates depending on initial data\footnote{This differs from the approach in \cite{RingstromGowdySCC1, RingstromGowdySCC2}, where the strategy is more akin to showing that the solutions not achieving the appropriate asymptotics are \emph{non-generic} in some sense; in particular there is little use of the relationship between the asymptotic information and initial data at some $t_0 > 0$.}?

    Moreover, for all such data can one fully understand the \emph{intermediate dynamics} of the spacetime between the initial data at $t = t_0$ and the asymptotics at $t = 0$?
\end{quote}
\vspace{3pt}

These questions are interesting since one expects that at $t = 0$ itself, the asymptotic quantity $V(\theta) = \lim_{t \to 0}(-t\partial_t P(t, \theta))$ must generically obey $0 < V(\theta) < 1$, thus there must be some nonlinear instability / transition between $t = t_0$ and $t = 0$. Later, we identify this nonlinear instability mechanism as a ``bounce'' akin to those identified by Belinski, Khalatnikov and Lifshitz in \cite{bkl71}, see Section~\ref{sub:intro_bkl}.

In this article, we prove both the existence of such a class of initial data (Theorem~\ref{thm:stab_rough}) together with a full quantitative understanding of the intermediate dynamics (Theorem~\ref{thm:bounce_rough}), thus answering the above questions in the affirmative. Moreover this class is open in the $C^{\infty}$ topology on initial data, thus our main results will represent a qualitative stability result in the sense that in the regime considered, certain aspects such as curvature blow-up, are stable to perturbation. 

In answering these questions we also uncover an interesting corollary of Theorem~\ref{thm:bounce_rough}, which we interpret as an \emph{instability result} (see Corollary~\ref{cor:bounce} for a precise statement): consider a \emph{polarized} Gowdy solution given by $(P, Q) = (P, 0)$ arising from suitably regular initial data. From Theorem~\ref{thm:asymp_smooth}, the quantity $V_0(\theta) = \lim_{t \to 0} (- t \partial_t P(t, \theta))$ is allowed to take any value. Now consider a perturbation of $(P, 0)$ to $(\tilde{P}, \tilde{Q})$ with $\tilde{Q} \not \equiv 0$. The instability result then says under a ``low-velocity'' assumption $0 < V(\theta) < 2$, the new $\tilde{V}(\theta)$ associated to $\tilde{P}(t, \theta)$ via \eqref{eq:P_asymp_2} will always satisfy $0 < \tilde{V} \leq 1$. In fact, we show that modulo spikes, $\tilde{V}(\theta) \approx \min \{ V(\theta), 2 - V(\theta) \}$, quantifying the ``bounce'' instability exactly.

Note that neither the unperturbed solution $(P, 0)$, nor the perturbation, are required to be spatially homogeneous, and this result together with our companion article \cite{MeSurfaceSymPaper} can be considered the first rigorous evidence of BKL-type bounces \emph{outside of homogeneity}. We refer the reader to \cite[Section 1.6]{MeSurfaceSymPaper} for a detailed comparison between the two results. (BKL bounces of this sort are much better understood in the spatially homogeneous setting, where the dynamics reduce to a system of finite dimensional ODEs \cite{Weaver_bianchi, RingstromBianchi, LiebscherRendallTchapnda, BeguinDutilleul}.)

\subsection{Our main theorems in rough form} \label{sub:intro_thm}

Our two main theorems concern solutions $(P, Q)$ to the Gowdy symmetric system \eqref{eq:P_evol}--\eqref{eq:Q_evol}; we no longer consider the remaining equations \eqref{eq:L_evol}--\eqref{eq:L_constraint}. Local existence for this system yield that for initial data $P_D, Q_D, \dot{P}_D, \dot{Q}_D: \mathbb{S}^1 \to \R$ in a Banach space $X$, e.g.~$X = (C^{k+1})^2 \times (C^k)^2$ or $X = (H^{s+1})^2 \times (H^s)^2$ for $k \geq 1$ or $s \geq \frac{3}{2}$, and any $t_0 > 0$ then there exists an interval $I \subset (0, + \infty)$ containing $t_0$ for which $(P, Q)$ solves \eqref{eq:P_evol}--\eqref{eq:Q_evol}, and
\begin{gather*}
    (P, Q, t \partial_t P, t \partial_t Q) \in C(I, X), \\[0.4em]
    (P, Q, t \partial_t P, t \partial_t Q)|_{t = t_0} = (P_D, Q_D, \dot{P}_D, \dot{Q}_D).
\end{gather*}
For $X$ as above, conservation of energy \eqref{eq:wavemap_energy} implies global existence, i.e.~one may take $I = (0, + \infty)$

We now introduce our two main theorems. The first, Theorem~\ref{thm:stab_rough} is a \emph{stability theorem}, and characterizes a subset in the moduli space of initial data such that for $(P_D, Q_D, \dot{P}_D, \dot{Q}_D)$ and $t_0 > 0$ in this subset, certain quantities including $-t \partial_t P$, $e^P t \partial_{\theta} Q$ and $e^P t \partial_t Q$ remain bounded for $t \in (0,t_0)$. The second, Theorem~\ref{thm:bounce_rough}, referred to as the \emph{bounce theorem}, characterizes the (nonlinear) dynamics arising from initial data as above, including a relationship between the initial data and the eventual asymptotics of $(P, Q)$ as described in \eqref{eq:P_asymp}--\eqref{eq:Q_asymp}.

To precisely define this subset of initial data, we introduce two real-valued parameters $0 < \upgamma < 1$ and $\upzeta > 0$. Then the initial data is chosen to satisfy three conditions:
\begin{itemize}
    \item
        (\emph{Weak subcriticality}) The functions $P_D, Q_D, \dot{P}_D, \dot{Q}_D$ satisfy, for all $\theta \in \mathbb{S}^1$:
        \begin{equation} \label{eq:weak_subcriticality}
            - \dot{P}_D > \upgamma, \quad (1 + \dot{P}_D)^2 + (e^{P_D} t \partial_{\theta} Q_D)^2 < (1 - \upgamma)^2.
        \end{equation}
    \item
        (\emph{Energy boundedness}) For some $N = N(\upgamma) \in \N$, an $N$th order $L^2$ norm of $(P_D, Q_D, \dot{P}_D, \dot{Q}_D)$ is bounded by $\upzeta$, see already \eqref{eq:data_energy_boundedness_P} and \eqref{eq:data_energy_boundedness_Q}.
    \item
        (\emph{Closeness to singularity}) The time $t_0$ at which initial data is prescribed is chosen to satisfy $0 < t_0 \leq t_*$ for some $t_*$ depending on $\upgamma$ and $\upzeta$.
\end{itemize}
Note in the more precise Theorem~\ref{thm:global} and Theorem~\ref{thm:bounce} we introduce additional parameters $\upgamma'$ and $\upgamma''$ which $N$ and $t_*$ could depend on. See more details in Section~\ref{sub:thm_data}.
We remark also that taking the union of all such initial data as $0 < \upgamma < 1$ and $\upzeta > 0$ vary, one characterises the class of initial data to which our results apply as a subset of the ``moduli space of initial data'' for the Gowdy system \eqref{eq:P_evol}--\eqref{eq:Q_evol} which is moreover open in the $C^{\infty}$ topology.

We defer further discussion of these conditions to after the statements below. Our first theorem, the stability theorem, is as follows:

\begin{theorem}[Stability, rough version] \label{thm:stab_rough}
    For $0 < \upgamma < 1$ and $\upzeta > 0$, let $t_0 > 0$ and $(P_D, Q_D, \dot{P}_D, \dot{Q}_D)$ be initial data obeying the three conditions above. Then the corresponding solution $(P, Q)$ to the Gowdy symmetric system \eqref{eq:P_evol}--\eqref{eq:Q_evol} is such that for all $t \in (0, t_0)$, 
    \[
        - t \partial_t P(t, \theta) \geq \upgamma, \quad (1 + t \partial_t P(t, \theta))^2 + (e^P t \partial_{\theta} Q(t, \theta))^2 \geq (1 - \upgamma)^2.
    \]
    Furthermore, there exist functions $V, q: \mathbb{S}^1 \to \mathbb{R}$, with $V$ bounded and $q$ continuous, such that $- t \partial_t P(t, \theta) \to V(\theta)$ pointwise and $Q(t, \theta) \to q(\theta)$ uniformly as $t \to 0$.
    If one further assumes $\upgamma > \frac{1}{3}$, then we have that $q$ is $C^1$ and the convergence $Q(t, \theta) \to q(\theta)$ holds in the $C^1$ topology.
\end{theorem}

Our second main theorem concerns the dynamical behaviour of certain quantities along causal curves. To fix notation, let $\gamma: I \to (0, +\infty) \times \mathbb{S}^1$ be any \emph{causal curve} parameterised by its $t$-coordinate. Using the Gowdy metric \eqref{eq:gowdy} and assuming $\gamma$ to not depend on $\sigma$ and $\delta$, the causal nature means $\gamma(t) = (t, \theta(t))$ with $|\theta'(t)| < 1$. Then define:
\[
    \mathscr{P}_{\gamma}(t) \coloneqq - t \partial_t P (\gamma(t)), \qquad \mathscr{Q}_{\gamma}(t) \coloneqq e^P t \partial_{\theta} Q (\gamma(t)).
\]
It is key that our theorem holds \underline{uniformly} in the choice of causal curve $\gamma$, meaning the ODEs for $\mathscr{P}_{\gamma}$ and $\mathscr{Q}_{\upgamma}$ being independent for different choices of $\gamma$ with distinct endpoints, at least up to error terms which reflect \emph{AVTD behaviour}, see Section~\ref{sub:intro_bkl}. This reflects explicitly that our result is \emph{spatially inhomogeneous}.

\begin{theorem}[Bounces, rough version] \label{thm:bounce_rough}
    Let $(P, Q)$ be a solution to the Gowdy symmetric system \eqref{eq:P_evol}--\eqref{eq:Q_evol} arising from initial data as in Theorem~\ref{thm:stab_rough}. Then for $\gamma(t)$ a timelike curve as above, the quantities $\mathscr{P}_{\gamma}(t)$ and $\mathscr{Q}_{\gamma}(t)$ obey the ODEs:
    \[
        t \frac{d}{dt} \mathscr{P}_{\gamma} = \mathscr{Q}_{\gamma}^2 + \mathscr{E}_{\mathscr{P}}, \qquad t \frac{d}{dt} \mathscr{Q}_{\gamma} = (1 - \mathscr{P}) \mathscr{Q}_{\gamma} + \mathscr{E}_{\mathscr{Q}},
    \]
    where the error terms vanish quickly as $t \to 0$. 

    Furthermore, $\mathscr{Q}_{\gamma}(t)$ converges to $0$ as $t \to 0$, while there exists some $\mathscr{P}_{\gamma, \infty}$ such that $\mathscr{P}_{\gamma}(t) \to \mathscr{P}_{\gamma, \infty}$ as $t \to 0$. In fact, for $V$ as in Theorem~\ref{thm:stab_rough}, $\mathscr{P}_{\gamma, \infty} = V(\theta_0)$ where $\theta_0 = \lim_{t \to 0} \theta(t)$ is the $\theta$-coordinate of the past endpoint of $\gamma$. Finally, we estimate $\mathscr{P}_{\gamma, \infty}$ depending on the value of $\partial_{\theta} q(\theta_0) = \lim_{t \to 0} \partial_{\theta} Q(\gamma(t))$:
    \begin{enumerate}[(i)]
        \item
            If either the limit $\partial_{\theta} q(\theta_0)$ does not exist or $\partial_{\theta} q (\theta_0)$ exists and is nonzero, then necessarily $V(\theta_0) = \mathscr{P}_{\gamma, \infty} \leq 1$ and moreover in the case that $\mathscr{Q}_{\gamma}(t_0)$ is small,
            \[
                \mathscr{P}_{\gamma, \infty} \approx \min \{ \mathscr{P}_{\gamma}(t_0), 2 - \mathscr{P}_{\gamma}(t_0) \} + O (\mathscr{Q}_{\gamma}(t_0)).
            \]
        \item
            Let $\upgamma > \frac{1}{2}$. By Theorem~\ref{thm:stab_rough}, $\partial_{\theta} q(\theta_0)$ exists. If $\theta_0$ is such that $\partial_{\theta} q(\theta_0) = 0$, then one instead has:
            \[
                \mathscr{P}_{\gamma, \infty} \approx \mathscr{P}_{\gamma}(t_0) + O (\mathscr{Q}_{\gamma}(t_0)).
            \]
    \end{enumerate}
\end{theorem}

One interprets Theorem~\ref{thm:bounce_rough}(i) as an instability result for the quantity $\mathscr{P}_{\gamma}$ in the following sense: consider initial data as in Theorem~\ref{thm:stab_rough} such that for some $\theta \in \mathbb{S}^1$ one sets $- \dot{P}_D = - t \partial_t P (t_0, \theta) > 1$ and $\mathscr{Q}_{\gamma}(t_0)$ small. Then for timelike curves through $(t_0, \theta)$, and in the expected-to-be-generic setting where $\partial_{\theta} Q(\gamma(t))$ does not converge to $0$, (i) suggests that as $t \to 0$, $\mathscr{P}_{\gamma}(r)$ will transition from $ - \dot{P}_D > 1$ to (approximately) $2 - \dot{P}_D < 1$.

An application of particular interest concerns not just the instability of $\mathscr{P}_{\gamma}$ along $\gamma$ but actually an \emph{instability of the global asymptotics} of certain \emph{polarized Gowdy spacetimes} with $Q = 0$, with the instability is triggered by \emph{unpolarized perturbations}. In light of Theorem~\ref{thm:asymp_smooth}, in the original polarized spacetime the function $V(\theta)$ always exists and is smooth. To be compatible with Theorem~\ref{thm:stab_rough}, suppose that $0 < V(\theta) < 2$, but now suppose also that $V$ exceeds $1$ somewhere. Upon adding perturbations with $Q \neq 0$, one expects the asymptotics will change; by considering Theorem~\ref{thm:bounce_rough}(i) with $\gamma$ a constant $\theta$-curve, a generic $\theta$ will now have $\mathscr{P}_{\gamma, \infty} = \tilde{V}(\theta_0) \approx \min \{ V(\theta_0), 2 - V(\theta_0) \}$. 

An illustration of this is seen in Figure~\ref{fig:bounce_instability} below (see Corollary~\ref{cor:bounce} for a more general statement). Here, in the polarized Gowdy spacetime with $Q = 0$, the asymptotic quantity $V(\theta)$ remains close to the (inhomogeneous) value of the initial data quantity $- t \partial_t P(t_0, \theta)$. However, in the non-polarized case with $Q \neq 0$, in order to be compatible with $0 < V(\theta) < 1$ we instead see a bounce for most $\theta \in \mathbb{S}^1$.

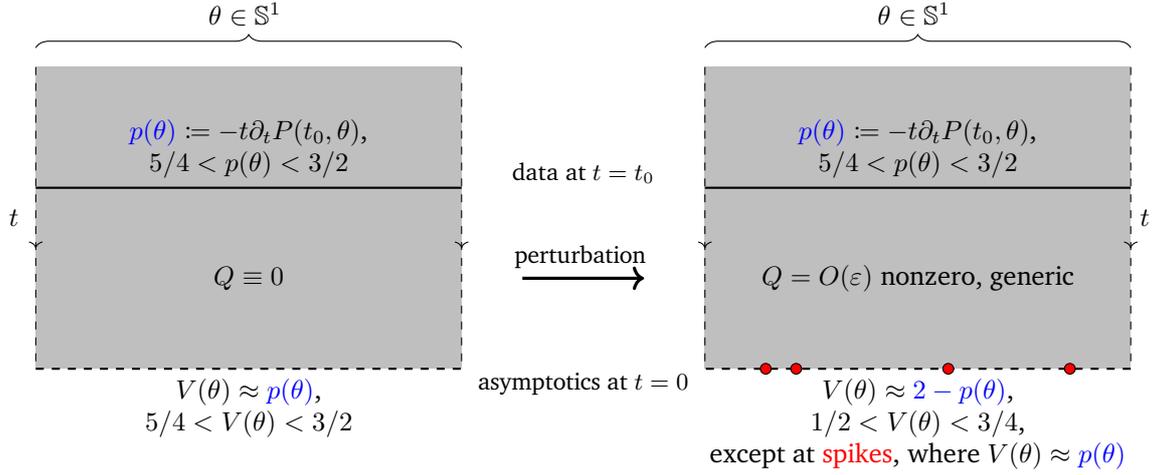
\begin{figure}[h]
    \centering

    \begin{tikzpicture}[scale = 0.8, every text node part/.style={align=center}]
        \coordinate (LLD) at (-9, 0);
        \coordinate (LLM) at (-9, 3);
        \coordinate (LLU) at (-9, 5);
        \coordinate (LRD) at (-2, 0);
        \coordinate (LRM) at (-2, 3);
        \coordinate (LRU) at (-2, 5);
    
        \begin{scope}[decoration=
            {markings, mark=at position 0.6 with {\arrow{>}}}
            ]
            \path[fill=lightgray, opacity=0.5] (LLU) -- (LLD) -- (LRD) -- (LRU) -- (LLU);
    
            \draw[decorate, decoration={brace, raise=5pt, amplitude=5pt}] (-9, 5.1) -- node[midway, above=10pt] {$\theta \in \mathbb{S}^1$} (-2, 5.1);
            \draw[dashed, postaction={decorate}] (LLU) -- node[midway, left] {$t$}  (LLD);  
            \draw[dashed, postaction={decorate}] (LRU) -- (LRD);  
            \draw[thick, dashed] (LLD) -- node[midway, below] {$V(\theta) \approx {\color{blue}p(\theta)}$, \\ $5/4 < V(\theta) < 3/2$} (LRD);
            \draw[thick] (LLM) -- node[midway, above] {${\color{blue}p(\theta)} \coloneqq - t \partial_t P(t_0, \theta)$, \\ $5/4 < p(\theta) < 3/2$} (LRM);

            \node (LC) at (-5.5, 1.5) {$Q \equiv 0$};
        \end{scope}
    
        \coordinate (RRD) at (+9, 0);
        \coordinate (RRM) at (+9, 3);
        \coordinate (RRU) at (+9, 5);
        \coordinate (RLD) at (+2, 0);
        \coordinate (RLM) at (+2, 3);
        \coordinate (RLU) at (+2, 5);

        \begin{scope}[decoration=
            {markings, mark=at position 0.6 with {\arrow{>}}}
            ]

            \path[fill=lightgray, opacity=0.5] (RRU) -- (RRD) -- (RLD) -- (RLU) -- (RRU);
    
            \draw[decorate, decoration={brace, raise=5pt, amplitude=5pt}] (2, 5.1) -- node[midway, above=10pt] {$\theta \in \mathbb{S}^1$} (9, 5.1);
            \draw[dashed, postaction={decorate}] (RRU) -- node[midway, right] {$t$} (RRD);  
            \draw[dashed, postaction={decorate}] (RLU) -- (RLD);  
            \draw[thick, dashed] (RRD) -- node[midway, below] {$V(\theta) \approx {\color{blue}2 - p(\theta)}$, \\ $1/2 < V(\theta) < 3/4$, \\ except at {\color{red} spikes}, where $V(\theta) \approx {\color{blue} p(\theta)}$} (RLD);
            \draw[thick] (RRM) -- node[midway, above] {${\color{blue}p(\theta)} \coloneqq - t \partial_t P(t_0, \theta)$, \\ $5/4 < p(\theta) < 3/2$} (RLM);

            \node (RC) at (5.5, 1.5) {$Q = O(\varepsilon)$ nonzero, generic};
        \end{scope}

        \node at (0, 3.25) {\small data at $t = t_0$};
        \node at (0, -0.25) {\small asymptotics at $t = 0$};
        \node at (3, 0) [circle, draw, inner sep=0.5mm, fill=red] {};
        \node at (3.5, 0) [circle, draw, inner sep=0.5mm, fill=red] {};
        \node at (6, 0) [circle, draw, inner sep=0.5mm, fill=red] {};
        \node at (8, 0) [circle, draw, inner sep=0.5mm, fill=red] {};

        \draw[very thick, ->] (-1, 1.5) -- node[midway, above] {\small perturbation} (1, 1.5);
    \end{tikzpicture}

    \captionsetup{justification = centering}
    \caption{Diagram illustrating an example of a global bounce instability arising from a non-polarized perturbation of a polarized Gowdy spacetime. In the unperturbed spacetime (left) $Q \equiv 0$ and the asymptotic quantity $V(\theta)$ is close to the (inhomogeneous) data {\color{blue} $p(\theta)$}, while in the perturbed spacetime (right) with $Q = O(\varepsilon)$ and generic the bounce of Theorem~\ref{thm:bounce_rough}(i) applied to constant $\theta$-curves $\gamma$ means that $V(\theta)$ is instead close to {\color{blue}$2 - p(\theta)$}, except at finitely many points representing spikes.}
    \label{fig:bounce_instability}
\end{figure}

Note, however, there may certain choices of $\theta_0 \in \mathbb{S}^1$ with $\partial_{\theta} q(\theta_0) = 0$ and where by (ii)\footnote{One might hope that the requirement of $\upgamma > \frac{1}{2}$ can be removed.} instead $\tilde{V}(\theta_0) \approx V(\theta_0)$. If such a $\theta_0$ is an isolated point this means that $\tilde{V}(\theta)$ is discontinuous at $\theta = \theta_0$, and $\theta_0$ is a ``true spike'' as defined in \cite[Section 9]{RingstromGowdyReview}. As mentioned in \cite{RingstromGowdyReview}, true spikes have previously been understood through a series of transformations (related to inversions and the so-called Gowdy--to--Ernst transformation) -- but these previous works do not explain how spikes are related to initial data. Our result is therefore a first step in understanding the dynamical formation of spikes from suitable initial data.

We return briefly to the three conditions above. Our first condition, weak subcriticality, is chosen such that we include spacetimes where the nonlinear bounces above can occur, but such that these bounces are not too wild enough to make quantities such as $\mathscr{P}_{\gamma}$ and $\mathscr{Q}_{\gamma}$ uncontrollable. 

The second condition of energy boundedness is natural as our proof relies on $L^2$ energy estimates at top order. Note that our choice of energy, see Section~\ref{sub:energy}, will be compatible with the asymptotics of Theorem~\ref{thm:asymp_smooth}. We will, however, allow the initial energy bound, $\upzeta > 0$, to be large, meaning we allow large spatial variation. In this case, our final condition, closeness to singularity, is required to prevent this spatial variation from playing a major role in the dynamics. Note that $t_*$ and $\upzeta$ play a similar role to that of $\zeta_1$ and $\zeta_0$ in \cite[Theorem 12]{groeniger2023formation}.

\subsection{Relationship to BKL and the Kasner map} \label{sub:intro_bkl}

We put our results in the context of the physics and mathematics literature regarding spacelike singularities arising in solutions of the Einstein equations. While celebrated ``singularity theorems'' of Penrose \cite{penrose} and Hawking \cite{hawking67} asser that ``singularities'' are a robust prediction of Einstein's equations \eqref{eq:einstein}, it is a known shortcoming that these theorems do not prove singularity formation in the sense of blow-up, but instead (causal) geodesic incompleteness, and thus provide no qualitative or quantitative description of spacetimes near their singular, or incomplete, boundaries.

Our approach is instead to outline the heuristics of Belinski, Khalatnikov and Lifshitz (often abbreviated to BKL) in their investigation of near-singularity solutions to Einstein's equation \cite{kl63, bkl71, bk77, bkl82}. \cite{kl63} proposes the following ansatz for singular solutions to the Einstein equations: a spacetime $\mathcal{M}^{1+3} = (0, T) \times \mathbf{\Sigma}^3$ with singularity located at $\{0\} \times \mathbf{\Sigma}$ has spacetime metric $\mathbf{g}$ with leading order expansion:
\begin{equation} \label{eq:bkl}
    \mathbf{g} = - d \tau^2 + \sum_{I=1}^3 \tau^{2 p_I(x)} \mathbf{\omega}^I(x) \otimes \mathbf{\omega}^I(x) + \cdots.
\end{equation}
The $p_I(x)$ are functions on $\mathbf{\Sigma}$ known as \textit{Kasner exponents}, while $\{\mathbf{\omega}^I(x)\}$ is a frame of $1$-forms on $\mathbf{\Sigma}$. Na\"ively inserting these into \eqref{eq:einstein}, BKL determine that the Kasner exponents must satisfy the \textit{Kasner relations}:
\begin{equation} \label{eq:kasner_relation}
    \sum_{I = 1}^3 p_I(x) = 1, \quad \sum_{I = 1}^3 p_I^2(x) = 1.
\end{equation}

Note \eqref{eq:kasner_relation} implies that aside from the exceptional case where $(p_1, p_2, p_3) = (1, 0, 0)$ and permutations thereof, exactly one of the $p_I(x)$ must be negative. 
A consistency check using the Einstein equations suggests that for the ansatz \eqref{eq:bkl} to remain valid all the way to $t = 0$, the 1-form $\omega^I$ associated to the negative exponent $p_I$ must be integrable in the sense of Frobenius, meaning $\omega^I \wedge d \omega^I = 0$. On top of the Kasner relations \eqref{eq:kasner_relation}, this integrability condition provides another obstruction to the validity of \eqref{eq:bkl}\footnote{In the original paper of Khalatnikov and Lifshitz \cite{kl63}, this observation together with a na\"ive function counting argument led them to believe that singularities are in fact non-generic! However in the later paper \cite{bkl71} this observation waas instead understood as an instability; this is the modern viewpoint.}.

In \cite{bkl71} BKL then give heuristics suggesting what happens if this integrability condition fails. In \cite{bkl71}, BKL assume \emph{Asymptotically Velocity Term Dominated (AVTD)} behaviour, meaning that `spatial derivatives' are subdominant in comparison to $\partial_{\tau}$-derivatives. This has the effect that dynamics in the future light cones emanating from distinct points $(0, p) \in \{ 0 \} \times \mathbf{\Sigma}$ on the singularity are decoupled.

With this assumption they then provide a computation suggesting that in such a light cone, the ansatz \eqref{eq:bkl} is valid for $\tau \gg \tau_B$ for $\tau_B$ some critical time, but for $\tau \ll \tau_B$ the metric transitions to a form that again resembles \eqref{eq:bkl} but with Kasner exponents $p_I(x)$ replaced by new exponents $\acute{p}_I(x)$. Between these regimes, the spacetime undergoes a nonlinear transition from one regime to another, often denoted in the literature as a \emph{BKL} or \emph{Kasner bounce}. Typically these nonlinear transitions, or bounces, then cascade indefinitely, giving rise to BKL's \emph{chaotic and oscillatory approach to singularity}. 

Within this computation, BKL even propose a formula describing how the Kasner exponents change during the transition: if $p_2 < 0$ and $\omega^2 \wedge d \omega^2 \neq 0$ then
\begin{equation} \label{eq:kasner_relation_bounce}
    \acute{p}_1 = \frac{p_1 + 2 p_2}{1 + 2 p_2}, \quad \acute{p}_2 = -\frac{p_2}{1 + 2 p_2}, \quad \acute{p}_3 = \frac{p_3 + 2 p_2}{1 + 2 p_2}.
\end{equation}

To relate this to our results in Gowdy symmetry, we need to apply a change of variables from the areal time coordinate $t$ in \eqref{eq:gowdy} to the Gaussian time coordinate $\tau$ in \eqref{eq:bkl}. Using \eqref{eq:gowdy}, we expect $d\tau \approx t^{-\frac{1}{4}} e^{\frac{\lambda}{4}} dt$, while by \eqref{eq:L_asymp_2} we have $\lambda(t, \theta) \approx V^2(\theta) \log t + O(1)$. So $\tau \sim t^{\frac{V^2 + 3}{4}}$. We then formally set $\omega^1 = d\theta$, $\omega^2 = d \sigma + q(\theta) d \delta$ and $\omega^3 = d \delta$, where $q(\theta) = \lim_{t \to 0} Q(t, \theta)$. Through Theorem~\ref{thm:asymp_smooth}, this allows one to make a formal analogy between \eqref{eq:gowdy} and \eqref{eq:bkl}, with Kasner exponents as in \eqref{eq:kasnerlike_intro}. 

While $\omega^1 = d \theta$ and $\omega^3 = d \delta$ are closed and thus integrable in the sense of Frobenius, one checks that $\omega^2 \wedge d \omega^2 = q'(\theta) d \sigma \wedge d \theta \wedge d \delta$. Hence the consistency check fails if $q'(\theta) \neq 0$ and $p_2 < 0$. From \eqref{eq:kasnerlike_intro}, $p_2 = p_2(\theta) < 0$ if and only if $V(\theta) > 1$. So one expects $V(\theta) = 1$ to be a threshold between stable and unstable behaviour. Indeed, for $0 < V(\theta) < 1$ we expect stable\footnote{The additional requirement of $V(\theta) > 0$ is more due to our choice of gauge.} self-consistent behaviour as in Theorem~\ref{thm:asymp_smooth}, while the bounce result Theorem~\ref{thm:bounce_rough} indicates that $V(\theta) = 1$ is a threshold. Furthermore, one may check that via the correspondence \eqref{eq:kasnerlike_intro} and the transition map \eqref{eq:kasner_relation_bounce}, a bounce exactly corresponds to a transition $V(\theta) \mapsto \acute{V}(\theta) = 2 - V(\theta)$. In the ``low-velocity'' regime of $0 < V(\theta) < 2$ that we consider, note that there is \emph{at most one Kasner bounce} before either $V(\theta)$ or $\acute{V}(\theta)$ lies in the interval $(0, 1)$.

We review the heuristics regarding how the bounce map \eqref{eq:kasner_relation_bounce} was found. BKL's proposal was that assuming AVTD, then in the future light cone emanating from any point on the singularity the metric is well-approximated by something spatially homogeneous. For spatially homogeneous $\mathbf{g}$, the Einstein equations \eqref{eq:einstein} reduce to a system of finite dimensional nonlinear autonomous ODEs for a correctly chosen time coordinate. The bounce map \eqref{eq:kasner_relation_bounce} then arises from understanding orbits of this ODE system, in particular heteroclinic orbits between its unstable fixed points. 

In the mathematical literature, progress regarding spacelike singularities and the BKL ansatz has largely regarded the ``subcritical'' setting, where either the integrability condition $\omega^I \wedge d \omega^I = 0$ holds due to symmetry, or due to the addition of matter fields e.g.~scalar fields or so-called stiff fluids, which allow more general stable regimes to be found. We mention the breakthrough work of Fournodavlos--Rodnianski--Speck \cite{FournodavlosRodnianskiSpeck}, as well as the related \cite{groeniger2023formation, RodnianskiSpeck1, RodnianskiSpeck2, SpeckS3, BeyerOliynyk, FajmanUrban}. There are also related results which involve prescribing the asymptotic data i.e.~$p_I(x), \omega^I(x)$ in \eqref{eq:bkl} and ``solving from the singularity'' to determine a spacetime achieving this near-singularity ansatz at leading order, see in particular Fournodavlos--Luk \cite{FournodavlosLuk}. 

Regarding the study of nonlinear bounces, to the best of the author's knowledge all previous work concerns only spatially homogeneous spacetimes, where the Einstein equations reduce an exact system of nonlinear ODEs. See for instance studies of solutions for various Einstein--matter systems in Bianchi symmetry \cite{Weaver_bianchi, RingstromBianchi}, as well as a recent work of the author together with Van de Moortel \cite{MeVdM}. 
In particular the current article and our companion article \cite{MeSurfaceSymPaper}, which concerns the Einstein--Maxwell--scalar field model in surface symmetry, are the first works to understand BKL bounces, albeit only a single such bounce, \emph{outside of the spatially homogeneous setting}. See \cite[Section 1.6]{MeSurfaceSymPaper} for a thorough introduction to the model discussed there and a detailed comparison of the two papers.

Before moving to a sketch of the proof, we propose two conjectures which would go beyond a \emph{single} BKL bounce, but remains in the $1+1$-dimensional setting. The first conjecture still concerns Gowdy symmetry but with multiple bounces. Note that this would mean exiting the low-velocity regime with $0 < V(\theta) <2$.

\begin{conjecture} \label{con:gowdymultiple}
    There exists a large, open class of initial data for the Gowdy symmetric system which exhibit any \underline{finite} number of BKL bounces. 
\end{conjecture}
One expects that this would require a different choice of gauge from \eqref{eq:gowdy}, since the results of this paper suggest that $\partial_{\theta}$ is not a good ``spatial derivative'' outside of the low-velocity regime; see already Section~\ref{sub:intro_proof}. To go beyond a finite number of bounces to a possibly infinite number, one needs to move beyond Gowdy symmetry, and into the realm of more general $\mathbb{T}^2$-symmetric spacetimes, see \cite{T2Symmetry} for definitions.
\begin{conjecture} \label{con:t2bounce}
    One can describe initial data for the Einstein vacuum equations in $\mathbb{T}^2$ symmetry which exhibit an \underline{infinite} number of BKL bounces. 
\end{conjecture}
The reason to study $\mathbb{T}^2$ symmetry is that the symmetry assumption no longer enforces that any of the $\omega^I(x)$ are integrable in the sense of Frobenius, thus giving the potential for infinite bounces. This is expected to significantly more difficult than Conjecture~\ref{con:gowdymultiple} since the corresponding autonomous ODE system is chaotic. Nevertheless, one would hope the resolution of Conjecture~\ref{con:t2bounce} would be a key initial step towards understanding the heuristics of BKL in vacuum outside of symmetry.

\subsection{Sketch of the proof} \label{sub:intro_proof}

The proofs of the stability result Theorem~\ref{thm:stab_rough} and the bounce result Theorem~\ref{thm:bounce_rough} are done in tandem and there are three major steps. Recalling the main evolution equations \eqref{eq:P_evol}--\eqref{eq:Q_evol}, we caricature these steps as follows:
\begin{enumerate}[Step 1:]
    \item
        \textbf{Analysis in the spatially homogeneous case:  } In this step, we ignore certain terms in \eqref{eq:P_evol}--\eqref{eq:Q_evol} involving $\partial_{\theta}$-derivatives which one expects to become negligible as $t \to 0$. We do, however, keep terms involving the expression $e^P t \partial_{\theta} Q$. The result is a nonlinear ODE system for the quantities $- t \partial_t P$ and $e^P t \partial_{\theta} Q$; one then analyzes this system to determine bounds for these quantities.

    \item
        \textbf{Linearization of the ODE system: } In the second step, we consider the \emph{best possible behaviour} of the terms involving further $\partial_{\theta}$-derivatives i.e.~the non-spatially homogeneous corrections. We do this by taking commuting $\partial_{\theta}$ with the ODE system in Step 1, resulting in a new \emph{linear} ODE system for new quantities whose coefficients are given by the solution (i.e.~some orbit) of the ODE system in Step 1.

    \item
        \textbf{Energy estimates: } Finally, one must ensure that we can close our argument without loss of derivatives (necessary due to the existence of top order terms such as $\partial_{\theta}^2 P$ in \eqref{eq:P_evol}.) One achieves this via $L^2$ energy estimates, which are allowed to blow up but only at a controlled rate as $t \to 0$.
\end{enumerate}

We now explain in a little more detail how these steps apply to the Gowdy symmetric system. For Step 1, with $\mathscr{P} = - t \partial_t P$ and $\mathscr{Q} = e^P t \partial_{\theta} Q$, the spatially homogeneous ODE system one gets is:
\begin{equation} \label{eq:ode_pq}
    t \partial_t \mathscr{P} = \mathscr{Q}^2, \qquad t \partial_t \mathscr{Q} = (1 - \mathscr{P}) \mathscr{Q}.
\end{equation}
(Eventually we also include another ODE variable $\mathscr{R} = e^P t \partial_t Q$ but we ignore $\mathscr{R}$ this for now.) The upshot is that one can study the 2-dimensional ODE system \eqref{eq:ode_pq} via phase plane analysis, as we see in Figure~\ref{fig:phase_portrait} below. The direction of the arrows in Figure~\ref{fig:phase_portrait} is with respect to $t \downarrow 0$.

\begin{figure}[ht]
    \centering
    \begin{tikzpicture}[scale=2, decoration={markings, 
        mark=at position 0.3 with {\arrow{>}}, mark=at position 0.7 with{\arrow{>}}}]
        \draw[->] (-2.5,0) -- (2.5, 0) node[right, xshift=4pt] {$\mathscr{Q} \coloneqq e^P t \partial_{\theta} Q$};
        \draw[->] (0,-0.5) -- (0, 2.5) node[above, yshift=4pt] {$\mathscr{P} \coloneqq - t \partial_t P$};
        
        \foreach \pos in {-2, -1, 0, 1, 2}
          \draw[shift={(\pos,0)}] (0pt,2pt) -- (0pt,-2pt) node[below] {$ $};
        \foreach \pos in {0 ,1,2}
          \draw[shift={(0,\pos)}] (2pt,0pt) -- (-2pt,0pt) node[left] {$ $};
        
        \foreach \alf in {0.25, 0.5, 0.75, 1.0, 1.25}
            \draw[thick, variable=\th, domain=0:3.1415, samples=100, postaction={decorate}]
                plot ({\alf * sin(deg(\th))}, {1 + \alf * cos(deg(\th))});
        \foreach \alf in {0.25, 0.5, 0.75, 1.0, 1.25}
            \draw[thick, dashed, variable=\th, domain=0:3.1415, samples=100, postaction={decorate}]
                plot ({- \alf * sin(deg(\th))}, {1 + \alf * cos(deg(\th))});
        
        \fill[blue] (0, 1) circle (1pt)
            node [left ,fill=white,xshift=-4pt] {$1$};
        \fill[blue] (0, 1.5) circle (1pt)
            node [left ,fill=white,xshift=-4pt] {$\alpha$};
        \fill[blue] (0, 0.5) circle (1pt)
            node [left ,fill=white,xshift=-4pt] {$2 - \alpha$};
    \end{tikzpicture}
    \captionsetup{justification = centering}
    \caption{Phase portrait showing the dynamics of $\mathscr{P}$ and $\mathscr{Q}$ towards $t = 0$ in the spatially homogeneous case. A bounce is a transition from {\color{blue} $\alpha$} to {\color{blue} $2 - \alpha$}.}
    \label{fig:phase_portrait}
\end{figure}
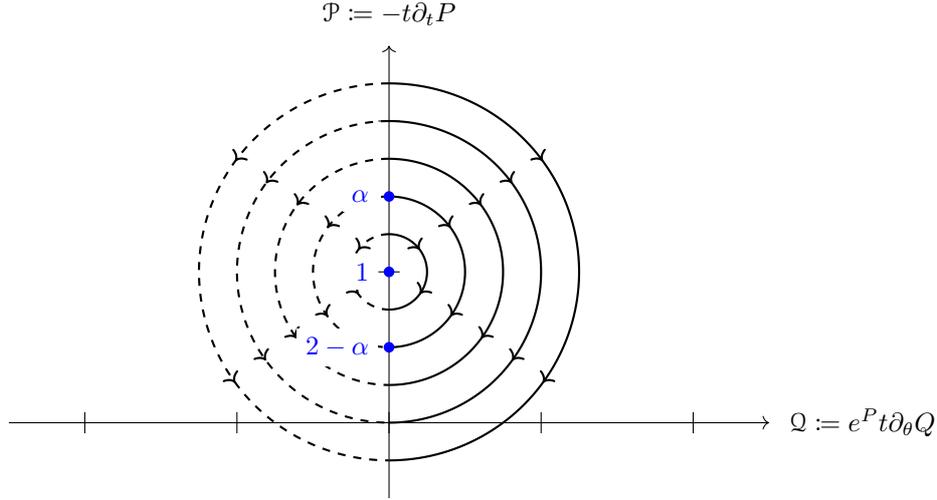

The ODE system \eqref{eq:ode_pq} contains a line of fixed points at $\mathscr{Q} = 0$. Each of these fixed points represents an exact Kasner solution\footnote{A Kasner solution to the Einstein equations \eqref{eq:einstein} is spatially homogeneous and takes the form of \eqref{eq:bkl} with $p_I(x)$ constant and $\omega^I$ equal to the exact differential $d x^I$ for coordinates $x^I$. See \cite{kl63}, or the original paper of Kasner \cite{kasner}.} with exponents $p_I$ as in \eqref{eq:kasnerlike_intro} where $V(\theta)$ is constant and equal to the value of $\mathscr{P}$ at the fixed point. These fixed points are (orbitally) stable as $t \downarrow 0$ if $\mathscr{P} \leq 1$ and unstable otherwise.

The dynamics of the ODE system in the remaining region $\mathscr{Q} \neq 0$ may be described as a union of \emph{heteroclinic orbits} linking an unstable fixed point to a stable fixed point. In fact, one may solve the system exactly using a conserved quantity $\mathscr{K} = (\mathscr{P}-1)^2 + \mathscr{Q}^2$. Thus the heteroclinic orbits, which we identify as BKL bounces, link the unstable fixed point $(\mathscr{P} = \alpha, \mathscr{Q} = 0)$ to the stable fixed point $(\mathscr{P} = 2 - \alpha, \mathscr{Q} = 0)$, for any $\alpha > 1$.

We link this to our conditions preceding Theorem~\ref{thm:stab_rough}. We note that the condition \eqref{eq:weak_subcriticality} of weak subcriticality together with the conserved quantity $\mathscr{K} = (\mathscr{P} - 1)^2 + \mathscr{Q}^2$ means that any $(\mathscr{P}, \mathscr{Q})$-orbit associated to $\mathcal{O}_{\upgamma, \upzeta}$ remains in the bounded region $\{ \mathscr{K} \leq (1 - \upgamma)^2 \}$. This will be essential in Steps 2 and 3.

In Step 2, we commute the ODEs \eqref{eq:ode_pq} with $\partial_{\theta}$.  This yields:
\begin{equation} \label{eq:ode_lm}
    t \partial_t (\partial_{\theta} \mathscr{P}) = 2 \mathscr{Q} \partial_{\theta} \mathscr{Q}, \qquad t \partial_t (\partial_{\theta} \mathscr{Q}) = - \mathscr{Q} \partial_{\theta} \mathscr{P} + (1 - \mathscr{P}) \partial_{\theta} \mathscr{Q}.
\end{equation}
We see this is a \emph{linear} ODE system for $\partial_{\theta} \mathscr{P}$ and $\partial_{\theta} \mathscr{Q}$ whose coefficients are functions of the dynamical $\mathscr{P}$ and $\mathscr{Q}$ orbit from Step 1. The fact that the orbit remains in the bounded region $\{ \mathscr{K} \leq (1 - \upgamma)^2 \}$ allows one to prove, using the linear system \eqref{eq:ode_lm}, that $|\partial_{\theta} \mathscr{P}| + |\partial_{\theta} \mathscr{Q}| = O(t^{- (1 - \upgamma)})$. 

So these $\partial_{\theta}$-derivatives are permitted to blow up as $t \to 0$, but at a controlled rate. This blow-up rate is sharp in the following sense: if our $(\mathscr{P}, \mathscr{Q})$-orbit is a constant orbit at an unstable fixed point with $\mathscr{P} = 2 - \upgamma - \varepsilon$ with $\varepsilon$ small then generic solutions of \eqref{eq:ode_lm} have $\partial_{\theta} \mathscr{Q} = O(t^{-(1 - \upgamma) + \varepsilon})$. This choice of $(\mathscr{P}, \mathscr{Q})$-orbit is related to spikes; it is unsurprising that near spikes certain $\partial_{\theta}$-derivatives will blow up.

But in any case, the outcome of Step 2 is that heuristically, each $\partial_{\theta}$-derivative comes with a loss of $t^{-(1- \upgamma)}$, for instance one expects $\partial_{\theta}^2 P = O(t^{- 2 (1 - \upgamma)})$; this confirms the AVTD expectation that $\partial_{\theta}$-derivatives cost fewer powers of $t$ than $\partial_t$-derivatives (note this requires $\upgamma > 0$ and is exactly where we require that our spacetimes are ``low-velocity'' with $0 < \mathscr{P} < 2$). This is good news, since in \eqref{eq:P_evol} the term $t^2 \partial_{\theta}^2 P$ will be $O(t^{2 \upgamma})$, and thus negligible as $t \to 0$. So throwing this term away to derive the first ODE of \eqref{eq:ode_pq} in Step 1 is justifiable, at least at the heuristic level. Similar arguments hold for the other terms we threw away in Steps 1 and 2.

Step 3 concerns making the argument of the above paragraph rigorous. In particular, we must overcome the issue of \emph{derivative loss} due to terms such as $\partial_{\theta}^2 P$. This is achieved via energy estimates; let $\mathcal{E}^{(K)}(t)$ represent an 
$L^2$-based energy which resembles \eqref{eq:wavemap_energy} but controlling $\partial_{\theta}^K P$ and $\partial_{\theta}^K Q$ in place of $P$ and $Q$, see already Definition~\ref{def:energy}. We explain the main energy estimate: commuting \eqref{eq:P_evol} with $\partial_{\theta}^K$ yields the following
\begin{multline*}
    (t \partial_t)^2 \partial_{\theta}^K P - (t \partial_{\theta})^2 \partial_{\theta}^K P = \overbrace{2 e^{2P} \left( t \partial_t Q \cdot t \partial_t \partial_{\theta}^K Q - t \partial_{\theta} Q \cdot t \partial_{\theta}^{K+1} Q \right)}^{(\mathrm{I})} \\
    + \underbrace{2 e^{2P} \left( (t \partial_t Q)^2 - (t \partial_{\theta} Q)^2 \right) \partial_{\theta}^K P}_{(\mathrm{II})} + \underbrace{\text{lower order quantities}}_{(\mathrm{III})}.
\end{multline*}
A standard argument will then yield the differentiated energy estimate:
\[
    \left| t \frac{d}{dt} \mathcal{E}^{(K)}(t) \right| \leq \sqrt{ 2 \mathcal{E}^{(K)}(t) } \cdot \left( \| (\mathrm{I}) + (\mathrm{II}) + (\mathrm{III}) \|_{L^2} \right).
\]

It remains to estimate each of $(\mathrm{I})$, $(\mathrm{II})$ and $(\mathrm{III})$ in $L^2$. Here we need to use Step 1 and Step 2, the key observation being that $e^P t \partial_{\theta} Q$, represented by $\mathscr{Q}$ in Step 1, is bounded by $1 - \upgamma$. It turns out that $e^P t \partial_t Q$ is similarly bounded, and thus for appropriately defined $\mathcal{E}^{(K)}(t)$ the terms $(\mathrm{I})$ and $(\mathrm{II})$ may be bounded by:
\[
    \| (\mathrm{I}) \|_{L^2} + \| (\mathrm{II}) \|_{L^2} \leq A_* \sqrt{ 2 \mathcal{E}_{\phi}^{(K)} (r) },
\]
where $A_*$ is a constant only depending on $\upgamma$.

For the remaining lower order terms $(\mathrm{III})$, we will also have to use Step 2. To simplify the exposition, in this sketch we replace $(\mathrm{III})$ by just the single term $c \cdot \partial_{\theta} (e^P t \partial_{\theta} Q) \cdot e^P t \partial_{\theta}^{K} Q$, where $c$ may depend on $K$. One can think of this as $c \cdot \partial_{\theta} \mathscr{Q} \cdot e^P t \partial_{\theta}^K Q$, and thus via Step 2 is $O(t^{-(1- \upgamma)}) \cdot e^P t \partial_{\theta}^K Q$. This seems alarming at first sight, since this $O(t^{- (1 - \upgamma)})$ prefactor is not integrable with respect to $\frac{dt}{t}$ as $t \to 0$. But $e^P t \partial_{\theta}^K Q$ is not dependent on the $K$th order energy $\mathcal{E}^{(K)}(t)$, but instead $\mathcal{E}^{(K-1)}(t)$. So
\[
    \| (\mathrm{III}) \|_{L^2} \leq C_{\upgamma, K} \, t^{- (1 - \upgamma)} \sqrt{ 2 \mathcal{E}^{(K-1)}(t)}.
\]
Our eventual estimate for $(\mathrm{III})$ will also involve $\mathcal{E}^{(k)}(t)$ for all $k < K$, see already Propositions~\ref{prop:P_energy} and \ref{prop:Q_energy}. 

But sticking to our simplified $(\mathrm{III})$, combining all of the above, as well as combining with a similar energy estimate using the $Q$-wave equation \eqref{eq:Q_evol}, yields
\begin{equation} \label{eq:energy_der_intro}
    \left| t \frac{d}{dt} \mathcal{E}^{(K)}(t) \right| \leq 2 A_* \mathcal{E}^{(K)}(t)  + 2  C_{\upgamma, K} \, t^{- (1 - \upgamma)} \sqrt{ \mathcal{E}^{(K)} (t) } \cdot \sqrt{ \mathcal{E}^{(K-1)}(t)}.
\end{equation}
When $K = 0$, the last term on the right hand side is absent. Because of the $2 A_*$, even the $0$th order energy $\mathcal{E}^{(0)}(t)$ will blow up as $t^{-2 A_*}$ as $t \to 0$. Further, the appearance of $t^{- (1 - \upgamma)}$ means that as $K$ increases the rate of blow up also increases. Indeed, one can use \eqref{eq:energy_der_intro} to show:
\[
    \mathcal{E}^{(K)}(t) \leq D_{\upgamma, K, \upzeta} \, t^{-2 A_* - 2 K (1 - \upgamma)},
\]
where $D_{\upgamma, K, \upzeta}$ depends on $\upgamma$ and the data, as well as the regularity index $K$. 

It is crucial that $A_*$ is \emph{independent of $K$}. The reason for this is that one now applies an interpolation argument together with the above, much like what is done in \cite{FournodavlosRodnianskiSpeck}, to show e.g.~that $\partial_{\theta}^2 P = O(t^{2 (1 - \upgamma) - \updelta})$, with $\updelta \to 0$ as $N \to \infty$, where $N$ is the maximum regularity index for which one applies the energy estimate. This means that one may interpret the ignored terms in Step 1 and Step 2 as negligible errors, as claimed.

To apply Steps 1 to 3 to the nonlinear problem, one uses a standard bootstrap argument. In our proof, we shall actually perform Step 3 first. That is, one bootstraps $L^{\infty}$ bounds on $0$th order and $1$st order quantities, see \eqref{eq:bootstrap_pq}--\eqref{eq:bootstrap_dth2}, and uses these to derive the energy estimate \eqref{eq:energy_der_intro}. Then one combines this estimate, an interpolation argument, and the ODE analysis of Steps 1 and 2, to improve the bootstrap assumptions. This will yield the stability result Theorem~\ref{thm:stab_rough}, while the bounce result Theorem~\ref{thm:bounce_rough} follows from more detailed ODE analysis.

\subsection*{Acknowledgements}

The author thanks Mihalis Dafermos for valuable advice in the writing of this manuscript. We also thank Igor Rodnianski and Hans Ringstr\"om for insightful discussions and suggestions.

%% file: theorem.tex

\section{Precise statement of the main theorems} \label{sec:theorem}

Below we state the detailed versions of our main results, the stability result Theorem~\ref{thm:global} and the bounce result Theorem~\ref{thm:bounce}. We first introduce various parameters $\upgamma, \upgamma', \upgamma'', \upzeta, N, \updelta, \ldots$ that appear throughout.

\subsection{Setup of the initial data} \label{sub:thm_data}

\subsubsection{Notation and key parameters} \label{subsub:param}

\begin{itemize}
    \item
        The real number $0 < \upgamma < 1$ captures the size of $- t \partial_t P$ allowable in our results. Many of the subsequent parameters will depend on $\upgamma$, as do the constants appearing in our quantitative estimates. We further define $\upgamma', \upgamma''$ to be related constants that satisfy
        \begin{equation} \label{eq:gamma_relation}
            0 < \frac{\upgamma}{2} < \upgamma' < \upgamma < \upgamma'' < 1.
        \end{equation}
        
    \item
        The natural number $N$ represents the largest number of $\partial_{\theta}$-derivatives with which we commute the Gowdy system \eqref{eq:P_evol}--\eqref{eq:Q_evol} in deriving energy estimates. Our choice of $N$ depends on $\upgamma, \upgamma', \upgamma''$, and we expect $N \to \infty$ as $\upgamma \to 0$. We often use $K$ to denote an integer with $0 \leq K \leq N$.

    \item
        The real number $\upzeta > 0$ represents the maximal admissible size of initial data in the $L^2$ sense. See already the initial data bounds \eqref{eq:data_energy_boundedness_P}--\eqref{eq:data_energy_boundedness_Q} below. 

    \item
        The real number $C_* > 0$ is used in the bootstrap argument, see Section~\ref{sub:bootstrap}. We choose $C_*$ depending on $\upgamma$, $\upgamma'$, $\upgamma''$ and $\upzeta$, though we do not make this explicit. The number $A_* > 0$ will depend on $C_*$ and represents the rate at which our energies may blow up.

    \item
        The real number $t_* > 0$ represents how close to the $t = 0$ singularity our initial data is required to be in order to obtain our results. That is, our results apply for initial data at $t = t_0$ for $0 < t_0 \leq t_*$. Our choice of $t_*$  will depend on $\upgamma$ and $\upzeta$ and we expect $t_* \to 0$ as either $\upgamma \to 0$ or $\upzeta \to \infty$.

    \item
        Other constants, often denoted $C$, will be allowed to depend on all of aforementioned parameters. The notation $\updelta$ will be used to represent quantities depending on all of these parameters (e.g.~$\upgamma, \upzeta, t_*$) such that $\updelta \downarrow 0$ as $t_* \downarrow 0$. We often allow abuse of notation such as ``$\updelta + \updelta = \updelta$'' or ``$C_* \updelta = \updelta$'' etc.

\end{itemize}

\subsubsection{Initial \texorpdfstring{$L^{\infty}$}{L\_infty} and \texorpdfstring{$L^2$}{L\_2} bounds}

For the Gowdy symmetric system \eqref{eq:P_evol}--\eqref{eq:Q_evol}, initial data is given for $(P, Q)$ and their $t \partial_t$--derivatives at $\{ t = t_0 \}$ for some $0 < t_0 \leq t_*$.
That is, we let $P_D, Q_D, \dot{P}_D, \dot{Q}_D: \mathbb{S}^1 \to \R$ with

\begin{equation} \label{eq:data}
    (P_D, Q_D, \dot{P}_D, \dot{Q}_D) = (P, Q, t \partial_t P, t \partial_t Q)|_{t = t_0}
\end{equation}

We characterize an open set $\mathcal{O}_{\upgamma, \upzeta}$ (perhaps more precisely $\mathcal{O}_{\upgamma, \upgamma', \upgamma'', \upzeta}$) to which our results apply. Initial data $(P_D, Q_D, \dot{P}_D, \dot{Q}_D)$ in this open set will obey the following pointwise bounds: 
\begin{gather} \label{eq:data_linfty}
    - \dot{P}_D \geq \upgamma'', \quad (1 - \dot{P}_D)^2 + (e^{P_D} t_0 \partial_{\theta} Q_D)^2 \leq (1 - \upgamma'')^2
    \\[0.3em] \label{eq:data_linfty2}
    e^{P_D} \dot{Q}_D \leq \upzeta t_0^{\upgamma}, \quad e^{P_D} \leq \upzeta t_0^{\upgamma},
\end{gather}
as well as the following $L^2$ energy bounds at $t = t_0$, where $N$ depends on $\upgamma$ and $\upzeta > 0$:
\begin{gather} \label{eq:data_energy_boundedness_P}
    \frac{1}{2} \sum_{K = 0}^N \left(  \int_{\mathbb{S}^1} (\partial_{\theta}^K \dot{P}_D)^2 + (t_0 \partial_{\theta}^{K+1} P_D)^2  + (t_0^{\upgamma} \partial_{\theta}^K P_D)^2 \, d\theta \right) \leq \upzeta, \\[0.4em]
    \label{eq:data_energy_boundedness_Q}
    \frac{1}{2} \sum_{K = 0}^N \left(  \int_{\mathbb{S}^1} e^{2P_D} \left( ( \partial_{\theta}^K \dot{Q}_D )^2 + (t_0 \partial_{\theta}^{K+1} Q_D)^2 \right) \, d\theta \right) \leq \upzeta.
\end{gather}
Finally, one assumes $0 < t_0 \leq t_*$, with $t_*$ is chosen small depending on $\upgamma, \upzeta, \upgamma', \upgamma''$.

\subsection{Theorem~\ref{thm:global} -- Stability}

Theorem~\ref{thm:global} is the precise version of our stability result, see the rough version Theorem~\ref{thm:stab_rough}.

\begin{theorem}[Stability] \label{thm:global}
    Consider initial data $(P_D, Q_D, \dot{P}_D, \dot{Q}_D)$ for the Gowdy symmetric system \eqref{eq:P_evol}--\eqref{eq:Q_evol}, which obeys \eqref{eq:data_linfty}, \eqref{eq:data_linfty2}, \eqref{eq:data_energy_boundedness_P}, \eqref{eq:data_energy_boundedness_Q} with $0 < t_0 \leq t_*$. Then for $N$ sufficiently large depending on $\upgamma$ and for $t_*$ chosen sufficiently small, depending on $\upgamma$, $N$ and $\upzeta$, for $0 < t \leq t_0$ we have the following $L^{\infty}$ bounds:
    \begin{equation} \label{eq:global_linfty}
        \upgamma \leq - t \partial_t P \leq 2 - \upgamma, \quad |e^P t \partial_{\theta} Q|  \leq 1 - \upgamma.
    \end{equation}
    There exists $C = C(\upgamma, \upgamma', \upgamma'', \upzeta)$ and $A_* = A_*(\upgamma)$ such that one has the $L^2$ energy bound:
    \begin{equation} \label{eq:global_energy}
        \sum_{K = 0}^N \int_{\mathbb{S}^1} (t \partial_t \partial_{\theta}^K P)^2 + t^2 (\partial_{\theta}^{K+1} P)^2 + ( \partial_{\theta}^K P)^2 + e^{2P} (t \partial_t \partial_{\theta}^K Q)^2 + e^{2P} t^2 (\partial_{\theta}^{K+1} Q)^2 \, d\theta \leq C t^{-2A_* -2K (1 - \upgamma')}.
    \end{equation}

    Furthermore, there exists $V: \mathbb{S}^1 \to \R$ bounded so that $-t \partial_t  P(t, \theta) \to V(\theta)$ pointwise as $t \to 0$. Finally for $k = 1$ whenever $\upgamma > \frac{1}{3}$ and $k = 0$ otherwise, there exists a $C^k$ function $q: \mathbb{S}^1 \to \R$ so that $Q(t, \theta) \to q(\theta)$ as $t \to 0$ in the $C^k$ topology. 
\end{theorem}

\subsection{Theorem~\ref{thm:bounce} -- BKL bounces}

In Theorem~\ref{thm:bounce}, we make precise our bounce result, see the rough version in Theorem~\ref{thm:bounce_rough}. The idea is that certain quantities will obey a system of nonlinear ODEs, plus error terms, and that one may use properties of the ODEs to understand certain aspects of the dynamics.

\begin{theorem}[Bounces] \label{thm:bounce}
    Let $(P, Q)$ be a solution to the Gowdy symmetric system \eqref{eq:P_evol}--\eqref{eq:Q_evol}, arising from initial data $(P_D, Q_D, \dot{P}_D, \dot{Q}_D)$ obeying \eqref{eq:data_linfty}--\eqref{eq:data_energy_boundedness_Q} and $0 < t_0 \leq t_*$ as in Theorem~\ref{thm:global}. 

    Let $\gamma: (0, t_0] \to (0, t_0] \times \mathbb{S}^1$ be a timelike curve parameterized by the $t$-coordinate, and let $\mathscr{P}_{\gamma}(t) = - t \partial_t P(\gamma(t))$ and $\mathscr{Q}_{\gamma}(t) = e^P t \partial_{\theta} Q (\gamma(t))$. Then there exist error terms $\mathscr{E}_{\mathscr{P}}(t), \mathscr{E}_{\mathscr{Q}}(t)$ depending on $\gamma$ but with $|\mathscr{E}_{\mathscr{P}}(t)|, |\mathscr{E}_{\mathscr{Q}}(t)| \leq t^{\upgamma'}$ uniformly in the choice of $\gamma$, such that
    \begin{equation} \label{eq:bounce_ode}
        t \frac{d}{dt} \mathscr{P}_{\gamma} = \mathscr{Q}_{\gamma}^2 + \mathscr{E}_{\mathscr{P}}, 
        \qquad t \frac{d}{dt} \mathscr{Q}_{\gamma} = (1 - \mathscr{P}_{\gamma}) \mathscr{Q}_{\gamma} + \mathscr{E}_{\mathscr{Q}}.
    \end{equation}

    Since $\gamma$ is timelike there exists $\theta_0 \in \mathbb{S}^1$ such that $\gamma(t) \to (0, \theta_0)$ as $t \to 0$. Further, $\mathscr{Q}_{\gamma}(t)$ converges to $0$ as $t \to 0$, while $\mathscr{P}_{\gamma}(t)$ converges to $\mathscr{P}_{\gamma, \infty} = V(\theta_0)$. Finally, there exists $C = C(\upgamma, \upgamma', \upgamma'', \upzeta) > 0$ such that:
    \begin{enumerate}[(i)]
        \item
            If the $\lim_{t \to 0} \partial_{\theta} Q( \gamma(t))$ does not converge to $0$, then necessarily $V(\theta_0) = \mathscr{P}_{\gamma, \infty} \leq 1$ and moreover
            \begin{equation} \label{eq:P_limit}
                \left| \mathscr{P}_{\gamma, \infty} - \min \{ \mathscr{P}_{\gamma}(t_0), 2 - \mathscr{P}_{\gamma}(t_0) \} \right| \leq C \cdot (t_0^{\upgamma'} + \mathscr{Q}^2_{\gamma}(t_0))^{1/2} .
            \end{equation}
        \item
            Let $\upgamma > \frac{1}{2}$. By Theorem~\ref{thm:stab_rough}, $\partial_{\theta} q(\theta_0)$ exists, and $\lim_{t \to 0} \partial_{\theta} Q(\gamma(t)) \to \partial_{\theta} q(\theta_0)$. If $\partial_{\theta} q(\theta_0) = 0$, then one instead has:
            \begin{equation} \label{eq:P_limit2}
                \left| \mathscr{P}_{\gamma, \infty} - \mathscr{P}_{\gamma}(t_0) \right|
                \leq C \cdot t_0^{2\upgamma' - 1}.
            \end{equation}
            where $\frac{1}{2} < \upgamma' < \upgamma$ is chosen appropriately.
    \end{enumerate}
\end{theorem}

A corollary of Theorem~\ref{thm:bounce} is the following stability / instability statement regarding \emph{nonpolarized} perturbations of a class of \emph{polarized} Gowdy solutions. 

\begin{corollary}[Stability / Instability] \label{cor:bounce}
    Let $P(t, \theta)$ be a \underline{smooth} solution to the polarized Gowdy equation \eqref{eq:P_evol_polarized}, arising from initial data given by $(P_D, \dot{P}_D)$ at $t = t_1 > 0$. By Theorem~\ref{thm:asymp_smooth}, there exists smooth $V, \Phi: \mathbb{S}^1 \to \R$ such that $P(t, \theta) = - V(\theta) \log t + \Phi(\theta) + o(1)$ as $t \to 0$.

    Suppose that $0 < V(\theta) < 2$ for all $\theta \in \mathbb{S}^1$. Then there exists $N \in \mathbb{N}$ such that for (possibly non-polarized) perturbations of this data with $\| (\tilde{P}_D, \tilde{Q}_D, \tilde{\dot{P}}_D, \tilde{\dot{Q}}_D) - (P_D, 0, \dot{P}_D, 0) \|_{(H^{N+1})^2 \times (H^N)^2} \leq \varepsilon$, then for $\varepsilon$ sufficiently small the perturbed solution still has $- t \partial_t \tilde{P}(t, \theta) \to \tilde{V}(\theta)$ pointwise as $t \to 0$, for $\tilde{V}$ satisfying $0 < \tilde{V}(\theta) < 2$.

    Next, suppose further that $\frac{1}{2} < V(\theta) < \frac{3}{2}$ for the original polarized solution. Then for perturbations as above, we moreover have that $\partial_{\theta} \tilde{Q} (t, \theta) \to \partial_{\theta} \tilde{q}(\theta)$ uniformly for some $C^1$ function $\tilde{q}: \mathbb{S}^1 \to \R$. Furthermore, for any $\tilde{\varepsilon} > 0$, $\varepsilon$ may be chosen small enough depending on $\tilde{\varepsilon}$ such that:
    \begin{enumerate}[(i)]
        \item $|\tilde{V}(\theta) - \min \{ V(\theta), 2 - V(\theta) \} | \leq \tilde{\varepsilon}$ if $\partial_{\theta} \tilde{q} (\theta) \neq 0$, while
            
        \item $|\tilde{V}(\theta) - V(\theta) | \leq \tilde{\varepsilon}$ if $\partial_{\theta} \tilde{q} (\theta) = 0$.
    \end{enumerate}
\end{corollary}

\begin{remark}
    We consider Corollary~\ref{cor:bounce} a stability / instability result in the sense that while the perturbed spacetime retains the spacelike singularity and curvature blowup, in the case where the original $V(\theta)$ has $1 < V(\theta_0) < 2$ for some $\theta_0 \in \mathbb{S}^1$ a generic unpolarized perturbation will have a corresponding $\hat{V}(\theta_0)$ with $\hat{V}(\theta_0) \approx 2 - V(\theta_0)$, which is ``far away'' from $V(\theta_0)$.

    Our result applies in particular when the background unperturbed spacetime is an exact Kasner spacetime with $P(\theta) = - V \log t$ where $0 < V < 2$ (and Kasner exponents given by \eqref{eq:kasnerlike_intro}). When $1 < V < 2$, this gives a precise instability mechanism for a certain range of Kasner exponents. Note our methods do not allow us to access this mechanism outside this range, especially regarding spatially inhomogeneous perturbations.

    Our instability is triggered when $\partial_{\theta} q(\theta) \neq 0$, while if $\partial_{\theta} q (\theta) = 0$ and $\frac{1}{2} < V(\theta) < \frac{3}{2}$ our result suggests that the instability is suppressed. Note that in light of \cite{RingstromGowdySCC2} it is true that for an open and dense subset of perturbations, the instability is indeed triggered at all but finitely many $\theta \in \mathbb{S}^1$; at the remaining points we leave open the possibility of spikes.
\end{remark}

%% file: energy.tex

\section{Bootstrap assumptions, energies and interpolation lemmas} \label{sec:energy}

\subsection{The \texorpdfstring{$L^{\infty}$}{L\_infty} bootstrap assumptions} \label{sub:bootstrap}

As explained in Section~\ref{subsub:param}, let $C_* = C_*(\upgamma, \upzeta) > 0$ be a large real number to be chosen later in the argument. We often make reference to the following four low order $L^{\infty}$ bootstrap assumptions:
\begin{gather}
    \label{eq:bootstrap_pq} \tag{B1}
    |t \partial_t P| \leq C_*, \quad |e^P t \partial_{\theta} Q| \leq C_*, \\[0.8em]
    \label{eq:bootstrap_other} \tag{B2}
    |e^P t \partial_{t} Q|, \, |e^{-P}| \leq C_* t^{\upgamma'}, \\[0.8em]
    \label{eq:bootstrap_dth} \tag{B3}
    |\partial_{\theta} P|, \, |e^P t \partial_{\theta}^2 Q| \leq C_* t^{-(1 - \upgamma)}, \\[0.8em]
    \label{eq:bootstrap_dth2} \tag{B4}
    |t \partial_t \partial_{\theta} P |, \, |e^P t \partial_t \partial_{\theta} Q| \leq C_* t^{-(1 - \upgamma)}.
\end{gather}

\subsection{Energies} \label{sub:energy}

\begin{definition} \label{def:energy}
    Let $0 \leq K \leq N$. Define the following $K$th order energies at fixed $t$:
    \begin{gather}
        \mathcal{E}^{(K)}_{P}(t) \coloneqq \frac{1}{2} \int_{\mathbb{S}^1} \left( (t \partial_t \partial_{\theta}^K P)^2 + t^2 (\partial_{\theta}^{K+1} P)^2 + (\partial_{\theta}^K P)^2 \right) d\theta, \\[0.8em]
        \mathcal{E}^{(K)}_{Q}(t) \coloneqq \frac{1}{2} \int_{\mathbb{S}^1} e^{2P} \left( (t \partial_t \partial_{\theta}^K Q )^2 + t^2 (\partial_{\theta}^{K+1} Q)^2 \right) \, d\theta, \\[0.8em]
        \mathcal{E}^{(K)}(t) \coloneqq \mathcal{E}^{(K)}_{P}(t) + \mathcal{E}^{(K)}_{Q}(t).
    \end{gather}
    To recover asymptotics for $Q(t, \theta)$ without the $e^P$ weight we also make use of the following energy:
    \begin{equation}
        \mathcal{E}^{(K)}_{Q,u}(t) \coloneqq \frac{1}{2} \int_{\mathbb{S}^1} \left( (t \partial_t \partial_{\theta}^K Q)^2 + t^2 (\partial_{\theta}^{K+1} Q)^2 + (\partial_{\theta}^K Q)^2 \right) d\theta.
    \end{equation}
\end{definition}

\subsection{Sobolev--type inequalities}

\begin{lemma}[Sobolev interpolation inequality] \label{lem:interpolation}
    Let $N, K$ be integers with $0 \leq N < K$, and let $f: \mathbb{S}^1 \to \R$ be such that $\partial_x^K f \in L^2(\mathbb{S}^1)$. Then the following $L^{\infty}$--$L^2$ interpolation inequality holds:
    \begin{equation}
        \| \partial_{x}^N f \|_{L^{\infty}(\mathbb{S}^1)} \lesssim_{N, K} \| f \|_{L^{\infty}(\mathbb{S}^1)}^{1 - \alpha} \| \partial_x^K f \|_{L^2(\mathbb{S}^1)}^{\alpha},
        \qquad
        \text{ where } \alpha = \frac{N}{K - \frac{1}{2}}.
    \end{equation}
\end{lemma}

\begin{proof}
    This is standard, see for instance \cite[Lecture II]{NirenbergElliptic}.
\end{proof}

\begin{lemma}[Sobolev product inequality] \label{lem:weightedl2}
    Let $M, N \geq 0$ be integers and let $K = M + N$. Then for $f, g$ sufficiently regular one has
    \begin{equation}
        \| \partial_x^M f \, \partial_x^N g \|_{L^2(\mathbb{S}^1)}
        \lesssim_{M, N} 
        \| f \|_{L^{\infty}(\mathbb{S}^1) } \| \partial_x^K g \|_{L^2(\mathbb{S}^1)}
        + 
        \| g \|_{L^{\infty}(\mathbb{S}^1) } \| \partial_x^K f \|_{L^2(\mathbb{S}^1)}.
    \end{equation}
\end{lemma}

\begin{proof}
    See Appendix B of our companion paper \cite{MeSurfaceSymPaper}; in that article we also introduce a weight function $w: \mathbb{S}^1 \to \R_{> 0}$, which may simply be set identically equal to $1$ here.
\end{proof}

%% file: energy_estimates.tex

\section{The energy estimate hierarchy} \label{sec:l2}

In this section, we derive energy estimates for $\mathcal{E}^{(K)}(t)$, at orders $0 \leq K \leq N$, where $N = N(\upgamma)$ is chosen sufficiently large, by commuting the equations \eqref{eq:P_evol}--\eqref{eq:Q_evol} with up to $K$ $\partial_\theta$-derivatives. See Section~\ref{sub:intro_proof} for an introduction to the main ideas in our hierarchy of energy estimates.

To handle ``error'' terms in the hierarchy (where the precise coefficients arising in the commuted equations are not crucial), it is useful to introduce the following schematic notation: expressions such as
\[
    \sum_{k_p + k_1 + \ldots + k_i = K} \partial_{\theta}^{k_p} f * \partial_{\theta}^{k_1} g * \cdots * \partial_{\theta}^{k_i} g
\]
will represent some linear combination of products of the form $\partial_{\theta}^{k_p} f \cdot \partial_{\theta}^{k_1} g \cdots \partial_{\theta}^{k_i} g$ such that $i \geq 1$, $k_p \geq 1$ and $k_j \geq 1$ for all $1 \leq j \leq i$ and $k_p + k_1 + \ldots + k_i = K$. We emphasize that \underline{unless explicitly stated otherwise}, in these schematic sums $i$ and $j$ will be \underline{positive} integers, as are the indices $k_p, k_1, k_i$ etc. 

In the event that any index e.g.~$k_p$ is allowed to be $0$, this will be explicitly stated, and similarly if there are further constraints on any index.

\subsection{Energy estimates for \texorpdfstring{$P$}{P}} \label{sub:l2_P}

\begin{proposition} \label{prop:P_energy}
    Let $(P, Q)$ be a solution to the Gowdy symmetric system \eqref{eq:P_evol}--\eqref{eq:Q_evol} obeying the bootstrap assumptions \eqref{eq:bootstrap_pq}--\eqref{eq:bootstrap_dth2}. Then there exists a constant $A_*$ depending only on $C_*$, as well as a constant $C^{(K)}$ depending on $C_*$ and the regularity index $K \in \{0, 1, \ldots, N \}$, such that
    \begin{equation} \label{eq:energy_der_P}
        \left| t \frac{d}{dt} \mathcal{E}_{P}^{(K)} (t) \right| \leq
        2 A_* \mathcal{E}^{(K)}(t) + \sum_{k=0}^{K-1} C^{(K)} t^{- 2(1 - \upgamma) (K - k)} \mathcal{E}^{(k)}(t),
    \end{equation}
    where it is understood that the final term is absent if $K=0$.
\end{proposition}

\begin{proof}
    For any $K \geq 1$, commuting the wave equation for $P$ \eqref{eq:P_evol} with $\partial_{\theta}^K$ yields
    \begin{align}
        \addtocounter{equation}{1}
        (t \partial_t)^2 \partial_{\theta}^K P - t^2 \partial_{\theta}^2 \partial_{\theta}^K P &= \nonumber \\[0.6em] 
        &\mkern-48mu 2 \, \partial_{\theta}^K P \, e^{2P} \left( (t \partial_t Q)^2 - (t \partial_{\theta} Q)^2 \right) + 2 e^{2P} (t \partial_t Q) (t \partial_t \partial_{\theta}^K Q) + 2 e^{2P} (t \partial_{\theta} Q)(t \partial_{\theta}^{K+1} Q)
        \tag{\theequation a} \label{eq:P_wave_K_a} \\[0.6em]
        &\mkern-36mu + t^2 e^{2P} \sum_{\substack{i \geq 0, k_r < K \\ k_1 + \cdots + k_i + k_r + k_s = K}} \partial_{\theta}^{k_1} P * \cdots * \partial_{\theta}^{k_i} P * \partial_t \partial_{\theta}^{k_r} Q  * \partial_t \partial_{\theta}^{k_s} Q
        \tag{\theequation b} \label{eq:P_wave_K_b} \\[0.6em]
        &\mkern-36mu + t^2 e^{2P} \sum_{\substack{i \geq 0, k_r < K \\ k_1 + \cdots + k_i + k_r + k_s = K}} \partial_{\theta}^{k_1} P * \cdots * \partial_{\theta}^{k_i} P * \partial_{\theta}^{k_r + 1} Q  * \partial_{\theta}^{k_s + 1} Q
        \tag{\theequation c} \label{eq:P_wave_K_c}
    \end{align}
    Note that in the uncommuted case $K = 0$, \eqref{eq:P_wave_K_a} is replaced by $e^{2P} \left( (t \partial_t Q)^2 + (t \partial_{\theta} Q)^2 \right)$, while the terms \eqref{eq:P_wave_K_b}--\eqref{eq:P_wave_K_c} are absent; this case will be simpler in the subsequent estimate.

    The first line \eqref{eq:P_wave_K_a} is the leading order contribution that gives rise to the $2 A_* \mathcal{E}^{(K)}(t)$ on the RHS of \eqref{eq:energy_der_P}. It remains to estimate the remaining lines \eqref{eq:P_wave_K_b} and \eqref{eq:P_wave_K_c} in $L^2$. For this purpose, we will use the product estimate Lemma~\ref{lem:weightedl2}. It will be convenient to replace the schematic expression \eqref{eq:P_wave_K_a} by
    \begin{equation} \tag{\theequation b'} \label{eq:P_wave_K_b'}
        \sum_{\substack{i \geq 0, 0 \leq k_r, k_s < K \\ k_1 + \cdots + k_i + k_r + k_s = K}} \partial_{\theta}^{k_1} P * \cdots * \partial_{\theta}^{k_i} P * \partial_{\theta}^{k_r} (e^P t \partial_t Q)  * \partial_{\theta}^{k_s} (e^P t \partial_t Q),
    \end{equation}
    noting these two schematic expressions are equivalent since expanding \eqref{eq:P_wave_K_b'} simply means that more derivatives can fall on $P$. We can apply Lemma~\ref{lem:weightedl2} with $w = 1$ to this expression; since we can use \eqref{eq:bootstrap_other} to bound $|e^P t \partial_t Q| \leq C_*$ if $k_r, k_s = 0$,  repeated use of this lemma yields:
    \begin{equation} \label{eq:P_wave_K_b'_est}
        \| \text{\eqref{eq:P_wave_K_b'}} \|_{L^2} \lesssim (1 + C_*^2) \sum_{k=0}^{K-1} ( \| \partial_{\theta}^k P \|_{L^2} + \| \partial_{\theta}^k (e^P t \partial_t Q) \|_{L^2} ) \cdot ( \|\partial_{\theta} P \|_{L^{\infty}} + \| \partial_{\theta} (e^P t \partial_t Q) \|_{L^{\infty}} )^{K - k}.
    \end{equation}
    Note in this expression it is critical that the sum \emph{does not include} $K = k$; these top order objects appeared instead in \eqref{eq:P_wave_K_a}. 

    To explain a little further how one arrives at this estimate, note that summands in \eqref{eq:P_wave_K_b'} containing $i+1$ terms in the product (not including undifferentiated copies of $e^P t \partial_t Q$ which are estimated by \eqref{eq:bootstrap_other}), the maximum number of $\partial_{\theta}$ derivatives landing on either $P$ or $e^P t \partial_t Q$ is exactly $K - i$. Then repeated use of Lemma~\ref{lem:weightedl2} allows us to put exactly $K-i$ $\partial_{\theta}$ derivatives on one of these, which we estimate in $L^2$, while the remaining $i$ terms are estimated by either $\| \partial_{\theta} P \|_{L^{\infty}}$ or $\| \partial_{\theta} ( e^P t \partial_t Q ) \|_{L^{\infty}}$. Setting $k = K - i \in \{1, \ldots, K-1\}$, we obtain the desired estimate \eqref{eq:P_wave_K_b'_est}.

    The next step is to represent the right hand side of \eqref{eq:P_wave_K_b'_est} in terms of our energies. One issue is that our energy $\mathcal{E}^{(K)}_Q(t)$ in Definition~\ref{def:energy} controls $\| e^P (t \partial_t \partial_{\theta}^K Q)^2 \|_{L^2}$ as opposed to $\| \partial_{\theta}^K (e^P t \partial_t Q) \|_{L^2}$. However, this will not be a major issue, since by Lemma~\ref{lem:weightedl2app}, which we defer to after this proof, will tell us that assuming \eqref{eq:bootstrap_other},
    \[
        \| \partial_{\theta}^k (e^P t \partial_t Q) \|_{L^2} \lesssim \sum_{j=1}^k ( \| \partial_{\theta}^j P \|_{L^2} + \| e^P (t \partial_t \partial_{\theta}^j Q) \|_{L^2} ) \cdot ( \| \partial_{\theta} P \|_{L^{\infty}} + \| \partial_{\theta} (e^P t \partial_t Q) \|_{L^{\infty}})^{k-j}.
    \]
    Thereby combining this with \eqref{eq:P_wave_K_b'_est} and the expression for $\mathcal{E}^{(K)}(t)$ in Definition~\ref{def:energy} -- and using that \eqref{eq:P_wave_K_b} and \eqref{eq:P_wave_K_b'} are equivalent -- and finally using the remaining bootstrap assumptions \eqref{eq:bootstrap_dth}--\eqref{eq:bootstrap_dth2} to estimate $\| \partial_{\theta} P \|_{L^{\infty}} + \| \partial_{\theta} (e^P t \partial_t Q) \|_{L^{\infty}}$, one therefore deduces that
    \begin{equation} \label{eq:P_wave_K_b_est}
        \| \text{\eqref{eq:P_wave_K_b}} \|_{L^2} \lesssim_{C_*, K} \sum_{k=0}^{K-1} t^{(1 - \upgamma) (K - k)} \sqrt{ \mathcal{E}^{(k)}(t) }.
    \end{equation}

    The estimate for \eqref{eq:P_wave_K_c} is essentially the same; we first rewrite \eqref{eq:P_wave_K_c} in the schematic form
    \begin{equation} \tag{\theequation c'} \label{eq:P_wave_K_c'}
        \sum_{\substack{i \geq 0, 0 \leq k_r, k_s < K \\ k_1 + \cdots + k_i + k_r + k_s = K}} \partial_{\theta}^{k_1} P * \cdots * \partial_{\theta}^{k_i} P * \partial_{\theta}^{k_r} (e^P t \partial_{\theta} Q)  * \partial_{\theta}^{k_s} (e^P t \partial_{\theta} Q),
    \end{equation}
    then use the product estimate Lemma~\ref{lem:weightedl2} along with the forthcoming Lemma~\ref{lem:weightedl2app} to express the $L^2$ norm of this in terms of familiar quantities:
    \begin{equation} \label{eq:P_wave_K_c'_est}
        \| \text{\eqref{eq:P_wave_K_c'}} \| \lesssim \sum_{k=0}^{K-1} ( \| \partial_{\theta}^k P \|_{L^2} + \| e^P ( t \partial_{\theta}^{k+1} Q ) \|_{L^2}) \cdot ( \| \partial_{\theta} P \|_{L^{\infty}} + \| \partial_{\theta} (e^P t \partial_{\theta} Q) \|_{L^{\infty}})^{K-k}
    \end{equation}

    We then combine with the definition of the energies $\mathcal{E}^{(K)}(t)$, and use the bootstrap assumption \eqref{eq:bootstrap_dth} to estimate $\| \partial_{\theta} P \|_{L^{\infty}} + \| \partial_{\theta} (e^P t \partial_{\theta} Q) \|_{L^{\infty}}$; one yields
    \begin{equation} \label{eq:P_wave_K_c_est}
        \| \text{\eqref{eq:P_wave_K_c}} \|_{L^2} \lesssim_{C_*, K} \sum_{k=0}^{K-1} t^{(1 - \upgamma) (K - k)} \sqrt{ \mathcal{E}^{(k)}(t) }.
    \end{equation}

    Now, to conclude, for $0 \leq K \leq N$ we write
    \begin{multline*}
        t \partial_t \left( \frac{1}{2} \left( (t \partial_t \partial_{\theta}^K P)^2 + t^2 (\partial_{\theta}^{K+1} P)^2 + (\partial_{\theta}^K P)^2 \right) \right) 
        \\[0.6em]
        = (t \partial_t \partial_{\theta}^K P) \left[ (t \partial_t)^2 \partial_{\theta}^K P - t^2 \partial_{\theta}^{K+2} P + \partial_{\theta}^K P \right] + t^2 (\partial_{\theta}^{K+1} P)^2 
        + \partial_{\theta} \left( t \partial_t \partial_{\theta}^K P \cdot t^2 \partial_{\theta}^{K+1} P \right).
    \end{multline*}

    Integrating over $\theta \in \mathbb{S}^1$ so that the last term vanishes, and inserting the commuted wave equation, one has the following identity for the $t \frac{d}{dt}$ derivative of the energy:
    \begin{equation*}
        t \frac{d}{dt} \mathcal{E}^{(K)}_{P}(t) 
        = \int_{\mathbb{S}^1} \left[ (t \partial_t  \partial_{\theta}^{K} P) \left( \text{\eqref{eq:P_wave_K_a}} + \partial_{\theta}^K P \right) + t^2 (\partial_{\theta}^{K+1} P)^2 \right] \, dx 
        + \int_{\mathbb{S}^1} \left[ (t \partial_t \partial_{\theta}^K P) \left( \text{\eqref{eq:P_wave_K_b}} + \text{\eqref{eq:P_wave_K_c}} \right) \right] \, dx.
    \end{equation*}
    From the bootstrap assumptions \eqref{eq:bootstrap_pq}--\eqref{eq:bootstrap_other}, it is clear that $\| \text{\eqref{eq:P_wave_K_a}} \|_{L^2} \leq 10 C_*^2 \sqrt{\mathcal{E}^{(K)}(t)}$. Using Cauchy--Schwarz, the first integral in the above expression can thus be bounded as:
    \[
        \left| \int_{\mathbb{S}^1} \left[ (t \partial_t  \partial_{\theta}^{K} P) \left( \text{\eqref{eq:P_wave_K_a}} + \partial_{\theta}^K P \right) + t^2 (\partial_{\theta}^{K+1} P)^2 \right] \, dx \right|
        \leq (10 C_*^2 + 9) \mathcal{E}^{(K)}(t).
    \]

    On the other hand, the latter integral can be estimated using \eqref{eq:P_wave_K_b_est} and \eqref{eq:P_wave_K_c_est}; one yields that for some $C^{(K)} > 0$,
    \begin{equation*}
        \left | \int_{\mathbb{S}^1} \left[ (t \partial_t \partial_{\theta}^K \phi) \left( \text{\eqref{eq:P_wave_K_b}} + \text{\eqref{eq:P_wave_K_c}} \right) \right] \, dx \right|
        \leq \sqrt{C^{(K)}} \sum_{k=0}^{K-1} t^{-(1 - \upgamma)(K - k)} \sqrt{\mathcal{E}^{(k)}(t)} \cdot \sqrt{\mathcal{E}^{(K)}(t)}.
    \end{equation*}
    Combining these and applying Young's inequality, Proposition~\ref{prop:P_energy} follows, with $2 A_* = 10 C_*^2 + 10$.
\end{proof}

We end this subsection with the promised Lemma~\ref{lem:weightedl2app}, used to relate $\| \partial_{\theta}^k (e^P t \partial_t Q) \|_{L^2}$ to $\| e^P t \partial_t \partial_{\theta}^k Q \|_{L^2}$ in the estimate \eqref{eq:P_wave_K_b'_est}, and $\| \partial_{\theta}^k (e^P t \partial_{\theta} Q) \|_{L^2}$ to $\| e^P t \partial_{\theta}^{k+1} Q \|_{L^2}$ in \eqref{eq:P_wave_K_c'_est}.

\begin{lemma} \label{lem:weightedl2app}
    Let $f: \mathbb{S}^1 \to \R$ be a sufficiently regular function, with $|e^P f| \leq C_*$. Then for $P: \mathbb{S}^1 \to \R$ regular function, and $k \in \N$, one has
    \begin{equation} \label{eq:weightedl2app}
        \| \partial_{\theta}^k (e^P f) \|_{L^2} \lesssim (1 + C_*) \sum_{j=0}^k ( \| \partial_{\theta}^j P \|_{L^2} + \| e^P \partial_{\theta}^j f \|_{L^2} ) \cdot ( \| \partial_{\theta} P \|_{L^{\infty}} + \| \partial_{\theta} (e^P f) \|_{L^{\infty}} )^{k-j}.
    \end{equation}
\end{lemma}

\begin{proof}
    We prove \eqref{eq:weightedl2app} using induction on $k \in \N$. The base case $k = 1$ follows from
    \[
        | \partial_{\theta} (e^P f) | \leq | \partial_{\theta} P | \cdot | e^P f | + | e^P \partial_{\theta} f | \leq C_* \cdot |\partial_{\theta} P| + |e^P \partial_{\theta} f|.
    \]

    Now, let us suppose that \eqref{eq:weightedl2app} holds for $k \leq K \in \N$. To go from $k = K$ to $k = K + 1$, we use our schematic notation to expand $\partial_{\theta}^{K+1} (e^P f)$, and write
    \begin{equation} \label{eq:induction_weighted}
        \partial_{\theta}^{K+1} (e^P f) = \partial_{\theta}^{K+1} P \cdot e^P f + e^P \partial_{\theta}^{K+1} f + \sum_{\substack{k_1 + \cdots + k_j + k_f = K + 1}} \partial_{\theta}^{k_1} P * \cdots * \partial_{\theta}^{k_j} P * \partial_{\theta}^{k_f} (e^P f),
    \end{equation}
    where crucially each of the $k_1, \ldots, k_j, k_f$ will not exceed $K$. Therefore, one may apply the product estimate Lemma~\ref{lem:weightedl2app} to find
    \begin{multline*}
        \left \| \sum_{\substack{k_1 + \cdots + k_j + k_f = K + 1}} \partial_{\theta}^{k_1} P * \cdots * \partial_{\theta}^{k_j} P * \partial_{\theta}^{k_f} (e^P f) \right \|_{L^2} 
        \\[0.4em] \lesssim \sum_{j=0}^{K} ( \| \partial_{\theta}^j P \|_{L^2} + \| \partial_{\theta}^j (e^P f) \|_{L^2} ) \cdot ( \| \partial_{\theta} P \|_{L^{\infty}} + \| \partial_{\theta} (e^P f) \|_{L^{\infty}} )^{K + 1 - j}.
    \end{multline*}

    We now apply the induction hypothesis to change $\| \partial_{\theta}^j (e^P f) \|_{L^2}$ to $\| e^P \partial_{\theta}^j P \|_{L^2}$. We thus deduce
    \begin{multline*}
        \left \| \sum_{\substack{k_1 + \cdots + k_j + k_f = K + 1}} \partial_{\theta}^{k_1} P * \cdots * \partial_{\theta}^{k_j} P * \partial_{\theta}^{k_f} (e^P f) \right \|_{L^2} 
        \\[0.4em] \lesssim (1 + C_*) \sum_{j=0}^{K} ( \| \partial_{\theta}^j P \|_{L^2} + \| e^P \partial_{\theta}^j f \|_{L^2} ) \cdot ( \| \partial_{\theta} P \|_{L^{\infty}} + \| \partial_{\theta} (e^P f) \|_{L^{\infty}} )^{K + 1 - j}.
    \end{multline*}
    On the other hand, the first two terms on the right hand side of \eqref{eq:induction_weighted} have $L^2$ norm bounded by $C_* \| \partial_{\theta}^{K+1} P \|_{L^2} + \| e^P \partial_{\theta}^{K+1} f \|_{L^2}$. This proves the estimate \eqref{eq:weightedl2app} for $k = K + 1$ and completes the proof of the lemma.
\end{proof}

\subsection{Energy estimates for \texorpdfstring{$Q$}{Q}}

We now prove the analogous energy estimate for $\mathcal{E}^{(K)}_Q(t)$.

\begin{proposition} \label{prop:Q_energy}
    Let $(P, Q)$ be a solution to the Gowdy symmetric system \eqref{eq:P_evol}--\eqref{eq:Q_evol} obeying the bootstrap assumptions \eqref{eq:bootstrap_pq}--\eqref{eq:bootstrap_dth2}. Then there exists a constant $A_*$ depending only on $C_*$, as well as a constant $C^{(K)}$ depending on $C_*$ and the regularity index $K \in \{0, 1, \ldots, N \}$, such that
    \begin{equation} \label{eq:energy_der_Q}
        \left| t \frac{d}{dt} \mathcal{E}_{Q}^{(K)} (t) \right| \leq
        (2 A_* + 2 C_* t^{\upgamma}) \mathcal{E}^{(K)}(t) + \sum_{k=0}^{K-1} C^{(K)} t^{- 2(1 - \upgamma) (K - k)} \mathcal{E}^{(k)}(t),
    \end{equation}
    where it is understood that the final term is absent if $K=0$.
\end{proposition}

\begin{proof}
    Commuting the $Q$ wave equation \eqref{eq:Q_evol} with $\partial_{\theta}^K$, one yields for $1 \leq K \leq N$:
    \begin{align}
        \addtocounter{equation}{1}
        (t \partial_t)^2 \partial_{\theta}^K Q - t^2 \partial_{\theta}^2 \partial_{\theta}^K Q &= \nonumber \\[0.6em] 
        &\mkern-72mu - 2 t \partial_t \partial_{\theta}^K P \cdot t \partial_t Q - 2 t \partial_t P \cdot t \partial_t \partial_{\theta}^{K+1} Q + 2 t \partial_{\theta}^{K+1} P \cdot t \partial_{\theta} Q + 2 t \partial_t P \cdot t \partial_{\theta}^{K+1} Q 
        \tag{\theequation a} \label{eq:Q_wave_K_a} \\[0.6em]
        &\mkern-36mu + \sum_{\substack{k_p + k_q = K}} \left( t \partial_t \partial_{\theta}^{k_p} P * t \partial_t \partial_{\theta}^{k_q} Q + t \partial_{\theta}^{k_p + 1} P * t \partial_{\theta}^{k_q + 1} Q \right).
        \tag{\theequation b} \label{eq:Q_wave_K_b}
    \end{align}
    The base case $K = 0$ is just the equation \eqref{eq:Q_evol} and will be easier to deal with. Note it again crucial that neither of $k_p, k_q$ is allowed to equal $K$. 

    In light of the $e^{2P}$ weight appearing in $\mathcal{E}^{(K)}_Q(t)$ in Definition~\ref{def:energy}, we will be required to estimate $e^P$ times the right hand side of this equation in $L^2$. For the lower order term \eqref{eq:Q_wave_K_b}, we use a similar method to the estimate for \eqref{eq:P_wave_K_b} in Proposition~\ref{prop:P_energy}. For the first term in \eqref{eq:Q_wave_K_b}, it is convenient to write instead
    \[
        e^P \sum_{k_p + k_q = K} t \partial_t \partial_{\theta}^{k_p} P * t \partial_t \partial_{\theta}^{k_q} Q = \sum_{\substack{i \geq 0 \\ k_1 + \cdots + k_i + k_p + k_q = K}} \partial_{\theta}^{k_1} P * \cdots * \partial_{\theta}^{k_i} P * t \partial_t \partial_{\theta}^{k_p} P * \partial_{\theta}^{k_q} (e^P t \partial_t Q) 
    \]
    Then one may estimate this in $L^2$ in the same way as the term \eqref{eq:P_wave_K_b'}, eventually getting
    \begin{multline*} 
        \left \| e^P \sum_{k_p + k_q = K} t \partial_t \partial_{\theta}^{k_p} P * t \partial_t \partial_{\theta}^{k_q} Q \right \| 
        \\[0.4em] \lesssim_{C_*} \sum_{k=0}^{K-1} ( \| \partial_{\theta}^k P \|_{L^2} +  \| t \partial_t \partial_{\theta}^k P \|_{L^2} + \| e^P t \partial_t \partial_{\theta}^k Q \|_{L^2}) ( \| \partial_{\theta} P \|_{L^{\infty}} + \| t \partial_t \partial_{\theta} P \|_{L^{\infty}} + \| \partial_{\theta} (e^P t \partial_t Q) \|_{L^{\infty}} )^{K - k}
    \end{multline*}

    Thereby using Definition~\ref{def:energy} and the bootstrap assumptions \eqref{eq:bootstrap_pq}--\eqref{eq:bootstrap_dth2}, we get
    \begin{equation*}
        \left \| e^P \sum_{k_p + k_q = K} t \partial_t \partial_{\theta}^{k_p} P * t \partial_t \partial_{\theta}^{k_q} Q \right \| 
        \lesssim_{C_*, K} 
        \sum_{k=0}^{K-1} t^{- (1 - \upgamma) (K - k)} \sqrt{ \mathcal{E}^{(k)}(t) }.
    \end{equation*}
    The same method allows us to estimate the second term in \eqref{eq:Q_wave_K_b} in the same way, thus
    \begin{equation} \label{eq:Q_wave_K_b_est}
        \| e^P \text{\eqref{eq:Q_wave_K_b}} \|_{L^2} \lesssim_{C_*, K}
        \sum_{k=0}^{K-1} t^{- (1 - \upgamma) (K - k)} \sqrt{ \mathcal{E}^{(k)}(t) }.
    \end{equation}

    To conclude, we first write the derivative identity; for $0 \leq K \leq N$:
    \begin{multline*}
        t \partial_t \left( \frac{1}{2} \left( e^{2P} (t \partial_t \partial_{\theta}^K Q)^2 + e^{2P} t^2 (\partial_{\theta}^{K+1} Q)^2 \right) \right) 
        \\[0.6em]
        = e^{P} (t \partial_t \partial_{\theta}^K Q) \, e^P \left[ (t \partial_t)^2 \partial_{\theta}^K Q - t^2 \partial_{\theta}^{K+2} Q \right] 
        + \partial_{\theta} \left( e^{2P} t \partial_t \partial_{\theta}^K Q \cdot t^2 \partial_{\theta}^{K+1} Q \right)
        \\[0.6em]
        + t \partial_t P \cdot e^{2P} (t \partial_t \partial_{\theta}^K Q)^2 + (1 + t \partial_t P) e^{2P} t^2 (\partial_{\theta}^{K+1} Q)^2 - 2 t \partial_{\theta} P \cdot e^{2P} t \partial_t \partial_{\theta}^K Q \cdot t \partial_{\theta}^{K+1} Q.
    \end{multline*}

    Next, integrating this over $\theta \in \mathbb{S}^1$, one deduces that
    \begin{multline*}
        t \frac{d}{dt} \mathcal{E}^{(K)}_{Q}(t) 
        = \int_{\mathbb{S}^1} \left[ e^P t \partial_t  \partial_{\theta}^{K} Q \left( e^P \text{\eqref{eq:Q_wave_K_a}} + t \partial_t P \cdot e^P t \partial_t \partial_{\theta}^K Q \right) + (1 + t \partial_t P) e^{2P} t^2 (\partial_{\theta}^{K+1} Q)^2 \right] \, dx  \\[0.6em]
        - \int_{\mathbb{S}^1} 2 t \partial_{\theta} P \cdot e^P t \partial_t \partial_{\theta}^K Q \cdot e^P t \partial_{\theta}^{K+1} Q \, dx
        + \int_{\mathbb{S}^1} \left[ (t \partial_t \partial_{\theta}^K Q) \left( e^P \text{\eqref{eq:Q_wave_K_b}} \right) \right] \, dx.
    \end{multline*}
    Finally, using the expression for \eqref{eq:Q_wave_K_a}, the bootstrap assumptions \eqref{eq:bootstrap_pq} and \eqref{eq:bootstrap_other}, and the expression for $\mathcal{E}^{(K)}(t)$ in Definition~\ref{def:energy}, one may check that the first line of the right hand side is bounded by say $(10 C_*^2 + 9) \mathcal{E}^{(K)}(t)$. 

    On the other hand, by using the bootstrap assumption \eqref{eq:bootstrap_dth} for the second integral and \eqref{eq:Q_wave_K_b_est} for the third, one yields:
    \begin{gather*}
        \left| \int_{\mathbb{S}^1} 2 t \partial_{\theta} P \cdot e^P t \partial_t \partial_{\theta}^K Q \cdot e^P t \partial_{\theta}^{K+1} Q \, dx \right| \leq 2 C_* t^{\upgamma} \mathcal{E}^{(K)}(t), \\[0.6em]
        \left| \int_{\mathbb{S}^1} \left[ (t \partial_t \partial_{\theta}^K Q) \left( e^P \text{\eqref{eq:Q_wave_K_b}} \right) \right] \, dx \right| \lesssim \sqrt{ \mathcal{E}^{(K)}(t) } \cdot \sum_{k=1}^{K-1} t^{(1 - \upgamma)(K - k)} \sqrt{ \mathcal{E}^{(k)}(t) }.
    \end{gather*}
    Therefore applying Young's inequality, one deduces \eqref{eq:energy_der_Q} for $2 A_* = 10(C_*^2 + 1)$.
\end{proof}

\subsection{The energy hierarchy} \label{sub:energy_induction}

We now use Propositions~\ref{prop:P_energy}, \ref{prop:Q_energy} together with the initial data assumption \eqref{eq:data_energy_boundedness_P}--\eqref{eq:data_energy_boundedness_Q}, to show that the total energy of order $K$, $\mathcal{E}^{(K)}(t)$, grows at most polynomially in $t$ as $t \downarrow 0$, and moreover that the rate of blow-up depends linearly in $K$. 

\begin{proposition} \label{prop:energy_hierarchy}
    Let $(P, Q)$ be a solution to the Gowdy symmetric system \eqref{eq:P_evol}--\eqref{eq:Q_evol} in the interval $t \in [t_b, t_0]$, such that the solution obeys the bootstrap assumptions \eqref{eq:bootstrap_pq}--\eqref{eq:bootstrap_dth2}. Assuming also the bounds \eqref{eq:data_energy_boundedness_P}--\eqref{eq:data_energy_boundedness_Q} for the initial data, then there exist a constant $A_*$ depending only on $C_*$, as well as constants $C^{(K)}$ depending on $C_*$, $K$ and $\upzeta$, such that for $0 \leq K \leq N$, the total energy $\mathcal{E}^{(K)}(t)$ satisfies the bound: 
    \begin{equation} \label{eq:energy_hierarchy}
        \mathcal{E}^{(K)}(t) \leq C^{(K)} \, t^{- 2 A_* - 2 K (1 - \upgamma)}.
    \end{equation}
\end{proposition}

\begin{proof}
    Combining Propositions~\ref{prop:P_energy} and \ref{prop:Q_energy}, it is straightforward to show that for some $A_* > 0$ depending only on $C_*$ and constants $C^{(K)} > 0$ (we allow $A_*$ to differ from the previous propositions), one has the following derivative estimate:
    \begin{equation*}
        \left| t \frac{d}{dt} \mathcal{E}^{(K)}(t) \right| \leq
        (2 A_* + 2 C_* t^{\upgamma} ) \mathcal{E}^{(K)}(t) +
        C^{(K)} \sum_{k = 0}^{K-1} r^{- (1 - \upgamma) (K-k)} \mathcal{E}^{(k)}(t).
    \end{equation*}

    As we integrate ``backwards'' i.e.~towards $r = 0$, the derivative estimate we actually use is the following:
    \begin{equation*}
        t \frac{d}{dt} \mathcal{E}^{(K)}(t) \geq
        - (2 A_* + 2 C_* t^{\upgamma}) \mathcal{E}^{(K)}(t) - 
        C^{(K)} \sum_{k = 0}^{K-1} r^{- \upgamma(K-k)} \mathcal{E}^{(k)}(t).
    \end{equation*}
    In fact, using the integrating factor $t^{2 A_*}$, which is \emph{crucially independent of $K$}, we write:
    \begin{equation} \label{eq:energy_der}
        t \frac{d}{dt} \left( t^{2 A_*} \mathcal{E}^{(K)}(t) \right) \geq
        - 2 C_* t^{\upgamma} \left( t^{2A_*} \mathcal{E}^{(K)}(t) \right) -
        C^{(K)} \sum_{k = 0}^{K-1} t^{- 2 (1 - \upgamma) (K-k)} \left( t^{2 A_*} \mathcal{E}^{(k)}(t) \right).
    \end{equation}
     
    We will now use \eqref{eq:energy_der} and induction on $K \in \{0, \ldots, N \}$ that
    \begin{equation} \label{eq:energy_hierarchy_induction}
        t^{2 A_*} \mathcal{E}^{(K)}(t) \lesssim t^{- 2 (1 - \upgamma) K},
    \end{equation}
    where the implied constant is now allowed to depend on $C_*$, $K$ and $\upzeta$. This is equivalent to \eqref{eq:energy_hierarchy}. Note that the dependence on $\upzeta$ comes from the fact that the initial data bound \eqref{eq:data_energy_boundedness_P}--\eqref{eq:data_energy_boundedness_Q} implies that
    \begin{equation} \label{eq:energy_hierarchy_init}
        \sum_{k=0}^N t_0^{2 A_*} \mathcal{E}^{(k)}(t_0) \leq \upzeta.
    \end{equation}

    For the base case $K = 0$, one simply applies Gr\"onwall's inequality to \eqref{eq:energy_der} for $K = 0$; then for $t \in [t_b, t_0]$
    \[
        t^{2A_*} \mathcal{E}^{(0)}(t) \leq \exp(F(t_0, t)) \cdot t_0^{2 A_*} \mathcal{E}^{(0)}(t_0), \quad \text{ where }
        \quad F(s_a, s_b) = \int^{s_a}_{s_b} 2 C_* \tilde{t}^{\upgamma} \frac{d \tilde{t}} {\tilde{t}}.
    \]
    Since $F(s_a, s_b)$ is uniformly bounded for $s_a, s_b \in [t_b, t_0]$, it follows from the initial data bound \eqref{eq:energy_hierarchy_init} that \eqref{eq:energy_hierarchy_induction} holds for $K = 0$.

    Moving onto the induction step, assume that \eqref{eq:energy_hierarchy_induction} holds for $0 \leq K < \bar{K} \leq N$; we wish to prove it also holds for $K = \bar{K}$. Applying Gr\"onwall's inequality to \eqref{eq:energy_der} for $K = \bar{K}$, we have that
    \[
        t^{2 A_*} \mathcal{E}^{(\bar{K})}(t) \leq \exp(F(t_0, t)) \cdot t_0^{2 A_*} \mathcal{E}^{(\bar{K})}(t_0) + \int^{t_0}_t \exp( F(\tilde{t}, t)) C^{(\bar{K})} \sum_{k=0}^{\bar{K}-1} \tilde{t}^{-2(1 - \upgamma)(\bar{K}-k)} \tilde{t}^{2 A_*} \mathcal{E}^{(k)}(\tilde{t}) \frac{d \tilde{t}}{\tilde{t}}.
    \]
    It therefore follows from the initial data bound \eqref{eq:energy_hierarchy_init} and the inductive hypothesis for $\tilde{t}^{2 A_*} \mathcal{E}^{(k)}(\tilde{t})$ that
    \[
        t^{2 A_*} \mathcal{E}^{(\bar{K})}(t) \lesssim \upzeta + \int_t^{t_0} \sum_{k=0}^{\bar{K}-1} \tilde{t}^{-2(1 - \upgamma)(\bar{K} - k)} \cdot \tilde{t}^{- 2 \upgamma k} \frac{d \tilde{t}}{\tilde{t}} \lesssim t^{-2 (1 -\upgamma) \bar{K}}
    \]
    as required. This completes the proof of the proposition.
\end{proof}

\subsection{An auxiliary energy estimate} \label{sub:energy_aux}

In order to recover precise asymptotics for $Q(t, \theta)$ as $t \to 0$ we also make use of the following energy estimate for the energy $\mathcal{E}_{Q,u}^{(K)}(t)$, see Definition~\ref{def:energy}.

\begin{proposition} \label{prop:energy_aux}
    Let $(P, Q)$ be a solution to the Gowdy symmetric system \eqref{eq:P_evol}--\eqref{eq:Q_evol} obeying the bootstrap assumptions \eqref{eq:bootstrap_pq}--\eqref{eq:bootstrap_dth2}. Then there exists a constant $A_*$ depending only on $C_*$, as well as a constant $C^{(K)}$ depending on $C_*$ and the regularity index $K \in \{0, 1, \ldots, N\}$ such that
    \begin{equation} \label{eq:energy_aux_der}
        \left| t \frac{d}{dt} \mathcal{E}_{Q, u}^{(K)}(t) \right| \leq 2 A_* \mathcal{E}^{(K)}_{Q, u}(t) 
        + C^{(K)} \sqrt{2 \mathcal{E}^{(K)}_{Q, u}(t)} \cdot \left( t^{-(1 - \upgamma)} \sqrt{2 \mathcal{E}^{(K-1)}_{Q,u}(t)} + t^{\upgamma'} \sqrt{2 \mathcal{E}^{(K)}(t)} \right).
    \end{equation}
    
    With initial data as in Proposition~\ref{prop:energy_hierarchy}, one can moreover show that for some $C^{(K)} > 0$,
    \begin{equation} \label{eq:energy_aux_bound}
        \mathcal{E}_{Q, u}^{(K)}(t) \leq C^{(K)} t^{-2 A_* - 2K(1 - \upgamma)}.
    \end{equation}
    Here $C^{(K)}$ may also depend on the initial data $(Q_D, \dot{Q}_D)$.
\end{proposition}

\begin{proof}
    For the auxiliary estimate we again commute the wave equation \eqref{eq:Q_evol} with $\partial_{\theta}^K$, though we shall group terms in a slightly different way to \eqref{eq:Q_wave_K_a} and \eqref{eq:Q_wave_K_b}. We have:
    \begin{align}
        \addtocounter{equation}{1}
        (t \partial_t)^2 \partial_{\theta}^K Q - t^2 \partial_{\theta}^2 \partial_{\theta}^K Q &= 
        - 2 t \partial_t P \cdot t \partial_t \partial_{\theta}^{K+1} Q + 2 t \partial_{\theta} P \cdot t \partial_{\theta}^{K+1} Q 
        \tag{\theequation a} \label{eq:Qu_wave_K_a} \\[0.6em]
        &\mkern-36mu + \sum_{\substack{k_p + k_q = K \\ 0 \leq k_q < K}} \left( t \partial_t \partial_{\theta}^{k_p} P * t \partial_t \partial_{\theta}^{k_q} Q + t \partial_{\theta}^{k_p + 1} P * t \partial_{\theta}^{k_q + 1} Q \right).
        \tag{\theequation b} \label{eq:Qu_wave_K_b}
    \end{align}
    The difference between this and \eqref{eq:Q_wave_K_a}--\eqref{eq:Q_wave_K_b} is that the top order term containing $P$ (i.e.~$k_p = K$) is now included in \eqref{eq:Qu_wave_K_b}. We shall now estimate \eqref{eq:Qu_wave_K_b} in $L^2$ without the $e^P$ weight.

    Using Lemma~\ref{lem:weightedl2} and the fact that $1 \leq k_p \leq K$, one finds
    \begin{multline*}
        \| \text{\eqref{eq:Qu_wave_K_b}} \|_{L^2} \lesssim \| t \partial_t \partial_{\theta}^K P \|_{L^2} \| t \partial_t Q \|_{L^{\infty}} + \| t \partial_t \partial_{\theta} P \|_{L^{\infty}} \| t \partial_t \partial_{\theta}^{K-1} Q \|_{L^2} \\[0.4em]
        + \| t \partial_{\theta}^{K+1} P \|_{L^2} \| t \partial_{\theta} Q \|_{L^{\infty}} + \| t \partial_{\theta}^2 P \|_{L^{\infty}} \| t \partial_{\theta}^{K} Q \|_{L^2}.
    \end{multline*}
    For the first and third terms on the right hand side, the $L^2$ norm is controlled by $\mathcal{E}^{(K)}_P(t)$, while the $L^{\infty}$ norm is controlled via the bootstrap assumptions \eqref{eq:bootstrap_pq} and \eqref{eq:bootstrap_other}, which together give $\|t \partial_t Q \|_{L^{\infty}}, \|t \partial_{\theta} Q \|_{L^{\infty}} \leq C_*^2 t^{\upgamma}$. For the second and fourth terms, the $L^2$ norm is controlled by $\mathcal{E}^{(K - 1)}_{Q,u}(t)$, while the $L^{\infty}$ norms are controlled using \eqref{eq:bootstrap_dth}--\eqref{eq:bootstrap_dth2}.

    Combining all of these will yield that
    \begin{equation} \label{eq:Qu_wave_K_b_est}
        \| \text{\eqref{eq:Qu_wave_K_b}} \| \lesssim t^{\upgamma'} \sqrt{ \mathcal{E}^{(K)}_P(t)} + t^{-(1 - \upgamma)} \sqrt{ \mathcal{E}^{(K-1)}_{Q,u}(t)}.
    \end{equation}
    Continuing, we write down the derivative identity:
    \begin{multline*}
        t \partial_t \left( \frac{1}{2} \left( (t \partial_t \partial_{\theta}^K Q)^2 + t^2 (\partial_{\theta}^{K+1} Q)^2 + (\partial_{\theta}^K Q)^2 \right) \right) 
        \\[0.6em]
        = (t \partial_t \partial_{\theta}^K Q) \left[ (t \partial_t)^2 \partial_{\theta}^K Q - t^2 \partial_{\theta}^{K+2} Q + \partial_{\theta}^K Q \right] + t^2 (\partial_{\theta}^{K+1} Q)^2 
        + \partial_{\theta} \left( t \partial_t \partial_{\theta}^K Q \cdot t^2 \partial_{\theta}^{K+1} Q \right).
    \end{multline*}

    Integrating over $\theta \in \mathbb{S}^1$, one thus yields that
    \begin{equation*}
        t \frac{d}{dt} \mathcal{E}^{(K)}_{Q}(t) 
        = \int_{\mathbb{S}^1} \left[ (t \partial_t  \partial_{\theta}^{K} Q) \left( \text{\eqref{eq:Qu_wave_K_a}} + \partial_{\theta}^K Q \right) + t^2 (\partial_{\theta}^{K+1} Q)^2 \right] \, dx 
        + \int_{\mathbb{S}^1} \left[ (t \partial_t \partial_{\theta}^K Q) \cdot \text{\eqref{eq:Qu_wave_K_b}} \right] \, dx.
    \end{equation*}
    Using the bootstrap assumptions \eqref{eq:bootstrap_pq} and \eqref{eq:bootstrap_dth} to deal with \eqref{eq:Qu_wave_K_a}, and using \eqref{eq:Qu_wave_K_b_est} to deal with \eqref{eq:Qu_wave_K_b}, one therefore deduces the estimate \eqref{eq:energy_aux_der} in a similar way to say Proposition~\ref{prop:P_energy}.

    Inserting the bound \eqref{eq:energy_hierarchy} into \eqref{eq:energy_aux_der}, we find that:
    \[
        t \frac{d}{dt} \mathcal{E}^{(K)}_{Q, u}(t) \geq - 2 A_* \mathcal{E}^{(K)}_{Q, u}(t) - 2 C^{(K)} \sqrt{\mathcal{E}^{(K)}_{Q, u}(t)} \cdot \left( t^{-(1 - \upgamma)} \sqrt{\mathcal{E}_{Q, u}^{(K-1)}(t)} + t^{\upgamma' - A_* - K(1 - \upgamma)} \right).
    \]
    Equivalently,
    \[
        t \frac{d}{dt} \sqrt{t^{2A_*} \mathcal{E}^{(K)}_{Q, u}(t)} \geq - C^{(K)} \left( t^{-(1- \upgamma)} \sqrt{t^{2A_*} \mathcal{E}_{Q, u}^{(K-1)} (t) } + t^{-\upgamma' - A_* - K(1 - \upgamma)} \right).
    \]
    The integrated estimate \eqref{eq:energy_aux_bound} then follows easily using induction on $K \in \{0, \ldots, N \}$.
\end{proof}

%% file: interpolation.tex

\section{Derivation of ODEs} \label{sec:interp}

\subsection{Low order interpolation estimates} 

The goal in this section is to use the energy estimates of Proposition~\ref{prop:energy_hierarchy} together with the Sobolev interpolation in Lemma~\ref{lem:interpolation} to provide $L^{\infty}$ bounds for low order $\partial_{\theta}$ derivatives of $P$ and $Q$, allowing us to treat \eqref{eq:P_evol}--\eqref{eq:Q_evol} as ODEs for certain quantities without worrying about losing derivatives.

\begin{lemma} \label{lem:low_order_linfty}
    Let $(P, Q)$ be as in Proposition~\ref{prop:energy_hierarchy}. Given any $0 < \upgamma' < \upgamma$, for $N$ chosen sufficiently large there exists a family of constants $\updelta = \updelta ( C_*, \upzeta, t_*) > 0$ with $\updelta \downarrow 0$ as $t_* \downarrow 0$, such that for any $0 \leq k \leq 3$ and $t \in [t_b, t_0]$ one has
    \begin{equation} \label{eq:low_order_linfty}
        \| \partial_{\theta}^k P (t, \cdot) \|_{L^{\infty}} + \| t \partial_t \partial_{\theta}^k P (t, \cdot) \|_{L^{\infty}} + \| e^P t \partial_t \partial_{\theta}^k Q (t, \cdot) \|_{L^{\infty}} + \| e^P t \partial_{\theta}^{k+1} Q \|_{L^{\infty}} \leq \updelta t^{-k (1 - \upgamma')}.
    \end{equation}
\end{lemma}

\begin{proof}
    We use the energy estimate \eqref{eq:energy_hierarchy} for $K = N$, alongside Lemma~\ref{lem:weightedl2app} applied to $f = t \partial_t Q$ and $f = t \partial_{\theta} Q$, to derive the following top order $L^2$ estimate:
    \begin{equation} \label{eq:l2_top}
        \| \partial_{\theta}^N P \|_{L^2}^2 + \| t \partial_t \partial_{\theta}^N P \|_{L^2}^2 + \| \partial_{\theta}^N (e^P t \partial_t Q) \|_{L^2}^2 + \| \partial_{\theta}^N (e^P t \partial_{\theta} Q) \|_{L^2}^2 \leq 2 C^{(N)}  t^{- 2 A_* - 2 N (1 - \upgamma)}.
    \end{equation}
    Note that while $C^{(N)}$ depends on $N$, the number $A_*$ does not, and we later choose $N$ depending on $A_*$.

    We now interpolate between \eqref{eq:l2_top} and the low-order $L^{\infty}$ bootstrap assumptions \eqref{eq:bootstrap_pq} and \eqref{eq:bootstrap_other}. Applying Lemma~\ref{lem:interpolation}, with $f \in \{P, t \partial_t P, e^P t \partial_t Q, e^P t \partial_{\theta} Q \}$ and $0 \leq k \leq 3$ one finds:
    \[
        \| \partial_{\theta}^k f (t, \cdot ) \|_{L^{\infty}} \lesssim_{N} \| f (t, \cdot) \|_{L^{\infty}}^{1-\alpha} \, \| \partial_{\theta}^N f (t, \cdot) \|_{L^2}^{\alpha}, \qquad \text{ where } \alpha = \frac{k}{N - \frac{1}{2}}.
    \]
    Inserting the bound \eqref{eq:l2_top} and the bootstrap assumptions \eqref{eq:bootstrap_pq} and \eqref{eq:bootstrap_other}, one finds
    \[
        \| \partial_{\theta}^k f (t, \cdot) \|_{L^{\infty}} \lesssim_{N, C_*} \left( t^{- 2 A_* - 2 N (1 - \upgamma) } \right)^{\frac{\alpha}{2}} = t^{- k (1 - \upgamma)} \cdot \left( t^{- 2 A_* - (1- \upgamma)} \right)^{\frac{\alpha}{2}}.
    \]

    Note that as $N \to \infty$, $\alpha \to 0$. In particular, for $N$ chosen sufficiently large (depending on $A_*$, $\upgamma$ and $\upgamma' < \upgamma$) one can guarantee that the second term in the product $(t^{-2 A_* - (1- \upgamma)})^{\frac{\alpha}{2}}$ can be bounded by say $t^{-\frac{\upgamma - \upgamma'}{2}}$. Thus for this choice of $N$, we deduce that for $f$ as above,
    \begin{equation} \label{eq:interp}
        \| \partial_{\theta}^k f (t, \cdot) \|_{L^{\infty}} \leq C_{N, C_*} t^{-k(1 - \upgamma) - \frac{\upgamma - \upgamma'}{2}} = C_{N, C_*} t^{-k(1 - \upgamma')} \cdot t^{(k - \frac{1}{2})(\upgamma - \upgamma')}.
    \end{equation}
    
    To go from \eqref{eq:interp} to \eqref{eq:low_order_linfty}, there are two more steps required. The first is that we can expand out $\partial_{\theta}^k (e^P t \partial_t Q)$ and $\partial_{\theta}^k (e^P t \partial_{\theta} Q)$ to yield, for some suitably modified constant $C_{N, C_*}$,
    \begin{equation*} 
        \| \partial_{\theta}^k P (t, \cdot) \|_{L^{\infty}} + \| t \partial_t \partial_{\theta}^k P (t, \cdot) \|_{L^{\infty}} + \| e^P t \partial_t \partial_{\theta}^k Q (t, \cdot) \|_{L^{\infty}} + \| e^P t \partial_{\theta}^{k+1} Q \|_{L^{\infty}} \leq t^{-k (1 - \upgamma')} \cdot (C_{N, C_*} t^{\frac{1}{2}(\upgamma - \upgamma')}).
    \end{equation*}
    Note we used here that $k - \frac{1}{2} \geq \frac{1}{2}$. Then recalling that $t \leq t_0 \leq t_*$, letting $\updelta = C_{N, C_*} \cdot t_*^{\frac{1}{2}(\upgamma - \upgamma')}$ completes the proof of the lemma.
\end{proof}

Using Proposition~\ref{prop:energy_aux}, we also derive interpolated estimates for $Q$, without the $e^P$ weight.

\begin{lemma} \label{lem:low_order_linfty2}
    Let $(P, Q)$ be as in Proposition~\ref{prop:energy_aux}. For $\upgamma, \upgamma', N$ as in Lemma~\ref{lem:low_order_linfty2} there exists a family of constants $\updelta = \updelta ( C_*, \upzeta, t_*) > 0$ with $\updelta \downarrow 0$ as $t_* \downarrow 0$, such that for any $0 \leq k \leq 3$ and $t \in [t_b, t_0]$ one has
    \begin{equation} \label{eq:low_order_linfty2}
        \| \partial_{\theta}^k Q (t, \cdot) \|_{L^{\infty}} + \| t \partial_t \partial_{\theta}^k Q (t, \cdot) \|_{L^{\infty}} \leq \updelta t^{-k (1 - \upgamma')}.
    \end{equation}
\end{lemma}

\begin{proof}
    The proof is identical to that of Lemma~\ref{lem:low_order_linfty}, using instead the energy $\mathcal{E}_{Q,u}^{(N)}(t)$ and the estimate \eqref{eq:energy_aux_bound}.
\end{proof}

\subsection{The main bounce ODEs}

\begin{corollary} \label{cor:ode}
    Let $(P, Q)$ be as in Proposition~\ref{prop:energy_hierarchy}. Then for $N$ chosen sufficiently large there exists a family of constants $\updelta = \updelta(C_*, \upzeta, t_*) > 0$, with $\updelta \downarrow 0$ as $r_* \downarrow 0$, such that for all $(t, x) \in [t_b, r_0] \times \mathbb{S}^1$ and all $a \in [-1, 1]$:
    \begin{gather}
        \left| (t \partial_t + a t \partial_{\theta}) (- t \partial_t P) - \left( (e^P t \partial_{\theta} Q)^2 - (e^P t \partial_t Q)^2 \right) \right| \, (t, x) \leq \updelta t^{\upgamma'}, \label{eq:P_ode_err} \\[0.6em]
        \left| (t \partial_t + a t \partial_{\theta}) (e^P t \partial_{\theta} Q) - (1 + t \partial_t P)(e^P t \partial_{\theta} Q ) \right| \, (t, x) \leq \updelta t^{\upgamma'}, \label{eq:Q_ode_err} \\[0.6em] 
        \left| (t \partial_t + a t \partial_{\theta}) (e^P t \partial_t Q) + (t \partial_t P)(e^P t \partial_t Q) \right| \, (t, x) \leq \updelta t^{\upgamma'}. \label{eq:R_ode_err}
    \end{gather}
\end{corollary}

\begin{proof}
    Each of these will follow straightforwardly from Lemma~\ref{lem:low_order_linfty} and basic manipulation of the equations \eqref{eq:P_evol} and \eqref{eq:Q_evol}. For \eqref{eq:P_ode_err}, from \eqref{eq:P_evol} one has:
    \begin{equation} \label{eq:P_aux}
        (t \partial_t + a t \partial_{\theta}) (- t \partial_t P) - \left( (e^P t \partial_{\theta} Q)^2 - (e^P t \partial_t Q)^2 \right) = - t^2 \partial_{\theta}^2 P - a t^2 \partial_t \partial_{\theta} P.
    \end{equation}
    Then applying Lemma~\ref{lem:low_order_linfty}, the right hand side is bounded by
    \[
        |t^2 \partial_{\theta}^2 P| + |a t^2 \partial_t \partial_{\theta} P| \leq \updelta t^{2 - 2(1 - \upgamma')} + \updelta t^{1 - (1 - \upgamma')} \leq 2 \updelta t^{\upgamma'},
    \]
    and the bound \eqref{eq:P_ode_err} follows upon redefining $\updelta$ appropriately.

    For \eqref{eq:Q_ode_err}, one simply applies the product rule on $\partial_t (e^P t \partial_{\theta} Q)$ to get:
    \begin{equation} \label{eq:Q_aux}
        (t \partial_t + a t \partial_{\theta}) (e^P t \partial_{\theta} Q) - (1 + t \partial_t P)(e^P t \partial_{\theta} Q ) = e^P t^2 \partial_t \partial_{\theta} Q + a t \partial_{\theta} (e^P t \partial_{\theta} Q).
    \end{equation}
    Applying Lemma~\ref{lem:low_order_linfty} repeatedly, one bounds the right hand side of this by $\updelta(2 + C_*) t^{\upgamma'}$, so \eqref{eq:Q_ode_err} follows upon further redefining $\updelta$.

    For \eqref{eq:R_ode_err}, one uses the equation \eqref{eq:Q_evol}, to derive
    \begin{equation} \label{eq:R_aux}
        (t \partial_t + a t \partial_{\theta}) (e^P t \partial_t Q) + (t \partial_t P)(e^P t \partial_t Q) = e^P t^2 \partial_{\theta}^2 Q + 2 e^P  t \partial_{\theta} Q \cdot t \partial_{\theta} P + a t \partial_{\theta} (e^P t \partial_t Q),
    \end{equation}
    and by Lemma~\ref{lem:low_order_linfty} is bounded by $ \updelta (2 + 3 C_*) t^{\upgamma'}$, thus completing the proof upon redefining $\updelta$.
\end{proof}

\subsection{The equations of variation}

On top of the ODEs for $-t \partial_t P$, $e^P t \partial_{\theta} Q$ and $e^P t \partial_t Q$ exhibited in Corollary~\ref{cor:ode}, we shall also require corresponding ODEs for their $\partial_{\theta}$-derivatives; that is, the linear system obtained by linearizing the ODEs of Corollary~\ref{cor:ode} around a given solution for $(- t \partial_t P, e^P t \partial_{\theta} Q, e^P t \partial_t Q)$.


\begin{corollary} \label{cor:ode2}
    Let $(P, Q)$ be as in Proposition~\ref{prop:energy_hierarchy}. Then for $N$ chosen sufficiently large there exists a family of constants $\updelta = \updelta(C_*, \upzeta, t_*) > 0$, with $\updelta \downarrow 0$ as $r_* \downarrow 0$, such that for all $(t, x) \in [t_b, r_0] \times \mathbb{S}^1$ and all $a \in [-1, 1]$:
    \begin{gather}
        \left| (t \partial_t + a t \partial_{\theta}) (- t \partial_t \partial_{\theta} P) - 2 (e^P t \partial_{\theta} Q) \partial_{\theta} (e^P t \partial_{\theta} Q) + 2 (e^P t \partial_t Q) \partial_{\theta} (e^P t \partial_t Q) \right| \, (t, x) \leq \updelta t^{- 1 + 2\upgamma'}, \label{eq:L_ode_err} \\[0.6em]
        \left| (t \partial_t + a t \partial_{\theta}) (\partial_{\theta}(e^P t \partial_{\theta} Q)) - t \partial_t \partial_{\theta} P \, e^P t \partial_{\theta} Q - (1 + t \partial_t P) \partial_{\theta} (e^P t \partial_{\theta} Q) \right| \, (t, x) \leq \updelta t^{- 1 + 2 \upgamma'}, \label{eq:M_ode_err} \\[0.6em] 
        \left| (t \partial_t + a t \partial_{\theta}) (\partial_{\theta} (e^P t \partial_t Q)) + t \partial_t \partial_{\theta} P \, e^P t \partial_t Q + t \partial_t P \, \partial_{\theta} (e^P t \partial_t Q) \right| \, (t, x) \leq \updelta t^{-1 + 2\upgamma'}. \label{eq:N_ode_err}
    \end{gather}
\end{corollary}

\begin{proof}
    Each of \eqref{eq:L_ode_err}--\eqref{eq:N_ode_err} will be derived by commuting the equations \eqref{eq:P_aux}--\eqref{eq:R_aux} used in the proof of Corollary~\ref{cor:ode} with a $\partial_{\theta}$-derivative, then using Lemma~\ref{lem:low_order_linfty} to bound the right hand side. We demonstrate this by deriving \eqref{eq:N_ode_err}; by commuting \eqref{eq:R_aux} with $\partial_{\theta}$, one yields
    \begin{multline*}
        (t \partial_t + a t \partial_{\theta}) (\partial_{\theta} (e^P t \partial_t Q)) + t \partial_t \partial_{\theta} P \, e^P t \partial_t Q + t \partial_t P \, \partial_{\theta} (e^P t \partial_t Q) \\[0.2em]
        = 3 t \partial_{\theta} P \, t \partial_{\theta}^2 Q + t \, e^P t \partial_{\theta}^3 Q + t \partial_{\theta}^2 P \, e^P t \partial_{\theta} Q + a t \partial_{\theta}^2 (e^P t \partial_t Q).
    \end{multline*}
    By Lemma~\ref{lem:low_order_linfty}, each $\partial_{\theta}$-derivative on the right hand side, not including the one derivative in $e^P t \partial_{\theta} Q$, costs $t^{- 1 + \upgamma'}$, thus the additional power of $t$ on the right hand side means 
    \begin{equation*}
        \left | (t \partial_t + a t \partial_{\theta}) (\partial_{\theta} (e^P t \partial_t Q)) + t \partial_t \partial_{\theta} P \, e^P t \partial_t Q + t \partial_t P \, \partial_{\theta} (e^P t \partial_t Q) \right |
        \lesssim_{C_*} \updelta t^{1 - 2 (1 - \upgamma')} = \updelta t^{-1 + 2 \upgamma'}.
    \end{equation*}
    Redefining $\updelta$ to absorb the implied constant, one deduces \eqref{eq:N_ode_err}.
\end{proof}

%% file: ode.tex
\section{Low order ODE analysis} \label{sec:ode}

In this section, we apply the ODEs derived in Section~\ref{sec:interp} to derive $L^{\infty}$ bounds for $0$th and $1$st order quantities. Eventually, these will be used to improve the bootstrap assumptions \eqref{eq:bootstrap_pq}--\eqref{eq:bootstrap_dth2}.

\subsection{The bounce ODE}

\begin{lemma} \label{lem:ode_bounce}
    For $t_0 \leq t_* < 1$ sufficiently small, let $\mathscr{P}, \mathscr{Q}, \mathscr{R}: [t_b, t_0] \subset \R_{>0}\to \R$ satisfy the following ODEs, where for some $0 < \upgamma' < \upgamma$ the error terms $\mathscr{E}_i$ obey $|\mathscr{E}_i| \leq \updelta t^{\upgamma'}$ for $i = 1, 2, 3$:
    \begin{gather}
        \label{eq:p_eq}
        t \partial_t \mathscr{P} = \mathscr{Q}^2 - \mathscr{R}^2 + \mathscr{E}_1, \\[0.5em]
        \label{eq:q_eq}
        t \partial_t \mathscr{Q} = (1 - \mathscr{P}) \mathscr{Q} + \mathscr{E}_2, \\[0.5em]
        \label{eq:r_eq}
        t \partial_t \mathscr{R} = \mathscr{P} \mathscr{R} + \mathscr{E}_3.
    \end{gather}
    Suppose furthermore that for some $\upgamma'' > \upgamma$, one has $(\mathscr{P} - 1)^2(t_0) + \mathscr{Q}^2(t_0) \leq (1 - \upgamma'')^2$ and $|\mathscr{R}(t_0)| \leq \updelta t_0^{\upgamma'}$. Then for $\updelta$ chosen sufficiently small (depending on all of $\upgamma, \upgamma', \upgamma''$) the solution obeys the following bounds for $t \in [t_b, t_0]$.
    \begin{equation} \label{eq:pq_bounds}
        (\mathscr{P} - 1)^2(t) + \mathscr{Q}^2(t) \leq (1 - \upgamma)^2, \qquad |\mathscr{R}(t)| \leq t^{\upgamma'}.
    \end{equation}
\end{lemma}

\begin{proof}
    We proceed via using a continuity / bootstrap argument, with bootstrap assumption which is exactly \eqref{eq:pq_bounds}. That is, we assume \eqref{eq:pq_bounds} holds on an interval $[\tilde{t}, t_0] \subset [t_b, t_0]$, and show that we may in turn improve upon \eqref{eq:pq_bounds} in this interval.
    For the improvement step, we use an approximate monotonicity property of the ODE system; let $\mathscr{K}$ be defined by
    \begin{equation} \label{eq:ode_k}
        \mathscr{K} \coloneqq (\mathscr{P} - 1)^2 + \mathscr{Q}^2 + \mathscr{R}^2.
    \end{equation}
    Using \eqref{eq:p_eq}--\eqref{eq:r_eq}, one may show that $t \partial_t \mathscr{K} = \mathscr{R}^2 + 2 \mathscr{E}_1 (\mathscr{P} - 1) + 2 \mathscr{E}_2 \mathscr{Q} + 2 \mathscr{E}_3 \mathscr{R}$. Thus using $\mathscr{R}^2 \geq 0$ and taking \eqref{eq:pq_bounds} as a bootstrap assumption, $t \partial_t \mathscr{K}$ may be bounded by
    \[
        t \partial_t \mathscr{K} \geq - \left( 4(1- \upgamma) + 2 t^{\upgamma'} \right) \updelta t^{\upgamma'}.
    \]

    Thus for $\updelta = \updelta(\upgamma, \upgamma', \upgamma'')$ chosen sufficiently small, one may guarantee that 
    \begin{equation} \label{eq:k_int}
        \int^{t_0}_{t} \partial_t \mathscr{K}(\tilde{t}) d \tilde{t} \leq \frac{1}{2} \left( ( 1 - \tfrac{\upgamma + \upgamma''}{2} )^2 - (1 - \upgamma'')^2 \right).
    \end{equation}
    Similarly using our assumptions on initial data we may choose $\updelta$ sufficiently small such that 
    \begin{equation} \label{eq:k_init}
        \mathscr{K}(t_0) \leq (1 - \upgamma'')^2 + \frac{1}{2} \left( (1 - \tfrac{\upgamma + \upgamma''}{2} )^2 + (1 - \upgamma'')^2 \right).
    \end{equation}
    Combining \eqref{eq:k_int} and \eqref{eq:k_init}, it is clear that one has $\mathcal{K}(t) \leq (1 - \tfrac{\upgamma + \upgamma''}{2})^2 < (1 - \upgamma)^2$ for $t \in [\tilde{t}, t_0]$, thereby improving the first inequality in \eqref{eq:pq_bounds}.

    For the second inequality in \eqref{eq:pq_bounds}, from \eqref{eq:r_eq} and $\mathscr{P} \geq \upgamma$ (which follows from the first part of \eqref{eq:pq_bounds}), one has the following differential inequality for $\mathscr{R}^2$:
    \[
        t \partial_t ( \mathscr{R}^2 ) \geq 2 \upgamma \mathscr{R}^2 + 2 \mathscr{E}_3 \mathscr{R}.
    \]
    Using an integrating factor and the assumptions on $\mathscr{R}$ and $\mathscr{E}_3$, we may thus write
    \[
        t \partial_t ( t^{- 2 \upgamma} \mathscr{R}^2 ) \geq 2 \updelta t^{- 2 (\upgamma - \upgamma')}.
    \]
    Our assumption at $t = t_0$ means $t_0^{- 2 \upgamma} \mathscr{R}^2(t_0) \geq \updelta^2 t_0^{-2 (\upgamma - \upgamma')}$. Thus for $\updelta$ chosen sufficiently small, we integrate the above inequality in the direction of decreasing $t$ and yield that $t^{-2 \upgamma} \mathscr{R}^2 < t^{-2(\upgamma - \upgamma')}$. Multiplying both sides by $t^{2\upgamma}$, this is a strict improvement of the second inequality in \eqref{eq:pq_bounds}. This completes our continuity argument, completing the proof of Lemma~\ref{lem:ode_bounce}.
\end{proof}

\subsection{The equations of variation}

\begin{lemma} \label{lem:ode_variation}
    For $0 < t_0 \leq t_*  < 1$, let $\mathscr{P}, \mathscr{Q}, \mathscr{R}: [t_b, t_0] \to \R$ satisfy the assumptions of Lemma~\ref{lem:ode_bounce}. Further, let $\mathscr{L}, \mathscr{M}, \mathscr{N}: [t_b, t_0] \to \R$ obey the following ODEs, where $|\mathscr{E}_i| \leq \updelta t^{-1 + 2 \upgamma'}$ for $i=4, 5, 6$:
    \begin{gather}
        t \partial_t \mathscr{L} = 2 \mathscr{Q} \mathscr{M} - 2 \mathscr{R} \mathscr{N} + \mathscr{E}_4, \label{eq:l_eq} \\[0.5em]
        t \partial_t \mathscr{M} = - \mathscr{Q} \mathscr{L} + (1 - \mathscr{P}) \mathscr{M} + \mathscr{E}_5, \label{eq:m_eq} \\[0.5em]
        t \partial_t \mathscr{N} = \mathscr{R} \mathscr{L} + \mathscr{P} \mathscr{N} + \mathscr{E}_6. \label{eq:n_eq}
    \end{gather}
    For some $\upzeta > 0$, impose the following conditions at initial data: $|\mathscr{L}(t_0)| + |\mathscr{M}(t_0)| + |\mathscr{N}(t_0)| \leq \upzeta t_0^{- 1 + \upgamma}$. We further assume that $2 \upgamma' > \upgamma$.

    Then for $\updelta$ chosen sufficiently small depending on $\upgamma, \upgamma', \upgamma''$ and $\upzeta$, there exists a constant $D > 0$ depending on the same parameters such that for all $t \in [t_b, t_0]$, one has
    \begin{equation} \label{eq:mn_bounds}
        |\mathscr{L} (t)| + |\mathscr{M} (t)| + |\mathscr{N} (t)| \leq D \, t^{- 1 + \upgamma}.
    \end{equation}
\end{lemma}

\begin{proof}
    We shall rewrite the system \eqref{eq:l_eq}--\eqref{eq:n_eq} in the following matrix form:
    \begin{equation} \label{eq:mn_matrix}
        t \partial_t \begin{bmatrix} \mathscr{L} \\ \mathscr{M} \\ \mathscr{N} \end{bmatrix} =
        \underbrace{
            \begin{bmatrix} 0 & 2 \mathscr{Q} & - 2 \mathscr{R} \\ - \mathscr{Q} & 1 - \mathscr{P} & 0 \\ \mathscr{R} & 0 & \mathscr{P} \end{bmatrix}
        }_{\eqqcolon \mathbf{L}}
        \begin{bmatrix} \mathscr{L} \\ \mathscr{M} \\ \mathscr{N} \end{bmatrix}
        +
        \begin{bmatrix} \mathscr{E}_4 \\ \mathscr{E}_5 \\ \mathscr{E}_6 \end{bmatrix}.
    \end{equation}

    The goal is to bound the operator norm of the matrix $\mathbf{L} = \mathbf{L}(r)$, as a function of $t \in [t_b, t_0]$, in such a way that allows one to integrate this equation. Note our operator norm will be with respect to the $\ell^{2}$-norm on $\R^3$, i.e.~for a $3 \times 3$ matrix $\mathbf{M}$ we write
    \[
        \| \mathbf{M} \|_{op} \coloneqq \sup_{ \mathrm{x} \in \R^3 \setminus \{0\}} \frac{ \| \mathbf{M} \mathrm{x} \|_{\ell^{2}}} { \| \mathrm{x} \|_{\ell^{2}}}.
    \]
    In fact, we shall actually estimate the operator norm of
    \[
        \tilde{\mathbf{L}} = 
        \begin{bmatrix} 0 & 2 \mathscr{Q} & - 2 \mathscr{R} \\ - \mathscr{Q} & 1 - \mathscr{P} & 0 \\ \mathscr{R} & 0 & 0 \end{bmatrix}
        = \mathbf{L} + \mathbf{P},
    \]
    where we note $\mathbf{P} = \mbox{diag}(0, 0, \mathscr{P})$ is a positive definite matrix with respect to the $\ell^2$ inner product.
    Our strategy will be to estimate $\| \tilde{\mathbf{L}} \|_{op}$ differently depending on the size of $\mathscr{Q}(t)$. The idea is that when $\mathscr{Q}(t)$ is small, the largest matrix element of $\tilde{\mathbf{L}}$ will be $1 - \mathscr{P}$, which is bounded by $1 - \upgamma$ by Lemma~\ref{lem:ode_bounce}. 

    On the other hand, when $\mathscr{Q}(t)$ is not small, we will only have a weaker quantitative bound $\| \tilde{\mathbf{L}} \|_{op} \leq 10$; this will be mitigated using an estimate on the size of the interval for which $\mathscr{Q}$ small is not small.

    \medbreak \noindent
    \underline{Case 1: $|\mathscr{Q}(t)| \leq \frac{1}{4} (\upgamma'' - \upgamma)$:} 
    \medbreak \noindent
    Recall from the proof of Lemma~\ref{lem:ode_bounce}, particularly \eqref{eq:k_int} and \eqref{eq:k_init}, that for $t \in [t_b, t_0]$ one has
    \[
        (\mathscr{P} - 1)^2 + \mathscr{Q}^2 \leq \left( 1 - \frac{\upgamma + \upgamma''}{2} \right)^2.
    \]
    In particular, the Hilbert--Schmidt\footnote{Recall the Hilbert--Schmidt norm of a $3 \times 3$  matrix $\mathbf{A}$ is given by the $\ell^2$-norm of all its matrix elements: 
    \[ \| \mathbf{A} \|_{HS} = \left( \sum_{i, j = 1}^3 \mathbf{A}_{ij}^2 \right)^{1/2}. \]}
    norm of $\tilde{\mathbf{L}}$ is given by:
    \[
        \| \tilde{\mathbf{L}} \|_{HS}^2 = (\mathscr{P}-1)^2 + \mathscr{Q}^2 + 4 \mathscr{Q}^2 + 5 \mathscr{R}^2 \leq \left( 1 - \frac{\upgamma + \upgamma''}{2} \right)^2 + \frac{1}{4} (\upgamma'' - \upgamma)^2 + 5 \mathscr{R}^2.
    \]
    
    Without the $5 \mathscr{R}^2$ term, the right hand side is strictly less than $(1 - \upgamma)^2$. From Lemma~\ref{lem:ode_bounce}, we have $\mathscr{R}^2 \leq t^{2 \upgamma'}$, and in particular for $t < t_*$ chosen small enough one has $\| \tilde{\mathbf{L}} \|_{HS} \leq 1 - \upgamma$. Since for any $3 \times 3$ matrix $\mathbf{M}$ one has $\| \mathbf{M} \|_{op} \leq \| \mathbf{N} \|_{HS}$ one concludes that:
    \begin{equation} \label{eq:L_op_small}
        \| \tilde{\mathbf{L}}(t) \|_{op} \leq 1 - \upgamma \qquad \text{ whenever } \, |\mathscr{Q}(t)| \leq \frac{1}{4} (\upgamma'' - \upgamma).
    \end{equation}

    \medbreak \noindent
    \underline{Case 2: $|\mathscr{Q}(t)| > \frac{1}{4}(\upgamma'' - \upgamma)$:} 
    \medbreak \noindent
    If one assumes no additional smallness for $|\mathscr{Q}(t)|$ then one may only use Lemma~\ref{lem:ode_bounce} to bound the individual matrix elements of  $\tilde{\mathbf{L}}(t)$. Using Lemma~\ref{lem:ode_bounce}, we simply crudely bound each nonzero matrix element of $\tilde{\mathbf{L}}(t)$ by $4$, then proceeding via the Hilbert--Schmidt norm as above one can show $\| \tilde{\mathbf{L}} \|_{op} \leq \sqrt{80} \leq 10$.  

    As mentioned previously, this will be mitigated using control on the size of the set $B = \{ t \in [t_b, t_0] : |\mathscr{Q}(t)| > \frac{1}{4} (\upgamma'' - \upgamma) \}$, at least upon assuming sufficient smallness on $\updelta$ and $t_*$. To justify this, suppose that $\updelta < \frac{1}{64} (\upgamma'' - \upgamma)^2$ and that $\mathscr{R}^2 \leq t_*^{2 \upgamma'} \leq \frac{1}{64}(\upgamma'' - \upgamma)^2$. Then using \eqref{eq:pq_bounds} and the equation \eqref{eq:p_eq} for $t \partial_t \mathscr{P}$,
    \[
        t \partial_t \mathscr{P} \geq \frac{1}{16} (\upgamma'' - \upgamma)^2 - \mathscr{R}^2 - \updelta \geq \frac{1}{32} (\upgamma'' - \upgamma)^2
    \]
    while $t \in B$.
    Furthermore, for $t \not\in B$, from Lemma~\ref{lem:ode_bounce} one still has $t \partial_t \mathscr{P} \geq - C t^{\upgamma'}$ for some $C > 0$. But from \eqref{eq:pq_bounds}, $\mathscr{P}(r)$ is certainly bounded between $0$ and $2$. So:
    \begin{align*}
        2 \geq \int_t^{t_0} t \partial_t P (\tilde{t}) \, \frac{d\tilde{t}}{\tilde{t}}
        &= \int_{\tilde{t} \in B} t \partial_t P(\tilde{t}) \, \frac{d \tilde{t}}{\tilde{t}} + \int_{t \not\in B} t \partial_t P (\tilde{t}) \, \frac{d \tilde{t}}{\tilde{t}} \geq \mu(B) \cdot \frac{1}{32}(\upgamma'' - \upgamma)^2 - C \upgamma'^{-1} t_*^{\upgamma'}.
    \end{align*}

    Here $\mu(B)$ is the measure of the set $B$ with respect to the measure $\frac{dt}{t}$. Let us choose $t_*$ small enough so that $C \upgamma'^{-1} t_*^{\upgamma'} \leq 2$. Then collecting all this information,
    \begin{gather} \label{eq:L_op_big}
        \| \tilde{\mathbf{L}}(t) \|_{op} \leq 10, \qquad \text{ whenever } \, |\mathscr{Q}| > \frac{1}{4}(\upgamma'' - \upgamma) \\[0.5em]
        \text{where } B = \left \{ t \in [t_b, t_0] : |\mathscr{Q}(t)| > \frac{1}{4} (\upgamma'' - \upgamma) \right \} \text{ has } \mu(B) \leq 128 (\upgamma'' - \upgamma)^{-2}. \nonumber
    \end{gather}

    \medbreak 
    We now use \eqref{eq:L_op_small} and \eqref{eq:L_op_big} to complete the proof of the lemma. We introduce the further notation:
    \[
        \mathbf{x} = \begin{bmatrix} \mathscr{L} \\ \mathscr{M} \\ \mathscr{N} \end{bmatrix}, \quad
        \mathbf{e} = \begin{bmatrix} \mathscr{E}_4 \\ \mathscr{E}_5 \\ \mathscr{E}_5 \end{bmatrix}.
    \]
    So that we further rewrite the ODE system \eqref{eq:mn_matrix} as $t \partial_t \mathbf{x} = ( \tilde{\mathbf{L}} + \mathbf{P} ) \mathbf{x} + \mathbf{e}$. Using the positivity of $\mathbf{P}$, one then shows that $\| \mathbf{x} \|_{\ell^2}^2 = \mathbf{x}^T \mathbf{x}$ satisfies
    \[
        t \partial_t \| \mathbf{x} \|_{\ell^2}^2 \geq 2 \mathbf{x}^T \tilde{\mathbf{L}} \mathbf{x} + 2 \mathbf{x}^T \mathbf{e} \geq - 2 \| \tilde{\mathbf{L}} \|_{op} \, \| \mathbf{x} \|_{\ell^2}^2 - 2 \sqrt{3} \updelta t^{-1 + 2 \upgamma'} \| \mathbf{x} \|_{\ell^2}.
    \]
    Note the final inequality follows from the definition of the operator norm and using Cauchy--Schwarz on $\mathbf{x}^T \mathbf{e}$, together with the assumed bounds on the error terms $\mathscr{E}_i$. Therefore
    \begin{equation} \label{eq:x_diff}
        t \partial_t \| \mathbf{x} \|_{\ell^2} \geq - \| \tilde{\mathbf{L}} \|_{op} \cdot \| \mathbf{x} \|_{\ell^2} - \sqrt{3} \updelta t^{-1 + 2 \upgamma'}.
    \end{equation}

    The idea is now simply to apply Gr\"onwall's inequality to this differential inequality. Recalling that we always integrate in the direction of decreasing $t$, the crucial bound will be the following, which follows from \eqref{eq:L_op_small} and \eqref{eq:L_op_big}:
    \begin{align}
        \int^{t_1}_{t_2} \| \tilde{\mathbf{L}}(\tilde{t}) \|_{op} \, \frac{d \tilde{t}}{\tilde{t}}  
        &\leq \int_{t \in [t_2, t_1] \setminus B} \| \tilde{\mathbf{L}} \|_{op} \, \frac{d \tilde{t}}{\tilde{t}} + \int_{t \in [t_2, t_1] \cap B} \| \tilde{\mathbf{L}} \|_{op} \, \frac{d \tilde{t}}{\tilde{t}} \nonumber \\[0.4em]
        &\leq \int_{t_2}^{t_1} (1 - \upgamma) \, \frac{d \tilde{t}}{\tilde{t}} + \int_{t \in B} 10 \, \frac{d\tilde{t}} {\tilde{t}} \nonumber \\[0.4em]
        &\leq (1 - \upgamma) \log( \frac{t_1}{t_2} ) + 10 \mu(B) = (1 - \upgamma) \log( \frac{t_1}{t_2} ) + 1280 (\upgamma'' - \upgamma)^{-2}.  \label{eq:L_op_int}
    \end{align}

    To conclude, integrating \eqref{eq:x_diff} and inserting our initial data bounds yields the following integral inequality, where $\beta(t) = \sqrt{3} ( \upzeta t_0^{-1 + \upgamma} + \updelta (1 - 2 \upgamma')^{-1} t^{-1 + 2 \upgamma'})$:
    \begin{equation*}
        \| \mathbf{x}(t) \|_{\ell^2} \leq \int^{t_0}_t \| \tilde{\mathbf{L}}(\tilde{t}) \|_{op} \cdot \| \mathbf{x} (\tilde{t}) \|_{\ell^2} \, \frac{d\tilde{t}}{\tilde{t}} + \beta(t).
    \end{equation*}
    Thus Gr\"onwall's inequality in integral form implies that
    \[
        \| \mathbf{x}(t) \|_{\ell^2} \leq \beta(t) + \int^{t_0}_t \beta(\tilde{t}) \cdot \| \tilde{\mathbf{L}}(\tilde{t}) \|_{op} \cdot \exp( \int^{\tilde{t}}_{t} \| \tilde{\mathbf{L}}(\tilde{\tilde{t}}) \| \, \frac{d \tilde{\tilde{t}}}{\tilde{\tilde{t}}} ) \, \frac{d \tilde{t}}{\tilde{t}}.
    \]

    By \eqref{eq:L_op_int} and the fact $\| \tilde{\mathbf{L}}(t) \|_{op} \leq 10$ holds everywhere, we can bound the integrand here by:
    \[
        \beta(\tilde{t}) \cdot \| \tilde{\mathbf{L}}(\tilde{t}) \|_{op} \cdot \exp( \int^{\tilde{t}}_{t} \| \tilde{\mathbf{L}}(\tilde{\tilde{t}}) \| \, \frac{d \tilde{\tilde{t}}}{\tilde{\tilde{t}}} ) \lesssim \beta(\tilde{t}) \cdot \left( \frac{\tilde{t}}{t} \right)^{1 - \upgamma} \lesssim \upzeta \left( \frac{\tilde{t}}{t_0 t} \right)^{1 - \upgamma} + \updelta \tilde{t}^{2 \upgamma' - \upgamma} t^{-1 + \upgamma}.
    \]
    Therefore, using that $2 \upgamma' > \upgamma$, inserting this into the above yields that $\| \mathbf{x}(t) \|_{\ell^2} \lesssim t^{-1 + \upgamma}$. By the definition of $\mathbf{x}$, and since the $\ell^2$ and $\ell^{\infty}$ norms on $\mathbb{R}^3$ are uniformly equivalent, the lemma follows.
\end{proof}

%% file: stability.tex
\section{The stability result}

\subsection{Proof of Theorem~\ref{thm:global}}

The proof of our stability result Theorem~\ref{thm:global} follows from a boostrap argument. By local existence for the Gowdy wave map system, for initial data as given there exists some $t_b \in (0, t_0)$ such that a solution to the Gowdy system \eqref{eq:P_evol}--\eqref{eq:Q_evol} exists in the interval $t \in [t_b, t_0]$, which moreover satisfies the bootstrap assumptions \eqref{eq:bootstrap_pq}--\eqref{eq:bootstrap_dth2}.

We show that assuming this, one may then improve upon the bootstrap assumptions, for instance showing the same \eqref{eq:bootstrap_pq}--\eqref{eq:bootstrap_dth2} hold with $C_*$ replaced by say $C_*/2$. By a standard continuity argument, $t_b$ may then be any number in the interval $(0, t_0)$, and the corresponding solution thus obeys (improved) bootstrap assumptions in the whole of $t \in (0, t_0)$. We then apply the results of Section~\ref{sec:l2} and Section~\ref{sec:ode} to conclude.

\subsubsection{Improving the bootstrap assumptions}

Let $(P, Q)$ be a solution of the Gowdy symmetric system in our bootstrap region $t \in [t_b, t_0]$ with assumptions on initial data as given. Then the results of Section~\ref{sec:l2} and Section~\ref{sec:interp} all apply, in particular Proposition~\ref{prop:energy_hierarchy}, Lemma~\ref{lem:low_order_linfty}, Corollary~\ref{cor:ode} and Corollary~\ref{cor:ode2}. We now proceed in the following steps:

\medbreak \noindent
\underline{Step 1: ODE analysis on timelike curves} 
\medbreak \noindent
The first step will be to use Corollary~\ref{cor:ode}, Corollary~\ref{cor:ode2} and the results of Section~\ref{sec:ode} to provide $L^{\infty}$ bounds for the following $6$ key quantities:
\[
    - t \partial_t P, \quad e^P t \partial_{\theta} Q, \quad e^P t \partial_t Q, \quad - t \partial_t \partial_{\theta} P, \quad e^P t \partial_{\theta}^2 Q, \quad e^P t \partial_t \partial_{\theta} Q.
\]

To do set, let $\gamma: [t_b, t_0] \to (0, +\infty) \times \mathbb{S}^1$ with $\gamma(t) = (t, \theta(t))$ be a $C^1$ \emph{past-directed timelike curve}. One example is a curve of constant $\theta$ i.e.~with $\theta'(t) = 0$. In view of the metric \eqref{eq:gowdy}, the timelike character is equivalent to $|\theta'(t)| < 1$. For such a curve, $\gamma$, define
\begin{gather}
    \mathscr{P}(t) \coloneqq - t \partial_t P(\gamma(t)), \quad \mathscr{Q}(t) \coloneqq e^P t \partial_{\theta} Q(\gamma(t)), \quad \mathscr{R}(t) \coloneqq e^P t \partial_t Q(\gamma(t)),
    \label{eq:PQR_def} \\[0.6em]
    \mathscr{L}(t) \coloneqq - t \partial_t \partial_{\theta} P(\gamma(t)), \quad \mathscr{M}(t) \coloneqq \partial_{\theta}(e^P \partial_{\theta} Q) (\gamma(t)), \quad \mathscr{N}(t) \coloneqq \partial_{\theta} ( e^P t \partial_t Q) (\gamma(t)).
    \label{eq:LMN_def}
\end{gather}

The initial data assumption \eqref{eq:data_linfty} means that $(\mathscr{P}(t_0) - 1)^2 + \mathscr{Q}^2(t_0) \leq ( 1 - \upgamma '')^2 $ and the second assumption \eqref{eq:data_linfty2} means that for $\updelta \leq \upzeta t_*^{\upgamma - \upgamma'}$, one also has $\mathscr{R}(t_0) \leq \updelta t_0^{\upgamma'}$. Combining these initial data bounds with Corollary~\ref{cor:ode}, it is clear that along the curve $\gamma$, $\mathscr{P}$, $\mathscr{Q}$ and $\mathscr{R}$ satisfy the assumptions of Lemma~\ref{lem:ode_bounce}, with $a = \frac{d\theta}{dt}$. Applying this lemma yields that:
\begin{equation} \label{eq:pqr_bounds}
    (1 - \mathscr{P}(t))^2 + \mathscr{Q}^2(t) \leq (1 - \upgamma)^2, \quad \mathscr{R}(t) \leq t^{\upgamma'}.
\end{equation}

We also provide bounds for $\mathscr{L}$, $\mathscr{M}$ and $\mathscr{N}$. To do this, we use the ODEs in Corollary~\ref{cor:ode2}. With error terms bounded as in this corollary, it is evident that the quantities $\mathscr{L}$, $\mathscr{M}$ and $\mathscr{N}$ satisfy the system \eqref{eq:l_eq}--\eqref{eq:m_eq} in Lemma~\ref{lem:ode_variation}. 

We also need to verify the conditions on initial data for the ODEs. To do so, we use the initial $L^2$ bounds \eqref{eq:data_energy_boundedness_P} and \eqref{eq:data_energy_boundedness_Q}. It follows from these, Sobolev embedding in $\mathbb{S}^1$, and also \eqref{eq:data_linfty}--\eqref{eq:data_linfty2}, that one has $|\partial_{\theta} \dot{P}_D(\theta)|, |\partial_{\theta} (e^{P_D} t_0 \partial_{\theta} Q_D )(\theta)|, |\partial_{\theta} (e^{P_D} \dot{Q}_D)(\theta)| \leq C_{\mathbb{S}^1} \upzeta$, where $C_{\mathbb{S}^1}$ is a Sobolev constant that is independent of $\upzeta$ and $\upgamma$. 
Thus one has $|\mathscr{L}(t_0)|, |\mathscr{M}(t_0)|, |\mathscr{N}(t_0)| \leq C_{\mathbb{S}^1} \upzeta \leq \upzeta t_0^{-1 + \upgamma}$, so long as $t_0 \leq t_*$ is chosen small enough to absorb the $C_{\mathbb{S}^1}$. Thus applying Lemma~\ref{lem:ode_variation}, we deduce that for some constant $D > 0$ depending on $\upgamma$, $\upzeta$ etc but not $C_*$, one bounds
\begin{equation} \label{eq:lmn_bounds}
    |\mathscr{L}(t)| + |\mathscr{M}(t)| + |\mathscr{N}(t)| \leq D t^{-1 + \upgamma}. 
\end{equation}

\medbreak \noindent
\underline{Step 2: Improving the bootstrap assumptions \eqref{eq:bootstrap_pq} and \eqref{eq:bootstrap_other}} 
\medbreak \noindent
Using \eqref{eq:pqr_bounds}, and allowing $\gamma$ to vary over all timelike curves, in particular all curves with constant $\theta$-coordinate, we find that $|t \partial_t P(t, x)| \leq 2 - \upgamma$ and $|e^P t \partial_{\theta} Q(t, x)| \leq 1 - \upgamma$ for all $(t, x) \in [t_b, t_0] \times \mathbb{S}^1$. Thus choosing $C_* \geq 2(2 - \upgamma)$, one easily improves upon \eqref{eq:bootstrap_pq}. Note that since $\mathscr{P}$ cannot change sign, \eqref{eq:pqr_bounds} actually yields $- t \partial_t P(t, x) \geq \upgamma$.

Similarly using \eqref{eq:pqr_bounds}, $|e^P t \partial_t Q(t, x)| \leq t^{\upgamma'}$, so the first part of \eqref{eq:bootstrap_other} is improved for $C_* \geq 2$.
For the second part of \eqref{eq:bootstrap_other}, we simply use:
\[
    e^{-P}(t, x) \leq e^{-P}(t_0, \theta) \cdot \exp ( - \int^{t_0}_t (- t \partial_t P)(\tilde{t}, \theta) \frac{d \tilde{t}}{\tilde{t}} ) \leq e^{-P_D}(\theta) \exp( \upgamma \log( \frac{t}{t_0} )),
\]
where in the second step we used that $- t \partial_t P(t, x) \geq \upgamma$. Using the initial data assumption \eqref{eq:data_linfty2}, we thus have $e^P (t, x) \leq \upzeta t^{\upgamma}$. Since $\upgamma' < \upgamma$ and $t \leq 1$, choosing $C_* \geq 2 \upzeta$ also improves the second part of \eqref{eq:bootstrap_other}.

\medbreak \noindent
\underline{Step 3: Improving the bootstrap assumptions \eqref{eq:bootstrap_dth} and \eqref{eq:bootstrap_dth2}}
\medbreak \noindent
We shall improve these using \eqref{eq:lmn_bounds}. Allowing $\gamma$ to vary over all timelike curves, for all $(t, \theta) \in [t_b, t_0] \times \mathbb{S}^1$:
\begin{equation} \label{eq:lmn_bounds_2}
    |t \partial_t \partial_{\theta} P(t, \theta)| \leq D t^{-(1 - \upgamma)}, \quad |\partial_{\theta}(e^P t \partial_{\theta} Q)(t, \theta)| \leq D t^{- (1 - \upgamma)}, \quad |\partial_{\theta}( e^P t \partial_t Q)(t, \theta)| \leq D t^{ - (1 - \upgamma)}.
\end{equation}
We now integrate the first inequality in \eqref{eq:lmn_bounds_2}. Using \eqref{eq:data_energy_boundedness_P} and Sobolev embedding to bound $|\partial_{\theta} P(t_0, \theta)|$, we deduce that $|\partial_{\theta} P(t, \theta)| \leq D t^{-(1 - \upgamma)}$ also, where $D$ is modified appropriately.

We now expand the second and third inequalities in \eqref{eq:lmn_bounds_2}, yielding that
\[
    |e^P t \partial_{\theta}^2 Q| \leq D t^{- (1 - \upgamma)} + |\partial_{\theta} P| \cdot |e^P t \partial_{\theta} Q|, \quad |e^P t \partial_t \partial_{\theta} Q| \leq D t^{- (1 - \upgamma)} + | \partial_{\theta} P | \cdot |e^P t \partial_t Q|.
\]
Now using the above bound for $|\partial_{\theta} P|$ and the bounds from Step 2, by modifying $D$ appropriately we also have $|e^P t \partial_{\theta}^2 Q|, |e^P t \partial_t \partial_{\theta} Q| \leq D t^{- (1 - \upgamma)}$. Choosing $C_* \geq 2D$, we thereby improve the remaining bootstrap assumptions.

\subsubsection{Completion of the proof}

To complete the proof, it remains to verify the bounds \eqref{eq:global_linfty} and \eqref{eq:global_energy}, and finally prove the statements contained in the final paragraph of Theorem~\ref{thm:global}. The former two bounds follow immediately from \eqref{eq:pqr_bounds} and Proposition~\ref{prop:energy_hierarchy} respectively. 

To show that $-t \partial_t P(t, \theta)$ converges pointwise as $t \to 0$, note that this is equivalent to showing that $\mathscr{P}(t)$, as defined in \eqref{eq:PQR_def}, has a limit as $t \to 0$, for curves $\gamma$ of constant $\theta$-coordinate. From Step 1, we know that $\mathscr{P}$ solves an ODE as in Lemma~\ref{lem:ode_bounce}, namely
\[
    t \partial_t \mathscr{P} = \mathscr{Q}^2 - \mathscr{R}^2 + \mathscr{E}_1,
\]
where $|\mathscr{E}_1| \lesssim t^{\upgamma'}$. By \eqref{eq:pqr_bounds}, one also has $\mathscr{R}^2 \lesssim t^{2 \upgamma'}$. We can integrate the above and write:
\[
    \mathscr{P}(t) - \mathscr{P}(t_0) + \int^{t_0}_{t} (\mathscr{E}_1 - \mathscr{R}^2) \frac{d\tilde{t}}{\tilde{t}} = - \int^{t_0}_{t} \mathscr{Q}^2 \frac{d\tilde{t}}{\tilde{t}}.
\]

By the lower bound for $\mathscr{P}(t)$ in \eqref{eq:pqr_bounds} and the aforementioned bounds for $\mathscr{E}_1$ and $\mathscr{R}$ implying a lower bound for the LHS, it is clear that the right hand side attains a limit as $t \to 0$. Thus the quantity $\mathscr{P}(t)$ also attains a limit as $t \to 0$.

Finally, we show that $Q(t, \theta)$ and $\partial_{\theta} Q (t, \theta)$ converge \emph{uniformly} as $t \to 0$, where for the latter statement we require $\upgamma > \frac{1}{3}$. Starting with $Q(t, \theta)$, we multiply \eqref{eq:Q_evol} by $t^{-\upsigma}$, for some $\upsigma > 0$ to be determined:
\begin{equation} \label{eq:tdtq}
    t \partial_t (t^{-\upsigma} t \partial_t Q) = (- 2 t \partial_t P - \upsigma) \cdot t^{- \upsigma} t \partial_t Q + t^{2 - \upsigma} \partial_{\theta}^2 Q + 2 t^{2 - \upsigma} \partial_{\theta} P \partial_{\theta} Q.
\end{equation}
We shall choose $\upsigma < \upgamma'$. Therefore, by \eqref{eq:global_linfty}, we have $- 2 \partial_t P - \upsigma > 0$, while by Lemma~\ref{lem:low_order_linfty} and Lemma~\ref{lem:low_order_linfty2}, we can bound the final two terms on the right hand side by
\[
    |t^{2 - \upsigma} \partial_{\theta}^2 Q| + |t^{2 - \upsigma} \partial_{\theta} P \partial_{\theta} Q| \lesssim t^{2 - \upsigma - 2(1 - \upgamma')} = t^{2 \upgamma' - \upsigma}.
\]

Since $2 \upgamma' - \upsigma \geq \upgamma' > 0$, this is therefore integrable with respect to $\frac{dt}{t}$ as $t \to 0$. Therefore integrating \eqref{eq:tdtq} yields that $t^{- \upsigma} t \partial_t Q$ is uniformly bounded in the region $0 < t \leq t_0$. This means that:
\[
    Q(t, \theta) = Q(t_0, \theta) - \int_t^{t_0} (t^{-\upsigma} t \partial_t Q(\tilde{t}, \theta)) \tilde{t}^{-1 + \upsigma} \, d \tilde{t}.
\]
indeed converges uniformly as $t \to 0$ to some $q(\theta)$, and $|Q(t, \theta) - q(\theta)| \lesssim t^{\upsigma}$.

To prove a similar statement for $\partial_{\theta} Q$, we need to differentiate \eqref{eq:tdtq} in $\theta$,
\begin{multline} \label{eq:tdtqt}
    t \partial_t (t^{-\upsigma} t \partial_t \partial_{\theta} Q) = 
    (- 2 t \partial_t P - \upsigma) \cdot t^{- \upsigma} t \partial_t \partial_{\theta} Q - 2 t \partial_t \partial_{\theta} P \cdot t^{-\upsigma} t \partial_t Q
    \\[0.2pt]
    + t^{2 - \upsigma} \partial_{\theta}^3 Q + 2 t^{2 - \upsigma} \partial_{\theta} P \partial_{\theta}^2 Q + 2 t^{2 - \upsigma} \partial_{\theta}^2 P \partial_{\theta} Q.
\end{multline}
To bound the last term on the first line, we use Step 2. In other words, we use the estimates \eqref{eq:bootstrap_other}, which together yield $t \partial_t Q \lesssim t^{2 \upgamma'}$. Therefore, using also Lemma~\ref{lem:low_order_linfty} we get that:
\[
    |t \partial_t \partial_{\theta} P \cdot t^{- \upsigma} t \partial_t Q| \lesssim t^{-(1 - \upgamma')} \cdot t^{2 \upgamma' - \upsigma} = t^{3 \upgamma' - 1 - \upsigma}.
\]
For the second line of \eqref{eq:tdtqt}, we again combine Lemma~\ref{lem:low_order_linfty} and Lemma~\ref{lem:low_order_linfty2}, to get
\[
    |t^{2 - \upsigma} \partial_{\theta}^3 Q| + |t^{2 - \upsigma} \partial_{\theta} P \partial_{\theta}^2 Q| + |t^{2 - \upsigma} \partial_{\theta}^2 P \partial_{\theta} Q| \lesssim t^{2 - \upsigma - 3(1- \upgamma')} = t^{3\upgamma' - 1 - \upsigma}.
\]

For $\upgamma > \frac{1}{3}$, we may choose $\upgamma' > \frac{1}{3}$ also and then choose $\upsigma < 3\upgamma' - 1$. For this choice, $- 2 t \partial_t P - \upsigma > 2 \upgamma' - (3 \upgamma' - 1) = 1 - \upgamma' \geq 0$, so the first term on the right hand side of \eqref{eq:tdtqt} can be ignored, and by integrability of $t^{3 \upgamma' - 1 - \upsigma}$ with respect to $\frac{dt}{t}$ one yields that $t^{-\upsigma} t \partial_t \partial_{\theta} Q$ is uniformly bounded in the region $0 < t \leq t_0$. One now concludes as in the previous case, and one yields:
\begin{equation} \label{eq:q_theta_limit}
    |\partial_{\theta} Q(t, \theta) - \partial_{\theta} q (\theta) | \lesssim_{\upsigma} t^{\upsigma}, \qquad \text{ for any } \upsigma < 3 \upgamma' - 1.
\end{equation}

%% file: bounce.tex

\section{BKL bounces}

\subsection{Proof of Theorem~\ref{thm:bounce}}

Under the assumptions of Theorem~\ref{thm:bounce}, we may apply the stability result Theorem~\ref{thm:global}; in particular Corollary~\ref{cor:ode} applies, and the ODEs \eqref{eq:bounce_ode} follow exactly as in the proof of Theorem~\ref{thm:bounce} -- note that the term involving $\mathscr{R}$ can now be ignored since we have $\mathscr{R}(t) \lesssim t^{\upgamma'}$ from \eqref{eq:bootstrap_other}.

It remains to prove the various convergence results stated in Theorem~\ref{thm:bounce}. The convergence of $\mathscr{P}_{\gamma}(t)$ to some $\mathscr{P}_{\gamma, \infty} = V(\theta_0)$ follows in the same manner to the pointwise convergence of $-t \partial_t P(t, \theta)$ in the proof of Theorem~\ref{thm:global}. To show the convergence of $\mathscr{Q}_{\gamma}(t)$ to $0$, we first use that from the first ODE in \eqref{eq:bounce_ode} and the convergence of $\mathscr{P}_{\gamma}$, the integral
\begin{equation} \label{eq:Q_int}
    \int_0^{t_0} \mathscr{Q}_{\upgamma}^2(\tilde{t}) \, \frac{d \tilde{t}}{\tilde{t}}
\end{equation}
is finite, or in other words $\mathscr{Q}_{\upgamma}^2(t)$ converges to $0$ in an averaged sense. In particular there is a sequence $\{t_k\}$ with $t_k \to 0$ such that $\mathscr{Q}_{\gamma}^2(t_k) \to 0$ as $k \to \infty$. To upgrade this sequential convergence to convergence of $\mathscr{Q}_{\gamma}(t)$, we use the second equation in \eqref{eq:bounce_ode}, or more precisely the ODE $t \partial_t \mathscr{Q}_{\upgamma}^2 = 2 \mathscr{Q}_{\gamma}^2 (1 - \mathscr{P}_{\gamma}) + 2 \mathscr{E}_{\mathscr{Q}} \mathscr{Q}_{\upgamma}$.

Integrating this gives that for any $0 < t \leq t_k \leq t_0$, we have:
\[
    |\mathscr{Q}_{\gamma}^2(t) - \mathscr{Q}_{\gamma}^2(t_k)| \leq C \int^{t_k}_0 \mathscr{Q}_{\gamma}^2(\tilde{t}) \frac{d \tilde{t}}{\tilde{t}} + C \int^{t_k}_0 \tilde{t}^{\upgamma'} \frac{d\tilde{t}}{\tilde{t}}.
\]
By the finiteness of the integral \eqref{eq:Q_int}, the right hand side is finite, and in fact converges to $0$ as $k \to \infty$. Since we already know $\mathscr{Q}_{\gamma}^2(t_k) \to 0$, this implies that $\mathscr{Q}_{\gamma}^2(t)$, and therefore also $\mathscr{Q}_{\gamma}(t)$, converges to $0$ as $t \to 0$.

Moving onto (i), it will be necessary to use the conserved quantity $\mathscr{K}$ encountered in the proof of Lemma~\ref{lem:ode_bounce}. Recall from there that defining $\mathscr{K}_{\gamma} = (\mathscr{P}_{\gamma} - 1)^2 + \mathscr{Q}_{\gamma}^2$, one can show
\[
    \left| t \frac{d}{dt} \mathscr{K}_{\gamma} \right| \lesssim t^{\upgamma'}.
\]
Using the convergence of $\mathscr{P}_{\gamma}$ and $\mathscr{Q}_{\gamma}$, it holds that $\mathscr{K}_{\gamma}(t) \to \mathscr{K}_{\gamma, \infty} \coloneqq (\mathscr{P}_{\gamma, \infty} - 1)^2$ as $t \to 0$, and moreover from the above that $|\mathscr{K}_{\gamma}(t) - \mathscr{K}_{\gamma, \infty}| \lesssim t^{\upgamma'}$. Inserting $t = t_0$, one has:
\begin{equation} \label{eq:P_conv_aux}
    | ( \mathscr{P}_{\gamma}(t_0) - 1)^2 - (\mathscr{P}_{\gamma, \infty} - 1)^2 | \lesssim t_0^{\upgamma'} + \mathscr{Q}^2_{\gamma}(t_0).
\end{equation}

In (i), we assumed that $\partial_{\theta} Q(\gamma(t))$ did not converge to $0$ as $t \to 0$. However, since $\mathscr{Q}_{\gamma}(t) = e^P t \partial_{\theta} Q(\gamma(t))$ does converge to $0$, this implies that $e^{P(\gamma(t))} t \to 0$ as $t \to 0$. Now simply note that
\[
    t \frac{d}{dt} \log( e^{P(\gamma(t))} t) = 1 - \mathscr{P}_{\gamma}(t) \to 1 - \mathscr{P}_{\gamma,\infty} \quad \text{ as } t \to 0,
\]
and thus in order for $e^{P(\gamma(t))} t \to 0$ we must have $\mathscr{P}_{\gamma, \infty} \leq 1$. Combining this with \eqref{eq:P_conv_aux} yields the estimate \eqref{eq:P_limit}. (For instance, one uses that if $|x^2 - y^2| \leq Z$ with $x, y \geq 0$ then $|x - y| \leq Z^{1/2}$.)

For (ii), if $\upgamma > \frac{1}{2}$ and $\partial_{\theta} Q(\gamma(t))$ converges to $0$, it will be helpful to use \eqref{eq:q_theta_limit} from the proof of Theorem~\ref{thm:global}, or rather its generalization to
\begin{equation*} 
    |\partial_{\theta} Q(\gamma(t)) - \partial_{\theta} q (\theta_0) | \lesssim_{\upsigma} t^{\upsigma}, \qquad \text{ for any } \upsigma < 3 \upgamma' - 1.
\end{equation*}
Note that this generalization, where $\gamma$ is allowed to be any timelike curve rather than only a constant $\theta$-curve, is proved easily using previous methods. In (ii), we may insert $\partial_{\theta} q(\theta_0) = 0$, and therefore $|\partial_{\theta} Q(\gamma(t))| \lesssim t^{\upsigma}$. 

Furthermore, from $-t \partial_t P \leq 2 - \upgamma$ it holds that $e^P \leq t^{-2 + \upgamma}$, therefore one has $|e^P t \partial_{\theta} Q(\gamma(t))| \lesssim t^{- 1 + \upgamma + \upsigma}$. But by the definition of $\mathscr{Q}_{\gamma}(t)$, this means that $|\mathscr{Q}_{\gamma}(t)| \lesssim t^{- 1 + \upgamma + \upsigma}$. Note that so long as $\upgamma > \frac{1}{2}$, $\upgamma' < \upgamma$ and $\upsigma < 3 \upgamma' - 1$ may be chosen such that $- 1 + \upgamma + \upsigma > 0$. In the case of (ii), we therefore have from \eqref{eq:bounce_ode} that
\[
    |t \partial_t \mathscr{P}_{\gamma}(t) | \lesssim t^{2(- 1 + \upgamma + \upsigma)} + t^{\upgamma'}.
\]
For convenience, let use choose $\sigma = \upgamma' - \upgamma + \frac{1}{2}$, where $\upgamma > \upgamma' > \frac{3}{2} - \upgamma$. Then $t^{2(-1 + \upgamma + \upsigma)} = t^{2 \upgamma' - 1}$, and integrating the above immediately yields \eqref{eq:P_limit2}.
\qed

\subsection{Proof of Corollary~\ref{cor:bounce}}

Lastly, we shall apply our Theorems~\ref{thm:global} and \ref{thm:bounce} to prove the stability / instability corollary. Let $P(t, \theta)$ be as stated. By Theorem~\ref{thm:asymp_smooth}, for all $k \in \mathbb{N}$ the convergence $- t \partial_t P(t, \theta) \to V(\theta)$ holds in the $C^k$ norm.

Since we are assuming $0 < V(\theta) < 2$, since $\theta \in \mathbb{S}^1$ it holds that there exists $\upgamma \in (0, 1)$ such that $\upgamma < V(\theta) < 2 - \upgamma$. In fact, by the above convergence we can find $\upgamma, \upgamma', \upgamma''$ and $t_0 > 0$ such that for $N$ chosen (depending on $\upgamma$) as in Theorem~\ref{thm:global}, we have that for $0 < t \leq t_0$ and all $\theta \in \mathbb{S}^1$:
\[
    0 < \upgamma' < \upgamma < \upgamma'' < - t \partial_t P (t, \theta) < 2 - \upgamma'' < 2 - \upgamma < 2 - \upgamma' < 2,
\]
and there exists some $\upzeta > 0$ so that for all $0 < t \leq t_0$:
\[
    \mathcal{E}(t) = \frac{1}{2} \sum_{K = 0}^N \int_{\mathbb{S}^1} \left( (t \partial_t \partial_{\theta}^K P)^2 + t^2 (\partial_{\theta}^{K+1} P)^2 + t^{2 \upgamma} (\partial_{\theta}^K P)^2 \right) \, d\theta \leq \frac{\upzeta}{2}.
\]
(This follows because in fact $\mathcal{E}(t) \to \frac{1}{2} \sum_{K=0}^N \int (\partial_{\theta}^K V)^2 \, d\theta$ as $t \to 0$.)

We note that $t_0$ above may differ from the initial data time $t_1 > 0$ in the statement of Corollary~\ref{cor:bounce}. This is mitigated a standard \emph{Cauchy stability} argument, for any $\varepsilon_0 > 0$ there exists $\varepsilon > 0$ such that a perturbation of size $\varepsilon$ at $t = t_1$  implies a perturbation of size $\varepsilon_0$ at $t = t_0$ i.e.~$\| (\tilde{P}_D, \tilde{Q}_D, \tilde{\dot{P}}_D, \tilde{\dot{Q}}_D) - (P_D, 0, \dot{P}_D, 0) \|_{(H^{N+1})^2 \times (H^N)^2} \leq \varepsilon$ implies $\| (\tilde{P}(t_0, \theta), \tilde{Q}(t_0, \theta), t \partial_t \tilde{P}(t_0, \theta), t \partial_t \tilde{Q}(t_0, \theta) ) - (P(t_0, \theta), 0, t \partial_t P(t_0, \theta), 0)\|_{(H^{N+1})^2 \times (H^N)^2 } \leq \varepsilon_0$. Due to this we may instead consider perturbations of initial data as perturbations at time $t = t_0$.

Thereby by choosing $\varepsilon$ (and thus $\varepsilon_0$) small enough one can guarantee that
\[
    (1 + t \partial_t \tilde{P})^2 (t_0, \theta) + (e^{\tilde{P}} t_0 \partial_{\theta} \tilde{P})^2 (t_0, \theta) \leq (1 - \upgamma'')^2, \quad |e^{\tilde{P}} t \partial_t \tilde{Q}(t_0, \theta)| \leq \upzeta t_0^{\upgamma}, \quad e^{\tilde{P}}(t_0, \theta) \leq \upzeta t_0^{\upgamma}
\]
as well as the energy bound
\[
    \frac{1}{2} \sum_{K=0}^N\int_{\mathbb{S}^1} \left( (t \partial_t \partial_{\theta}^K \tilde{P})^2 + t^2 (\partial_{\theta}^{K+1} \tilde{P})^2 + t^{2 \upgamma} (\partial_{\theta}^K \tilde{P})^2 + e^{2\tilde{P}} (t \partial_t \partial_{\theta}^K \tilde{Q})^2 + e^{2 \tilde{P}} (t \partial_{\theta}^{K+1} \tilde{Q})^2 \right) d \theta \leq \upzeta.
\]
That is, the perturbed data at $t = t_0$ satisfies the assumptions \eqref{eq:data_linfty}--\eqref{eq:data_energy_boundedness_Q}. Moreover, $t_0$ can be chosen small enough so that $0 < t_0 \leq t_* = t_*(\upgamma, \upzeta)$, so that one may apply Theorems~\ref{thm:global} and \ref{thm:bounce}.

The remainder of the argument is then a direct application of these theorems. Applying Theorem~\ref{thm:global} to the perturbed data, it is clear that $-t \partial_t \tilde{P}(t, \theta)$ converges to some $\tilde{V}(\theta)$ pointwise, with $\upgamma \leq \tilde{V}(\theta) \leq 2 - \upgamma$. 

In the setting where $\frac{1}{2} < V(\theta) < \frac{3}{2}$, we are allowed to in fact choose $\upgamma > \upgamma' > \frac{1}{2}$. By the final statement of Theorem~\ref{thm:global}, $\partial_{\theta} \tilde{Q} \to \partial_{\theta} \tilde{q}(\theta)$ uniformly as claimed. We next apply Theorem~\ref{thm:bounce}. In the case that $\partial_{\theta} \tilde{q}(\theta_0) \neq 0$, Theorem~\ref{thm:bounce}(i) implies that for $\gamma$ the timelike curve with constant $\theta = \theta_0$,
\begin{align*}
    \tilde{V}(\theta_0) 
    &= \tilde{\mathscr{P}}_{\gamma, \infty} = \min \{ - t \partial_t \tilde{P}(t_0, \theta_0), 2 + t \partial_t \tilde{P}(t_0, \theta_0) \} + O(t_0^{\upgamma'/2}) + O( e^{\tilde{P}} t_0 \partial_{\theta} \tilde{Q}(t_0, \theta_0)) \\[0.3em]
    &= \min \{ - t \partial_t P(t_0, \theta_0), 2 + t \partial_t P(t_0, \theta_0)\} + O(t_0^{\upgamma'/2}) + O(\varepsilon_0) = \min \{ V(\theta_0), 2 - V(\theta_0) \} + O(t_0^{\upgamma'/2}) + O(\varepsilon_0).
\end{align*}
Since we have license to choose $t_0$ and $\varepsilon_0$ small depending on $\tilde{\varepsilon}$, Corollary~\ref{cor:bounce}(i) follows.

Similarly, in the case that $\partial_{\theta} \tilde{q}(\theta_0) = 0$, Theorem~\ref{thm:bounce}(ii)implies that for the same timelike curve $\gamma$,
\begin{align*}
    \tilde{V}(\theta_0) 
    &= \tilde{\mathscr{P}}_{\gamma, \infty} = - t \partial_t \tilde{P}(t_0, \theta_0) + O(t_0^{2 \upgamma' - 1}) \\[0.3em]
    &= - t \partial_t P(t_0, \theta_0) + O(t_0^{2 \upgamma' - 1}) + O(\varepsilon_0) = V(\theta_0) + O(t_0^{2 \upgamma' - 1}) + O(\varepsilon_0).
\end{align*}
Again choosing $t_0$ and $\varepsilon_0$ appropriately small yields Corollary~\ref{cor:bounce}(ii). \qed